\documentclass[letterpaper]{article} 
\usepackage{aaai23}  
\usepackage{times}  
\usepackage{helvet}  
\usepackage{courier}  
\usepackage[hyphens]{url}  
\usepackage{graphicx} 
\usepackage{natbib}  
\usepackage{caption} 
\usepackage{newfloat}
\usepackage{listings}
\usepackage{relsize}
\usepackage[utf8]{inputenc} 
\usepackage[T1]{fontenc}    
\usepackage{url}            
\usepackage{booktabs}       
\usepackage{amsfonts}       
\usepackage{nicefrac}       
\usepackage{microtype}      
\usepackage{xcolor}         
\usepackage{enumitem}
\usepackage{adjustbox}
\usepackage{algorithmic} 
\usepackage{algorithmic} 
\usepackage[linesnumbered,ruled,vlined]{algorithm2e}
\usepackage{algorithmic}
\SetKwInput{KwInput}{Input}                
\SetKwInput{KwOutput}{Output} 
\usepackage{amssymb}
\usepackage{amsthm}
\usepackage{xcolor}
\usepackage{subcaption}
\usepackage{booktabs}
\usepackage{multirow, tabu}
\usepackage{comment}
\usepackage{svg}
\usepackage{mathtools}

\usepackage{amsmath}
\DeclareMathOperator*{\argmax}{arg\,max}

\newtheorem{theorem}{Theorem}

\newtheorem{lemma}{Lemma}
\newtheorem{corollary}{Corollary}

\newtheorem{assumption}{Assumption}

\DeclareCaptionStyle{ruled}{labelfont=normalfont,labelsep=colon,strut=off} 
\title{Winning the CityLearn Challenge: Adaptive Optimization \\ with 
Evolutionary Search under Trajectory-based Guidance}
\author{
    Vanshaj Khattar and Ming Jin
}
\affiliations{Electrical and Computer Engineering, Virginia Tech

}

\usepackage{bibentry}

\begin{document}
\maketitle

\begin{abstract}
\begin{quote}
Modern power systems will have to face difficult challenges in the years to come: frequent blackouts in urban areas caused by high power demand peaks, grid instability exacerbated by intermittent renewable generation, and global climate change amplified by rising carbon emissions. While current practices are growingly inadequate, the path to widespread adoption of artificial intelligence (AI) methods is hindered by missing aspects of trustworthiness. The CityLearn Challenge is an exemplary opportunity for researchers from multiple disciplines to investigate the potential of AI to tackle these pressing issues in the energy domain, collectively modeled as a reinforcement learning (RL) task. Multiple real-world challenges faced by contemporary RL techniques are embodied in the problem formulation. In this paper, we present a novel method using the solution function of optimization as policies to compute actions for sequential decision-making, while notably adapting the parameters of the optimization model from online observations. Algorithmically, this is achieved by an evolutionary algorithm under a novel trajectory-based guidance scheme. Formally, the global convergence property is established. Our agent ranked first in the latest 2021 CityLearn Challenge, being able to achieve superior performance in almost all metrics while maintaining some key aspects of interpretability.  
\end{quote}
\end{abstract}

\section{Introduction}
\label{sec:Introduction}
{Rapid urbanization in the past decades has led to a substantial increase in energy use that puts stress on the grid assets, while the integration of additional renewable generation and energy storage at the distribution level introduces both opportunities and new challenges \cite{rolnick2022tackling}. The cornerstone of addressing emerging issues is the deployment of proper control and coordination strategies, which have a potential impact on enhancing energy flexibility and resilience in the face of a surge in climate-induced demand (as already seen in places like California, where rolling blackouts are increasingly frequent during the summer) \cite{vazquez2020citylearn}}.{
Current industry practice is heavily based on optimization models, such as energy dispatch and unit commitment, where parameters (e.g., technological and physical constraints) are fixed throughout the lifecycle; however, such an approach is increasingly confronted by environmental uncertainty,  renewable generation stochasticity, and the ever-increasing complexity of the distribution grid \cite{abedi2019review}. On the other hand, there has been a surge in machine learning research, notably RL, because it allows the agent to act without the need to access the true model---a feature of particular interest for large-scale, complex systems, where it is not cost-effective to develop models of such high fidelity. 
Despite recent progress, real-world RL is still in its infancy \cite{dulac2021challenges}. }

{
Against this backdrop, the CityLearn Challenge aims to spur RL solutions for the control of modern energy systems by providing a set of benchmarks for urban energy management, load shaping, and demand response in a range of climate zones \cite{vazquez2020citylearn}. The agent is tasked with exploring and exploiting the best control strategy for energy storage distributed in a community of buildings. Performance is evaluated against standard metrics such as ramping costs, peak demands, and carbon emissions. The CityLearn encapsulates 4  of the 9  real-world RL challenges identified by \cite{dulac2021challenges}, including \emph{1)} the ability to learn on live systems from limited samples---there is no training period; \emph{2)} dealing with system constraints that should never or rarely be violated---there are strict balancing equations for electricity, heating, and cooling energy; \emph{3)} the ability to provide quick action---there is a strict time limit for completing the 4-year evaluation on Google's Colab; and \emph{4)} providing system operators with explainable policies---a necessity to facilitate real-world adoption and deployment.}

\begin{figure*}[t]
  \centering
  \includegraphics[width=1.98\columnwidth]{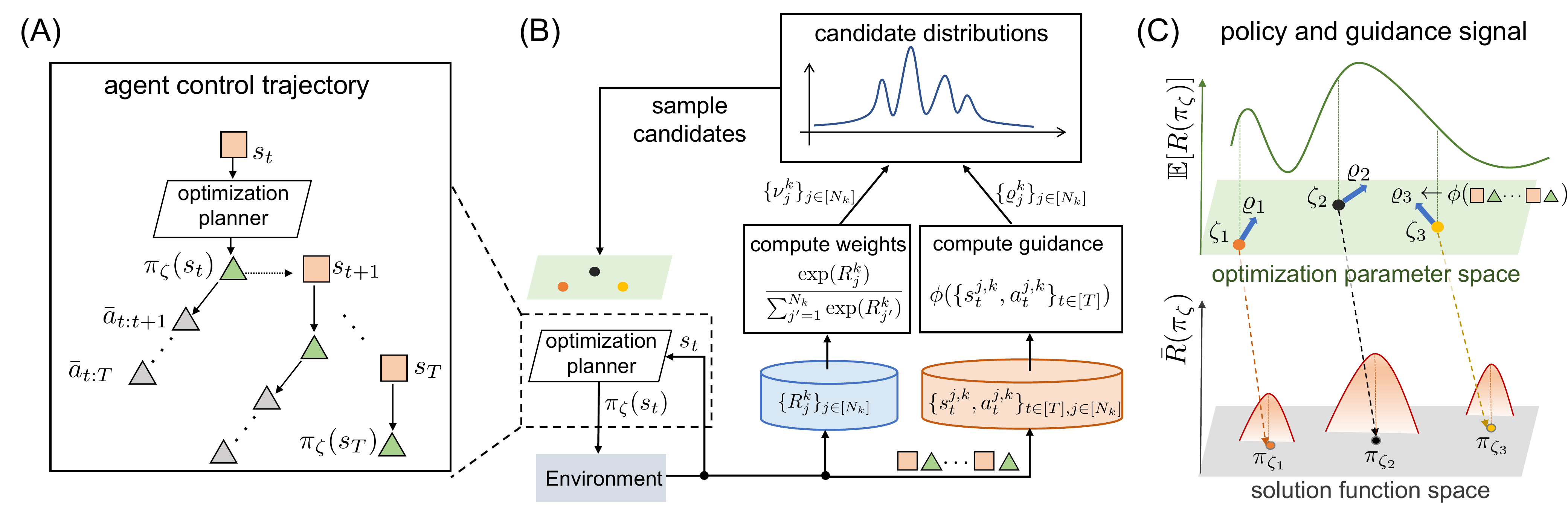}
\caption{\textbf{Overview.} (A) Agent trajectory of observed states (square) and actual (green triangle)/planned (grey triangle) actions. At each step, the agent solves optimization \eqref{eq:sol-map} to plan ahead but only executes the most immediate action; $\bar{a}_{t:t'}$ represents the action planned at time $t$ for time $t'\geq t$. (B) Illustration of the system, including interaction with the environment and adaptation of optimization parameters. Each iteration of $k$ consists of evaluating $N_k$ (three in this case) policies sampled from the candidate distribution. The observed rewards $\{R_j^k\}_{j\in[N_k]}$ and trajectories $\{s_t^{j,k},a_t^{j,k}\}_{j\in[N_k]}$ are stored in some buffers, which are then used to compute weights $\{\nu_j^k\}_{j\in[N_k]}$ and guidance signals $\{\varrho_j^k\}_{j\in[N_k]}$ to update the distribution as in \eqref{eq:iteration-measure}. (C) The two-level structure indicates the correspondence between an optimization parameter $\zeta$ and the solution function $\pi_\zeta$ (used as policy). The upper-level objective is the expected episodic rewards \eqref{eq:SeqDecMaking} and the lower-level objective is the objective function of \eqref{eq:sol-map}, the extremum of which is the policy action (see \eqref{eq:OuterLayer}). The guidance signal $\varrho$ is computed based on the trajectory of each candidate and applied on top of the original parameter during the update. }
\label{fig:overview}
\end{figure*}

In this paper, we describe our winning solution for the 2021 CityLearn Challenge based on the idea broadly categorized as \emph{adaptive optimization}. Indeed, optimization (especially convex optimization) has become the de facto standard in industrial systems with profound theoretical foundations, and various formulations for control and planning applications \cite{boyd2004convex}. Such approaches can easily encode domain-specific constraints (in the form of nonlinear functions and variational inequalities) and can gracefully handle problems with millions of decision variables \cite{facchinei2007finite}. Although well established, optimization models, once built, typically do not adapt to real-world conditions, rendering current approaches rather ``rigid.'' 

To address this fundamental limitation, we exploit that the solution of optimization lies on a manifold implicitly defined by a general equation \cite{dontchev2009implicit}. The crux of our idea is to shape this manifold by adapting the parameters of the optimization model while extracting insights from trajectory data to design guidance signals (see Fig. \ref{fig:overview}  for a detailed illustration). Key differences between our method and well-established optimization techniques (e.g., stochastic optimization \cite{powell2020reinforcement}, bi-level optimization \cite{dempe2020bilevel}) include:
\begin{enumerate}
    \item[\emph{1)}] we only allow access to the environment through interactive samples (reward, states, etc.) but not the true dynamics or reward function;
    \item[\emph{2)}] we make full use of the control trajectory of a Markov decision process (MDP) to obtain a guidance signal.
\end{enumerate}
While \emph{1)} is similar to the RL setup, \emph{2)} can be viewed as an \emph{augmentation of zeroth-order search methods with insights extracted from control data}, which, to the best of our knowledge, is the \emph{first of its kind}. In general, zeroth-order algorithms, such as simultaneous perturbation \cite{spall2005introduction,mania2018simple}, evolutionary algorithms \cite{salimans2017evolution,zhou2019evolutionary}, and Bayesian optimization \cite{snoek2012practical,frazier2018bayesian} are natural candidates for RL and easy to implement, but can potentially suffer scalability problems \cite{ghadimi2013stochastic}. Nevertheless, the parameters of an optimization model (i.e., variables to be learned) usually have clear interpretations. Thus, we design a mechanism to leverage domain knowledge for the design of appropriate guidance during the search for these parameters; such guidance can be specified as a general function of trajectory data (including  observed states and actions of an MDP). 
The method works well in an online environment without an extensive training period, which is particularly advantageous in a real-world setting where an offline environment for model training is usually not available. According to independent evaluations, the proposed method achieved the highest performance in the recent 2021 CityLearn Challenge. To demonstrate effectiveness against existing techniques, we further validate the method by comparing it with a range of baselines. Key contributions are as follows:
\begin{itemize}
    \item A framework of adaptive optimization for online control, winning the \textbf{1st place in the 2021 CityLearn Challenge}
    \item A novel evolutionary search (ES) algorithm  with a guidance function based on state-action-trajectory data
    \item Theoretical analysis of asymptotic convergence to global optima with noisy function evaluations
    \item Empirical comparison against a range of baselines
\end{itemize}

\subsection{Related work}
\label{sec:related}

{Optimal control and stochastic optimal control are well-known approaches to sequential decision-making problems} \cite{bertsekas2019reinforcement}. {Convex optimization  is another avenue} \cite{agrawal2020learning}. Most existing works assume a known dynamic model of the system. Various large-scale stochastic programs have been proposed in the literature to deal with future uncertainty \cite{prekopa2013stochastic,powell2020reinforcement}. The major drawback is the computational burden of rapidly expanding scenario trees in multi-stage stochastic programming. Our method relieves computation by using plug-in estimators, a.k.a., deterministic approximation of future uncertainty within convex optimization. 

Recently, RL has gained popularity in various domains \cite{chen2022reinforcement,haydari2020deep,nian2020review}. To contextualize the present approach, we make a few remarks about the relation to model-based RL (MBRL). In particular \emph{implicit} MBRL, where the entire procedure (e.g., model learning and planning) is optimized for optimal policy computation \cite{moerland2020model}. However, unlike existing works (e.g., \cite{tamar2016value,karkus2017qmdp,racaniere2017imagination,guez2018learning,schrittwieser2020mastering} that build a model based on (recurrent/convolutional) neural networks (NNs) with primary restrictions to discrete state and action space, our method learns how to plan by solving optimization and adapting its parameters; hence, it is amenable to a wide range of applications with continuous states and actions. The present work is closely related to \cite{ghadimi2020reinforcement,agrawal2020learning}, which also use convex optimization as a policy class to handle uncertainty. In particular, convex optimization control policies are learned in \cite{agrawal2020learning} with implicit differentiation \cite{agrawal2019differentiable}. We extend their method to the RL setting and propose a novel ES algorithm for learning.

Differentiated from existing ES-based strategies \cite{szita2006learning,salimans2017evolution,khadka2018evolution,gangwani2018policy,conti2018improving} or zeroth-order optimization \cite{snoek2012practical,liu2020primer}, our ES is augmented with a guidance function that depends on the data collected as the policy interacts with the environment; hence, it can be viewed as a type of \emph{MDP-augmented ES}. The guidance mechanism is also flexible enough to allow the effective incorporation of domain knowledge, as demonstrated in CityLearn. We further provide theoretical justification for this rather complex scheme.

\section{Preliminaries}
\label{sec:Preliminaries}

\subsection{Problem setup}

Consider an MDP $(\mathcal{S},\mathcal{A},\mathcal{P},r)$, where $\mathcal{S}$ is the (possibly infinite) state space, $\mathcal{A}$ is the set of actions, $\mathcal{P}:\mathcal{S}\times\mathcal{A}\rightarrow\mathcal{M}(\mathcal{S})$ is the transition probability kernel with $\mathcal{M}(\mathcal{S})$ denoting the set of all probability measures over $\mathcal{S}$, and $r(s,a)$ gives the corresponding immediate reward (can be time-dependent). The goal in episodic RL is to learn a policy $\pi:\mathcal{S} \rightarrow \mathcal{A}$ that maximizes cumulative rewards over a finite time horizon:
\begin{equation}
    \begin{aligned}
    \underset{\pi \in \Pi}{\max} \quad & \mathbb{E}\left[R(\pi)  \right]{,}
    \end{aligned}
    \label{eq:SeqDecMaking}
\end{equation}
where $R(\pi)\coloneqq{\sum}_{t=0}^T \hspace{0.1cm} r_{t}\big(s_t, \pi(s_t) \big)$ is the episodic reward, with $s_t \in \mathcal{S}\subseteq\mathbb{R}^{n_s}$ denoting the state at time $t$ and $T$ as the horizon. The expectation is taken over initial state distributions and transition dynamics (under deterministic policy $\pi$). We require access to a random sample $R(\pi)$ of the episodic reward as well as trajectory data $\{(s_t,a_t)\}_{t\in[T]}$ to be used to later compute the guidance signal. Here, we use the shorthand $[T]=\{1,...,T\}$.

\subsection{Canonical approaches and solution functions}
\label{subsec:canonical approaches}

The proposed method is based on canonical  stochastic programming methods \cite{powell2020reinforcement}. For example, in multi-stage stochastic programming \cite{pflug2014multistage}, the action at state $s_t$ is computed as
{
\begin{equation*}
     \underset{a_t\in\mathcal{A}}{\argmax} \Big( \tilde{r}_{t}(s_t, a_t  )\! + \! \!\! \!  \underset{\{a_{t'}\}_{t'=t+1}^T}{\max} \!\!\!\tilde{\mathbb{E}}\Big[\!\! \mathlarger{\sum}_{t' = t+1}^T \! \tilde{r}_{t'}\big(s_{t'}, a_{t'})\Big | s_t,a_t
    \Big] \Big),
\end{equation*}}
where $\tilde{r}_t:\mathcal{S} \times \mathcal{A}\rightarrow \mathbb{R}$ is the {surrogate} reward function and $\tilde{\mathbb{E}}[\cdot]$ is the {surrogate} expectation operator (e.g., model-based scenario trees) designed to {approximate the true environment}. The above formulation can also be seen as finding the solution to a Bellman equation in dynamic programming. The main limitations, nevertheless, are the high computational cost of evaluating the expectation operator and the potential model mismatch due to approximations. 

A simpler yet more practically appealing method, widely adopted in today's industries, is to use deterministic approximations of the future and capture the dependence of future states on prior decisions through constraints as part of the lookahead model \cite{powell2020reinforcement}: $\pi_{\zeta}(s_t)$ is given by
\begin{equation}
    \begin{aligned}
     & \underset{\bar{a}_t\in\mathcal{A}}{\argmax} \max_{\{\bar{s}_{t'}, \bar{a}_{t'}\}_{t'=t+1}^{T}}\bar{R}(s_t,\bar{a}_{t},\{\bar{s}_{t'},\bar{a}_{t'}\}_{t'=t+1}^T;\zeta)\\
     & \hspace{0.02cm}\text{s. t.} \quad  g_i(s_t,\bar{a}_{t},\{\bar{s}_{t'},\bar{a}_{t'}\}_{t'=t+1}^T;\zeta) \leq 0 \hspace{0.17cm}; \;\; \hspace{0.09cm}\forall i \in\mathcal{I}\\
     & \hspace{0.04cm}\quad\;\quad h_i(s_t,\bar{a}_{t},\{\bar{s}_{t'},\bar{a}_{t'}\}_{t'=t+1}^T;\zeta) = 0 \hspace{0.19cm} ; \;\;\; \forall i \in\mathcal{E},
     \end{aligned}
     \label{eq:sol-map}
\end{equation}
with the objective $\bar{R}(s_t,\bar{a}_{t},\{\bar{s}_{t'},\bar{a}_{t'}\}_{t'=t+1}^T;\zeta)$ defined as
\begin{equation*}\label{eq:surrogate-reward}
    \bar{r}_{t}\big({s_t}, \bar{a}_t;\zeta \big) + {\sum}_{t' = t+1}^T \bar{r}_{t'}\big(\bar{s}_{t'}, \bar{a}_{t'};\zeta\big),
\end{equation*}
where $\{\bar{s}_{t'}\}_{t'=t+1}^T$ and $\{\bar{a}_{t'}\}_{t'=t}^T$ are optimization variables corresponding to the planned states and actions, $\bar{r}_t$ is surrogate reward, and the feasible set is defined by $\{g_i\}_{i \in\mathcal{I}}$ and $\{h_i\}_{i\in\mathcal{E}}$. We denote the parameters of the objective function and the constraints collectively by $\zeta\in\mathcal{Z}\subset\mathbb{R}^d$. The dependencies of future states on current and \emph{planned} states/actions are encoded as constraints in \eqref{eq:sol-map} as part of the lookahead model. Many examples can be found in, e.g., \cite{borrelli2017predictive}. We remark that some optimization parameters may be provided by predictors based on the current state $s_t$, the parameter of which is also collected by $\zeta$ in this case.

 The policy $\pi_{\zeta}(s_t)$, a.k.a., the solution function \cite{dontchev2009implicit}, provides the action at the current state $s_t$ as the optimal solution to \eqref{eq:sol-map}. Since this function is generally set-valued \cite{dontchev2009implicit}, we make the following assumption.
\begin{assumption}\label{asmptn:uniqueness}
For each $\zeta\in \mathcal{Z}$ and $s_t\in\mathcal{S}$: \emph{a)} the objective function in  \eqref{eq:sol-map} is continuous, strictly convex, $g_i$ is continuous and convex for each $i\in\mathcal{I}$, and $h_i$ is affine for each $i\in\mathcal{E}$; \emph{b)} the feasible set of \eqref{eq:sol-map} is closed, absolutely bounded, and has a nonempty interior.
\end{assumption}
The above assumption can be satisfied by imposing proper conditions on the design of the surrogate model, i.e., objective and constraints in \eqref{eq:sol-map}. Note that in our approach, we make no convexity assumption about the true dynamics or rewards of the environment, which can be seen as a blackbox. The convexity condition is only stipulated for the surrogate model for computational efficiency. Our goal is simply to learn the parameters of the optimization model in order to have a good decision-making capability.\footnote{Perhaps surprisingly, despite the fact that \eqref{eq:sol-map} is convex, the policy (given by the solution function) can be highly nonconvex with high representational capacity. In particular, the contemporary work by \citep{2023_4C_Sol} shows a ``universal approximation'' property of the solution functions of linear programs (LPs). } 
An immediate consequence of the above assumption is that the solution to \eqref{eq:sol-map} is unique; furthermore, it implies continuity with respect to parameters.
\begin{lemma}\label{lem:continuity}
The solution function $\pi_\zeta(s_t)$ defined in \eqref{eq:sol-map} is continuous  with respect to parameter $\zeta$ for each $s_t\in\mathcal{S}$.
\end{lemma}
The proof is a direct application of the Berge maximum theorem \cite{berge1997topological}. To conclude this section, let us make some comments on the construction of the surrogate model in \eqref{eq:sol-map}. By analogy to reward design \cite{prakash2020guiding}, the objective function should be chosen to promote desirable behaviors. The set of constraints introduces inductive bias on the transition dynamics of the environment. It is beneficial, though oftentimes unlikely and non-essential, that the surrogate model matches the functional forms of true reward or dynamics, an idea shared in model-based RL \cite{moerland2020model}. It is, nevertheless, desirable to ensure the computational efficiency of \eqref{eq:sol-map} to provide actions quickly---hence the choice of convex programs.

\section{Policy adaptation with ES under guidance}
\label{sec:ZO-iRL framework}

The potential mismatch between the surrogate model and the real environment and errors due to predictions may adversely affect the decision quality of \eqref{eq:sol-map}. Thereby, we aim to adapt the parameters of the surrogate model to shape the solution function. The task of finding the optimal parameter within the set of solution functions $\Pi=\{\pi_\zeta:\zeta\in\mathcal{Z}\}$ can be compactly written as:
\begin{equation}
    \Upsilon\coloneqq \underset{\zeta\in\mathcal{Z}}{\argmax}\hspace{0.05cm} \mathbb{E}\Bigg[\mathlarger{\sum}_{t=0}^T r_{t}\big(s_t, \pi_{\zeta}(s_t) \big) \Bigg],
    \label{eq:OuterLayer}
\end{equation}
where $\Upsilon$ is the set of global optima for policy parameters. Note that $\zeta$ is not part of the true reward (which remains unknown to the agent) but only the parameters of the surrogate model that implicitly defines the policy in \eqref{eq:sol-map}. Since $\pi_{\zeta}(s_{t})$ is given by an optimization \eqref{eq:sol-map}, \eqref{eq:OuterLayer} can also be viewed as a bi-level problem (c.f., \cite{dempe2020bilevel}): the outer level aims to learn parameters to maximize rewards, while the inner level defines  policy action as a solution to~\eqref{eq:sol-map}. The key challenge in solving \eqref{eq:OuterLayer} as a bi-level problem, however, is that the outer level objective is not analytically revealed (thus prohibiting any direct differentiation approach) and can be nonconvex with respect to $\zeta$.

\begin{algorithm}[t]
\caption{Evolutionary search under guidance}\label{alg}
\KwInput{Hyperparameters $\{N_k\}$, uniform distribution $\mu$, initial point $z_0$ }
  \begin{algorithmic}[1]
  \STATE Initialize $P_1(d\zeta) \sim \exp(\|\zeta-z_0\|)\mu(d\zeta)$ 
  \FOR{$k = 1, 2,\ldots$}
    \STATE Sample $N_k$ candidates from the distribution $p_k$: $\zeta_1^{k}, \zeta_2^{k}, \cdots, \zeta_{N_k}^{k}\overset{\text{iid}}{\sim} p_k$
    \FOR{$j = 1, \ldots, N_k$}
            \STATE Deploy policy $\pi_{\zeta_j^{k}}$ for one episode and observe an episodic reward $R_j^{k}\leftarrow R(\pi_{\zeta_j^{k}})$
            \STATE Compute the guidance signal $\varrho_j^k$ by \eqref{eq:guidance}
        \ENDFOR
    \STATE Update the distribution $p_{k+1}$ for the next iteration according to \eqref{eq:iteration-measure}.
    \ENDFOR
  \end{algorithmic}
  \label{alg:Zeroth_Order method}
\end{algorithm}

\subsection{Guided evolutionary search}

This section discusses the proposed ES algorithm (Algorithm~\ref{alg}), inspired by the method of generations \cite{zhigljavsky2012theory}. At each iteration $k$, the algorithm randomly samples a set of $N_k$ parameter candidates, $\zeta_1^{k}, \cdots ,\zeta_{N_k}^{k} \overset{iid}{\sim} p_k$.
{ For each candidate $j\in[N_k]$, we evaluate the corresponding policy in the environment and observe an episodic reward $R_j^k\sim \mathcal{R}(\pi_{\zeta_j^{k}})$, where $\mathcal{R}(\pi_{\zeta_j^{k}})$ denotes the distribution of episodic reward for policy $\pi_{\zeta_j^{k}}$, as well as all the state and action pairs $\{s^{j,k}_t,a^{j,k}_t\}_{t\in[T]}$ in the past episode. Based on  trajectory data, we compute a guidance signal 
\begin{equation}
    \label{eq:guidance}
    \varrho_j^k=\phi(\{s^{j,k}_t,a^{j,k}_t\}_{t\in[T]}),
\end{equation}
where $\phi:(\mathcal{S}\times\mathcal{A})^T\rightarrow\mathcal{Z}'$ is a function that may have complicated dependence on past states and actions with $\mathcal{Z}'$ as the range. The design of such a guidance function is often based on domain knowledge (to be discussed later in CityLearn). Then, we update the distribution for the next iteration as
\begin{equation}\label{eq:iteration-measure}
        p_{k+1}(d\zeta) = \sum_{j = 1}^{N_k}\nu_j^kQ_k(\zeta_j^k,\varrho_j^k, d\zeta),
\end{equation}
where
\begin{equation}\label{eq:compute-weights}
    \nu_j^k = \frac{\exp(R_j^{k})}{\sum_{j=1}^{N_k}\exp(R_j^{k})}
\end{equation}
are the weights obtained by taking the softmax over candidate rewards. The probability measure $Q_k(\zeta_j^k,\varrho_j^k,d\zeta)$ is the transition probability given candidate $\zeta_j^k$ and guidance signal $\varrho_j^k$. Hence, $p_{k+1}(d\zeta)$ is a mixture of distributions weighted by observed rewards in the current iteration $k$, which can be sampled by the standard superposition method: at first the index $j$ is sampled from the discrete distribution $\{\nu_j^k\}$, followed by sampling from $Q_k(\zeta_j^k,\varrho_j^k,d\zeta)$. For example, in our algorithm for CityLearn,
\begin{equation}\label{eq:Q_normal-main}
    Q_k(z,\varrho,d\zeta)\sim\exp(\|\zeta-z-\alpha_k\varrho\|/\iota_k)\mu(d\zeta),
\end{equation}
where $\|\cdot\|$ is the Euclidean norm, $\mu(d\zeta)$ is a uniform measure over $\mathcal{Z}$, and $\iota_k>0$ and $\alpha_k\geq 0$ are such that their sum over time is bounded. Other candidates are possible and can still ensure convergence to global optimal, as long as certain conditions are met; intuitively, we require that the span of $Q_k$ decreases over time but not so rapidly that it fails to reach a global optimum. Note that to simplify the presentation, in the above algorithm, we assume that each candidate policy is evaluated only on one episode; extending this to the case of multiple episodes is straightforward (e.g., we would instead take the average of the evaluations among the episodes in the computation of weights \eqref{eq:compute-weights}).}

\section{Theoretical analysis}
\label{sec:theory}

We now analyze the convergence property of the sequence generated by Algorithm \ref{alg}. The following notations are used: $f(\zeta)=\mathbb{E}[R(\pi_\zeta)]$ is the expected episodic reward of policy $\pi_\zeta$, $\Upsilon=\argmax_{\zeta\in\mathcal{Z}} f(\zeta)$ is the set of global maximizers (may not be unique), $f^*=\max_{\zeta\in\mathcal{Z}} f(\zeta)$ is the global maximum, and $\lambda(d\zeta)$ is some measure over $\Upsilon$; $\mathbb{B}(\zeta,\epsilon)=\{\zeta'\in\mathcal{Z}:\|\zeta'-\zeta\|\leq \epsilon\}$ is a ball centered at $\zeta$ with radius $\epsilon$, $\mathbb{B}^*(\epsilon)=\{\zeta\in\mathcal{Z}:\min_{\zeta'\in\Upsilon}\|\zeta'-\zeta\|\leq \epsilon\}$ is a set of points that are $\epsilon$ away from the optimal solution set $\Upsilon$; $\delta_\zeta(dz)$ is the probability measure concentrated at the point $\zeta$. We use $\Rightarrow $ to denote the weak convergence of measures. We can consider $\|\cdot\|$ as any norm (e.g., Euclidean norm).

The measures $p_{k+1}(d\zeta)$, $k\in\mathbb{N}$ defined in \eqref{eq:iteration-measure} are distributions of random points $\zeta_j^{k+1}$, for any $j\in[N_{k+1}]$, conditional on the results of preceding evaluations of $\{R_j^{k}\}_{j\in[N_k]}$ and realizations of $\{\zeta_j^k,\varrho_j^k,\xi_j^k\}_{j\in[N_{k}]}$. Let $P_k(d\zeta_1,...,d\zeta_{N_k})$ represent their unconditional joint distributions at iteration $k$, and $$\tilde{P}_k(d\zeta)=\int_{\mathcal{Z}^{N_k-1}}P_k(d\zeta,dz_2,...,dz_{N_k})$$ is the unconditional marginal distribution (note that we introduce $z$ for $\zeta$ as the need arises in  integration).

The formalism of the guidance signal requires some basics of random process and measure theory (details can be found in the appendix). Essentially, the guidance signal $\varrho_j^k\in\mathcal{Z}'$ is a random variable (adapted to the $\sigma$-algebra generated by the trajectory within an episode) with probability measure $M_k(\zeta_j^k,d\varrho)$. Note that $M_k(\zeta_j^k,d\varrho)$ is dependent on $\zeta_j^k$ because the stochastic process that generates the trajectory depends on policy $\pi_{\zeta_j^k}$, but $M_k(\zeta_j^k,d\varrho)$ is conditionally independent of all other candidates $\{\zeta_{j'}^k\}_{j'\neq j}$. For analysis, we make the following assumptions.
\begin{assumption}\label{asm:main}
The followings hold:
\begin{enumerate}[label=(\alph*),leftmargin=2em]\setlength\itemsep{-0.058em}
    \item $R_j^k=f({\zeta_j^{k}})+\xi_j^k$, where $\xi_j^k\overset{iid}{\sim} F_k(d\xi)$ for any $k\in\mathbb{N}$ are independent with distribution $F_k(d\xi)$ bounded on a finite interval $[-c_\xi,c_\xi]$ and $\mathbb{E}\exp(\xi_j^k)=1$;
    \item $|f(\zeta)|\leq c_f$ for all $\zeta\in\mathcal{Z}$ and $\mathcal{Z}$ is compact;
    \item there exists $\epsilon>0$ such that $f$ is continuous on $\mathbb{B}^*(\epsilon)$;
    \item $Q_k(z,\varrho,d\zeta)=q_k(z,\varrho,\zeta)\mu(d\zeta)$, with $\sup_{z,\varrho,\zeta\in\mathcal{Z}}q_k(z,\varrho,\zeta)\leq L_k<\infty$ for all $k\in \mathbb{N}$, where $\mu$ is a probability measure on $\mathcal{Z}$ such that $\mu(\mathbb{B}^*(\epsilon))>0$ for any $\epsilon>0$; for any $z\in\mathcal{Z}$, the sequence of probability measures $Q_k(z,\varrho,d\zeta)$ weakly converges to $\delta_z(d\zeta)$;
    \item $\{N_k\}$ is a sequence of natural numbers $N_k\in\mathbb{N}$ such that $N_k\rightarrow\infty$ for $k\rightarrow\infty$;
    \item $\sup_{z,\varrho}M_k(z,d\varrho)<\infty$ for all $k\in\mathbb{N}$;
    \item $\tilde{P}_1(\mathbb{B}(\zeta,\epsilon))>0$ for all $\epsilon>0$ and $\zeta\in\mathcal{Z}$; 
    \item for any $\epsilon>0$, there are $\delta>0$ and a natural $\bar{k}$ such that $\tilde{P}_k(\mathbb{B}^*(\epsilon))\geq \delta$ for all $k\geq \bar{k}$.
\end{enumerate}
\end{assumption}

Let us comment on the assumptions above. Condition $(a)$ requires that the evaluation noise be independent and bounded; the expectation requirement can be satisfied for truncated log-normal distributions \cite{thompson1950truncated}. The iid requirement can be relaxed to mixing processes at the cost of more complex analysis \cite{doukhan2012mixing}; the boundedness condition, on the other hand, seems necessary to keep iterates in the vicinity of global maximum if they are already there. Condition $(b)$ is non-restrictive for practical problems. Condition $(c)$ is natural since $\pi_\zeta$ is continuous by Lemma \ref{lem:continuity}, and can be satisfied if the true reward functions $r_t$ are continuous for all $t\in[T]$. Assumptions $(d), (e), (f), (g)$, and $(h)$ formulate necessary requirements on the parameters of the algorithm. Intuitively, conditions $(d)$, $(e)$, and $(f)$ stipulate that the search becomes more ``focused'' over time in order to concentrate on the global optima; however, conditions $(g)$ and $(h)$ indicates that the decrease of span cannot be too fast in order not to miss the global optima. Condition $(f)$ on the guidance function $\phi$ can be met by proper smoothing if needed. Condition $(e)$ can be relaxed to $N_k=N$ for some finite integer $N$ for all $k\in\mathbb{N}$, but the convergence will only be towards the vicinity of $\Lambda$ due to the finite sample effect (see Lemma 4 in the appendix, which states the rate to be on the order $N^{-1/2}$). Unlike $(c), (d), (e), (f)$, and $(g)$, condition $(h)$ is not constructive; hence, we provide some verifiable conditions sufficient for $(h)$ to hold in Corollary \ref{cor:normal} (also see Corollary 2 in the appendix). Next, we analyze the update rule of Algorithm \ref{alg}.

\begin{lemma}\label{lem:distr-update-rule-guide}
The probability distribution $P_{k+1}(d\zeta_1,...,d\zeta_{N_{k+1}})$ can be written in terms of the distribution $P_k(d\zeta_1,...,d\zeta_{N_k})$ as:
\begin{align}
    \label{eq:distr-update-rule}
    \int_{\Omega^{N_k}}\!\!\chi_k(d\omega_{N_k})\!\prod_{j=1}^{N_{k+1}}\! \left\lbrace\beta(\omega_{N_k})\sum_{i=1}^{N_k}\Lambda(z_i,\varrho_i,\xi_i,d\zeta_j)\right\rbrace,
\end{align}
where $\Omega=\mathcal{Z}\times \mathcal{Z}'\times[-c_\xi,c_\xi]$, 
\begin{align*}
    &\omega_{N_k}=\{z_1,...,z_{N_k},\varrho_1,...,\varrho_{N_k},\xi_1,...,\xi_{N_k}\}\in\Omega^{N_k},\\
    &\chi_k(d\omega_{N_k})=P_k(dz_1,...,dz_{N_k})\prod_{j=1}^{N_{k}}F_k(d\xi_j)M_k(z_j,d\varrho_j),\\
    &\beta(\omega_{N_k})=\frac{1}{\sum_{j=1}^{N_k}\exp(f(z_j)+\xi_j)},\quad\text{ and}\\
    &\Lambda(z,\varrho,\xi,d \zeta)=\exp(f(z)+\xi)Q_k(z,\varrho,d\zeta).
\end{align*}
\end{lemma}
The proof is immediate by recognizing that the bracket term in \eqref{eq:distr-update-rule} is the conditional distribution $p_{k+1}(d\zeta_j)$ defined in \eqref{eq:iteration-measure} and the integration is over the distribution from the preceding iteration. We take the product over $N_{k+1}$ candidates since they are drawn iid from $p_{k+1}$.

Now, we provide the main result on the convergence of $\tilde{P}_k(d\zeta)$  to some distribution $\lambda(d\zeta)$ over the global optima.

\begin{theorem}\label{thm:conv-general}
Suppose that Assumption \ref{asm:main} holds true, and let $\{\tilde{P}_k\}$ be the sequence of unconditional marginal distributions determined by Algorithm \ref{alg}.  Then, the distribution sequence  weakly converges to some measure  $\lambda$ over the optimal set, i.e., $\tilde{P}_k\Rightarrow \lambda$ as $k\rightarrow\infty$.
\end{theorem}
The key stage of proof is to show that there exists a subsequence in $\{\tilde{P}_k\}$ that weakly converges to  \begin{equation*}
    \vartheta_m(d\zeta)=\frac{\exp(mf(\zeta))\kappa(d\zeta)}{\int\exp(mf(z))\kappa(dz)}
\end{equation*}
for some measure $\kappa$, where $m$ is the index of the subsequence. The above distribution is effectively a softmax function over function values and converges to the extrema as $m\rightarrow\infty$. 

All the conditions in Assumption \ref{asm:main} are natural  with the exception of $(h)$, which requires some further justification. In the following, we present a sufficient condition for $(h)$ with a proper design of $Q_k(z,\varrho,d\zeta)$, which applies to the case of noisy function evaluations; see the appendix for another example in the case of noiseless function evaluations.

\begin{corollary}\label{cor:normal}
Under Assumption \ref{asm:main} (except $(h)$), and let the transition probability $Q_k(z,\varrho,d\zeta)$ be
\begin{equation}\label{eq:Q_normal}
    Q_k(z,\varrho,d\zeta)=c_k(z,\varrho)\psi((\zeta-z-\alpha_k\varrho)/\iota_k)\mu(d\zeta),
\end{equation}
where $c_k(z,\varrho)=({\int\psi((\zeta-z-\alpha_k\varrho)/\iota_k)\mu(d\zeta)})^{-1}$ is the normalization term, $\psi$ is a continuous symmetrical finite density, and
\begin{equation*}
    \iota_k>0,\quad\sum_{k=1}^\infty \iota_k<\infty,\quad\alpha_k\geq 0,\quad\sum_{k=1}^\infty \alpha_k<\infty.
\end{equation*}
Then, there exists a sequence of natural numbers $\{N_k\}$ such that $\{\tilde P_k\}$ weakly converges to $\lambda$.
\end{corollary}

Our analysis accounts for the effect of trajectory-based guidance, which is a novel contribution to the ES literature. 
By examining the proof, we can relax the condition that $\sum_{k=1}^\infty \alpha_k<\infty$, i.e., continue applying the guidance without the need to diminish its impact in the long run, as long as the guidance signal ``approximately'' points to the global optima in the proximity (see the appendix for exact conditions). However, such guidance can be difficult to design or even verify in practice; thus, it is still advisable to relinquish human knowledge and let data drive the decision, eventually.

\section{Results from the CityLearn Challenge}
\label{sec:Experiments}

\emph{Challenge overview.} The competition has an online setup with a simulation period of 1 or 4 years, where agents exploit the best policies to optimize the coordination strategy. The goal of each agent is to minimize environmental costs, such as ramping costs, peak demands, 1-load factor, and carbon emissions. The state space contains information such as daylight hours, outdoor temperature, storage device state of charge (SOC), net electricity consumption of the building, carbon intensity of the power grid, among a total of 30 continuous states. The agent is allowed to control the charging/discharging actions of storage devices for domestic hot water (DHW), chilled water, and electricity (i.e., 3 continuous actions per building). The environment is seen as a blackbox to the agent as a standard RL setup, where the transition dynamics depend on the responses of various devices (e.g., air-to-water heat pumps, electric heaters) as well as the energy loads of buildings, which include space cooling, dehumidification, DHW demand, and solar generation. 

\emph{Evaluation.} The submission of each team is evaluated on a set of metrics, including: \emph{(1)} ramping: $\sum |e_t-e_{t-1}|$, where $e$ is the net electricity consumption at each time step; \emph{(2)} 1-load factor:  average net electricity load divided by  maximum electricity load; \emph{(3)} average daily peak demand; \emph{(4)} maximum peak electricity demand; \emph{(5)} total electricity consumed; \emph{(6)} carbon emissions. The competition evaluates performance by computing the ratio of costs with respect to a rule-based controller (RBC)---lower ratios indicate better performances.\footnote{Note that the RBC controller is ubiquitous in traditional building control systems and is a simple form of ``take action $a_h$ in hour $h$,'' where $a_h$ is a constant independent of current states except for the hour of the day ($h\in[24]$). } The average of the above metrics for the full simulated period is the \emph{total score}, while the average of the metrics \emph{(1)}-\emph{(4)} is \emph{coordination score}. The performance of the top 4 teams is listed in Table \ref{tab:performance}. Refer to \cite{vazquez2020citylearn} for more details on the contest. 

\begin{figure*}[t]
  \centering
  \includegraphics[width=1.61\columnwidth]{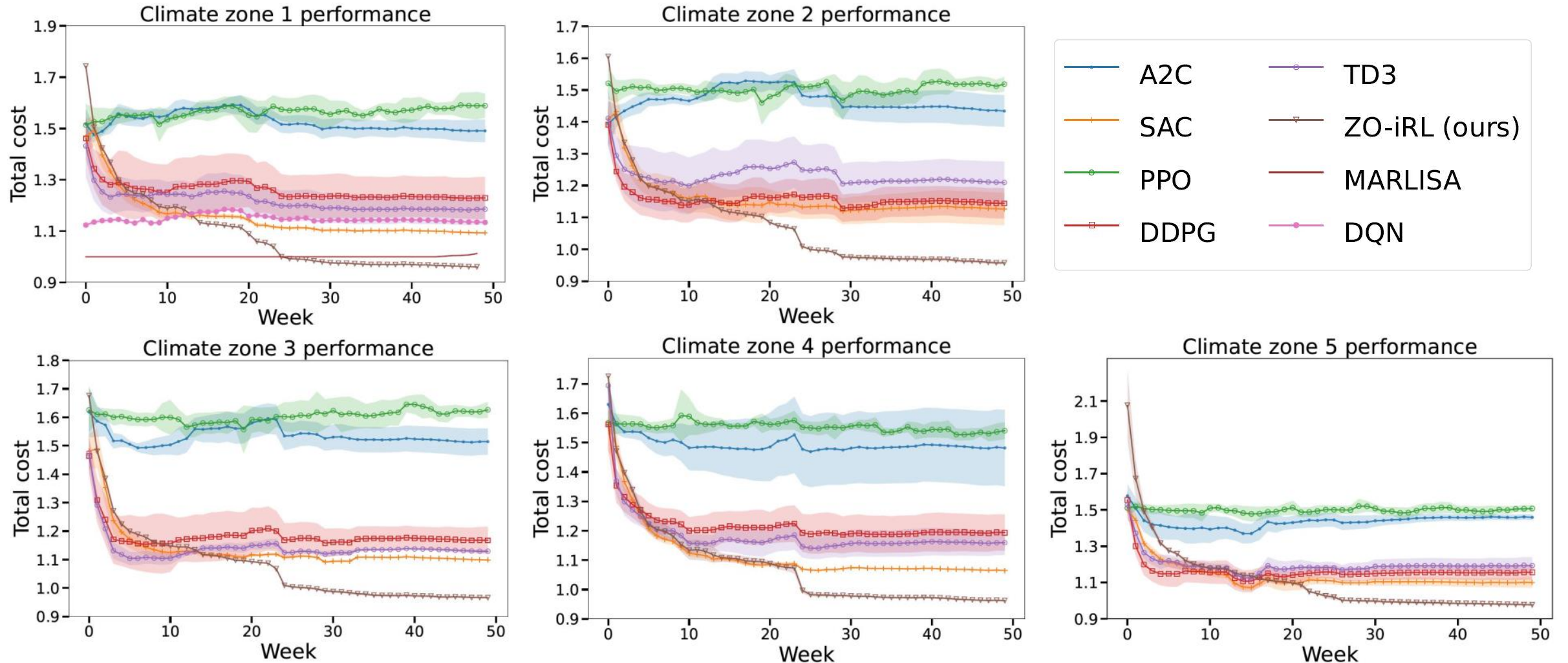}
\caption{Learning curves of ZO-iRL and baselines. We perform 10 runs on each baseline to obtain performance plots with standard deviations for Climate Zones 1 to 5.}
\label{fig:Learning curve}
\end{figure*}

\begin{table}[t]
\centering
\begin{adjustbox}{width=\columnwidth,center}
\begin{tabular}{@{}llllll@{}}
& 
  \multicolumn{1}{c}{\textbf{ZO-iRL}(ours)} &
  \multicolumn{1}{c}{\textbf{ICD-CA}} &
  \multicolumn{1}{c}{\textbf{IDLab}} &
  \multicolumn{1}{c}{\textbf{Breakfast Club}} \\ \toprule
\textbf{Total score}      

& \hspace{0.2cm}\textbf{0.944}
& \hspace{0.1cm}1.070
& 1.070
& \hspace{0.4cm}1.130
 \\
\textbf{Coord. score}     
& \hspace{0.2cm}\textbf{0.915}
& \hspace{0.1cm}1.107
& 1.098
& \hspace{0.4cm}1.095
 \\
\end{tabular}
\end{adjustbox}
\caption{Total and coordination scores of the top 4 teams in 2021 CityLearn Challenge.  }\label{tab:performance}
\end{table}

\begin{table*}[t]
\centering
\begin{adjustbox}{width=2\columnwidth,center}
\begin{tabular}{@{}llllllllll@{}}
& 
  
  \multicolumn{1}{c}{\textbf{SAC}} &
  \multicolumn{1}{c}{\textbf{A2C}} &
  \multicolumn{1}{c}{\textbf{DDPG}} &
  \multicolumn{1}{c}{\textbf{DQN}} &
  \multicolumn{1}{c}{\textbf{PPO}} &
  \multicolumn{1}{c}{\textbf{TD3}}&
  \multicolumn{1}{c}{\textbf{MARLISA}} &
  \multicolumn{1}{c}{\textbf{ZO-iRL}}  \\ \toprule
\textbf{ramping cost}      

& 1.145 (0.015)
& 1.189 (0.002)
& 1.174 (0.029)
& 1.302 (0.008)
& 1.638 (0.005)
& 1.178 (0.026)
& \underline{1.022} (0.010)
& \textbf{0.781} (0.005) 
 \\
\textbf{1-load factor}     

& 1.158 (0.002)
& 1.146 (0.002)
& 1.143 (0.007)
& 1.159 (0.003)
& 1.168 (0.003)
& 1.142 (0.004)
& \underline{1.026} (0.006)
& \textbf{1.010} (0.008) 
 \\
\textbf{avg. daily peak}   

& 1.180 (0.002)
& 1.184 (0.008) 
& 1.195 (0.011)
& 1.212 (0.002)
& 1.242 (0.001)
& 1.193 (0.006) 
& \underline{1.015} (0.001) 
& \textbf{0.996} (0.001) 
 \\
\textbf{peak demand}       

& 1.077 (0.008)
& 1.088 (0.007) 
& 1.100 (0.010)
& 1.115 (0.009)
& 1.132 (0.006)
& 1.098 (0.013) 
& \underline{1.000} (6e-5)
& \textbf{0.962} (0.005) 
 \\
\textbf{net electric. peak} 

& \underline{0.995} (0.001)
& \textbf{0.994} (0.002) 
& 0.997 (8e-4)
& 0.997 (1e-4)
& 1.003 (7e-5)
& 0.997 (7e-4) 
& 1.000 (5e-4)
& 1.006 (2e-4) 
 \\
\textbf{carbon emissions}     
&   \textbf{1.000} (0.001)
& \underline{1.000} (0.002)    
& 1.003 (7e-4)
& 1.005 (1e-4)
& 1.009 (7e-5)
&1.004 (7e-4)  
&  1.001 (5e-4)  
&   1.007 (2e-4)    \\  \hline
\textbf{total score}       
& 1.092 (0.003)
& 1.101 (0.003)
& 1.102 (0.008)
& 1.132 (0.002)
& 1.199 (0.001)
& 1.102 (0.006)
& \underline{1.011} (0.002)
&\textbf{0.962} (0.001)
\end{tabular}
\end{adjustbox}
\caption{Comparison with baselines: SAC \cite{kathirgamanathan2020centralised} and MARLISA \cite{vazquez2020marlisa} have been officially implemented for CityLearn; other baselines are implemented by \cite{raffin2019stable}. The reported values are the average and standard deviation (in brackets) across 10 independent runs on Climate Zone 1 data.}\label{tab:overall}
\end{table*}

\textbf{ZO-iRL: zeroth-order implicit RL.} \footnote{We name our method ZO-iRL because the policy is implicitly determined by solving an optimization problem and the learning algorithm is zeroth-order in an RL setting.} As our method is designed for single-agent episodic RL, we first reduce the original task that consists of a single period of 1 or 4 years into episodes of 24 hours. We use the per-step reward $-\max(0,e_t)^3$ as recommended by \cite{vazquez2020citylearn}, where $e_t$ is the net electricity consumption (or generation if $e_t<0$). This reward favours consumption patterns that are smoothly averaged without demand peaks, aligned with multiple metrics used in the evaluation, such as the 1-load factor and peak electricity demand. Another reduction is from multi-agent RL to single-agent RL, where each building's policy is updated independently, reducing the problem to decentralized control with additive rewards; such a reduction is computationally efficient for large-scale problems \cite{de2021constrained}. We omit the notational dependence on candidate $j$ and iteration $k$ when presenting the method.

\emph{Optimization planner.} We instantiate the optimization in \eqref{eq:sol-map} as follows.
The planned states $\bar{s}_t$ consist of state variables such as net electricity consumption and SOCs of storage devices; the action $\bar{a}_t\in\mathcal{A}$ is the action of the MDP; the surrogate reward $$\bar{r}_{t}(\bar{s}_t;\zeta)=-|e_t - e_{t-1}| - \theta_t e_t$$
is a combination of the negated ramping cost and the ``virtual'' electricity cost, where $\zeta=\{\theta_t\in[0,5]\}_{t\in[24]}$ can be viewed as virtual electricity prices to be learned to encourage desirable consumption patterns (e.g., load flattening and smoothing). Intuitively, a higher value of $\theta_t$ discourages planned electricity consumption in the corresponding hour $t$. 

The inequalities are grouped into technological constraints (e.g., maximum/minimum cooling power) and constraints on states and actions. The equalities are grouped into physics accounting for energy balances (i.e., consumption equal to supply) and technology (e.g., SOC update rules). Further details are provided in the appendix. Note that to set up the optimization \eqref{eq:sol-map}, we also need to provide predictions of energy demands and solar generation. For simplicity, our predictors are based on a simple averaging scheme that takes the average in the corresponding hours of the last 2 weeks of data; thus, there are no specific parameters to learn. 

\emph{Transition and guidance.} We use \eqref{eq:Q_normal-main} as the transition probability, with variance $\iota_k=0.4/k^2$ that is initialized to 0.4 and decreases by $k^2$ in each iteration. The guidance signal $\varrho$ is computed as follows. By the end of each episode, we examine the net electricity usage in the past 24 hours, $e_t$ for $t\in[24]$ and find the top 2 hours with the most electricity usage, denoted by $t_1$ and $t_2$. Then, the guidance signal $\varrho_t$ is $0.02$ if $t\in\{t_1,t_2\}$ and $-0.04/22$ otherwise. 
Note that we have centered the signal ($\sum_t \varrho_t=0$) by assigning negative values for hours other than peaks. We choose $\alpha_k=1$ for all $k$ over the entire 4-year period, as there is no training phase in the CityLearn Challenge and we prefer to adapt quickly during the test phase; this is not a violation of our theory, as we can choose to diminish $\alpha_k$ after a while to still satisfy the condition $\sum_{k=1}^\infty \alpha_k<\infty$. 

\textbf{Results.} For baselines, we use the implementation of \cite{raffin2019stable} with the default ADAM optimizer, where the policy is an NN architecture with tanh activation and two layers of 256 units each. From Table \ref{tab:overall}, we see that \mbox{ZO-iRL} has achieved the lowest cost ratios (i.e., best scores) of all, which is consistent with the official result of the competition (Table \ref{tab:performance}). In particular, as shown in Fig.~\ref{fig:Learning curve}, ZO-iRL is able to find a good policy in the first few months, while baselines seem to struggle; we speculate that more samples would eventually improve the performance of baselines, and all methods may benefit from schemes to handle the potentially nonstationary environment due to seasonal patterns. 
\section{Conclusion and future directions}
\label{sec:Conclusion}
We presented a novel adaptive optimization framework that has been shown to be very effective for energy storage management. Using solution functions as policies offers a promising way to introduce data-driven algorithms into the real world where convex optimization has been widely adopted. To adapt the optimization parameters, we developed an evolutionary search algorithm that can incorporate insights from control trajectory data as guidance for parameter updates. The method outperforms several baselines and ranked first in the latest 2021 CityLearn Challenge. Some potential future directions could be to extend the proposed framework to other methods such as Bayesian optimization or first-order methods such as actor-critic.
 


\section{Acknowledgments}

We would like to thank all the ZORL/ROLEVT team members at the CityLearn Challenge 2021: Qasim Wani, Zhiyao Zhang, and Mingyu Kim. We would also like to thank the organizers of the CityLearn Challenge: Dr. José R. Vázquez-Canteli, Dr. Zoltan Nagy, Dr. Gregor Henze, and Sourav Dey for organizing this competition and inspiring novel RL-based approaches to the contemporary challenges facing the power and energy sector. This project is funded by NSF, the Commonwealth Cyber Initiative (CCI), C3.ai Digital Transformation Institute, and the U.S. Department of Energy.

\bibliography{references.bib}
\newpage
\appendix
\onecolumn

\section{Proof of results in the main paper}

\subsection{Formalism of the guidance signal}

The formalism of the guidance signal requires some basics from random process and measure theory (interested readers are referred to \cite{hajek2015random}). We keep our presentation minimal but sufficient enough to carry out the analysis. Consider a stochastic process $(S_t^{j,k},A_t^{j,k})_{t\in[T]}$ defined by policy $\pi_{\zeta_j^k}$ interacting with the MDP environment, where $S_t^{j,k}$ and $A_t^{j,k}$ are random variables representing the state and action at time $t$. Let $\mathcal{I}_t=(\mathcal{S}\times\mathcal{A})^t$ be the set of possible histories up to time step $t$ within an episode, and $I_t^{j,k}\coloneqq(S_1^{j,k},A_1^{j,k},...,S_t^{j,k},A_t^{j,k})\in\mathcal{I}_t$ is a random vector taking values in $\mathcal{I}_t$ containing all state-action pairs observed up to step $t$. Denote by $\mathcal{F}_t^{j,k}$ a non-decreasing sequence of $\sigma$-algebra (a \emph{filtration}) generated by $I_t^{j,k}$. Then, the guidance signal $\varrho_j^k\in\mathcal{Z}'$ is a random variable adapted to the filtration $\mathcal{F}_T^{j,k}$, i.e., $\varrho_j^k$ is $\mathcal{F}_T^{j,k}$-measurable, with probability measure $M_k(\zeta_j^k,d\varrho)$ associated with a properly defined probability space, the existence of which is ensured by the Ionescu-Tulcea theorem. Note that the distribution of $\varrho_j^k$ depends on $\zeta_j^k$, since the stochastic process $I_t^{j,k}$ is determined by the policy $\pi_{\zeta_j^k}$, but is conditionally independent of all other candidates $\zeta_{j'}^k$ for $j'\neq j$.

\subsection{Proof of Lemma \ref{lem:continuity}}

Let $\Phi(s_t,\zeta)$ represent the feasible set of \eqref{eq:sol-map}. By Assumption \ref{asmptn:uniqueness}, $\Phi(s_t,\zeta)$ is convex for fixed $s_t$ and $\zeta$ and has a nonempty interior. This implies that $\Phi(s_t,\zeta)$ is continuous in $s_t$ and $\zeta$ \cite[example 5.10]{rockafellar2009variational}. Hence, by Berge maximum theorem \cite{berge1997topological}, $\pi_\zeta(s_t)$ is upper hemicontinuous in $\zeta$ for fixed $s_t\in\mathcal{S}$. However, we know that $\pi_\zeta(s_t)$ contains a single point due to the strict convexity of the objective function. Thus, for fixed $s_t\in\mathcal{S}$, $\pi_\zeta(s_t)$ is a single-valued function continuous in its parameter $\zeta$.

\subsection{Proof of Theorem \ref{thm:conv-general}}
 
Select from $\{\tilde{P}_k\}$ a weakly convergent subsequence $\{\tilde{P}_{k_i}\}$, which is possible due to Prohorov's theorem \cite[Ch. 6]{billingsley2013convergence}, and denote the limit by $\kappa(d\zeta)$. By Lemma \ref{lem:update-distr-guide}, we have that
\begin{equation}\label{eq:thm-update-P}
    \tilde{P}_{k+1}(d\zeta)=\left(\int\tilde{P}_k(dz)\exp(f(z))\right)^{-1}\int \tilde{P}_k(dz)\exp(f(z))\Big(Q_k(z,\varrho,d\zeta)M_k(z,d\varrho)+\Delta_{N_k}(d\zeta)\Big).
\end{equation}
By Assumption \ref{asm:main} $(d)$ and $(h)$, it follows that the subsequence $\{\tilde{P}_{k_i+1}\}$ weakly converges to the distribution $\vartheta_1(d\zeta)=c_1\exp(f(\zeta))\kappa(d\zeta)$, where $c_1$ is the normalization constant. Similarly, the subsequence $\{\tilde{P}_{k_i+m}\}$ weakly converges to the distribution \begin{equation*}
    \vartheta_m(d\zeta)=\frac{\exp(mf(\zeta))\kappa(d\zeta)}{\int\exp(mf(z))\kappa(dz)},
\end{equation*}
which, by Lemma \ref{lem:soft-max-converge-guide}, converges to $\lambda$. Thus, by the standard  diagonalization argument \cite{billingsley2013convergence}, we can show that there exists a subsequence $\{\tilde{P}_{k_j}\}$ that weakly converges to $\lambda$. Applying Lemma \ref{lem:update-distr-guide} again yields that $\{\tilde{P}_{k_j+1}\}$ converges to the same limit. Thus, any subsequence of $\{\tilde{P}_k\}$ converges to this limit, and the same holds for the sequence itself.

\subsection{Proof of Corollary \ref{cor:normal}}

Under Assumption \ref{asm:main} (except $(h)$), the distributions \eqref{eq:update-eq-distr} have continuous densities with respect to the Lebesgue measure. Let $A(\epsilon)=\{\zeta\in\mathcal{Z}:f(\zeta)\geq f^*-\epsilon\}$. By \eqref{eq:Q_normal} and Lemma \ref{lem:distr-update-rule-guide}, we have that $\tilde P_k(d\zeta)>0$ for any $k\in\mathbb{N}$. Fix an arbitrary $\delta>0$.  We shall choose $\{N_k\}$ such that for any $k\geq k_n$ and $\epsilon>0$, the following holds  
\begin{equation}
    \label{eq:improve-A}
    \tilde P_{k+1}(A(\epsilon+\epsilon_k))\geq (1-\delta_k)\tilde P_{k}(A(\epsilon)),
\end{equation}
where 
\begin{equation}
    \label{eq:cond_delta}
    0<\delta_k<1 \quad\text{for $k\in\mathbb{N}$}, \qquad\quad\sum_{k\in\mathbb{N}} \delta_k<\infty,
\end{equation}
and $\epsilon_k\geq 0$ are determined in terms of $\iota_k$, $\alpha_k$, and the sizes of the support of density $\psi$,
\begin{equation}
\label{eq:cond_eps}
    \sum_{k=1}^\infty\epsilon_k\leq\mathrm{constant}\sum_{k=1}^\infty\iota_k<\infty .
\end{equation}
Such sequence of $\{N_k\}$ and $k_n$ exist by Lemma \ref{lem:distr-update-rule-guide}, the finiteness of $\psi$, and the condition that $\sum_{k=1}^\infty\alpha_k<\infty$. Next, select $k_o\geq k_n$ such that $$\sum_{k=k_o}^\infty\epsilon_k<\frac{1}{2}\delta, $$
and let $\delta_1=\tilde P_{k_o}(A(\delta/2)).$ Then, for any $k\geq k_o$, we have
\begin{align*}
    \tilde P_{k+1}(A(\delta))&\geq \tilde P_{k_o}(A(\delta/2+\sum_{i=k_o}^k\delta_i))\prod_{i=k_o}^k(1-\delta_i) \\
    &\geq \delta_1\prod_{i=k_o}^\infty(1-\delta_i)\\
    &>0
\end{align*}
where the last inequality is implied by \eqref{eq:cond_delta}. The proof is complete. 

\textbf{Remarks on the guidance signal.} From the proof of Corollary \ref{cor:normal}, it can be observed that we can relax the condition that $\sum_{k=1}^\infty \alpha_k<\infty$ as long as the guidance signal $\alpha_k\varrho^k$ is chosen in such a way that \eqref{eq:improve-A}, \eqref{eq:cond_delta}, and \eqref{eq:cond_eps} are satisfied. This means that we can continue applying the guidance signal without the need to diminish its impact in the long run. However, \eqref{eq:improve-A} is difficult to ensure, as it requires designing a guidance that always points to the global optimal. Therefore, in practice, it is recommended to diminish the effect of guidance and eventually let the data drive the decision.

\subsection{An example of transition probability with noiseless function evaluations}

\begin{corollary}\label{cor:better}
Under Assumption \ref{asm:main} (except for $(h)$), and further assume that $f$ can be evaluated without noise (i.e., $\xi=0$). Let the transition probability $Q_k(z,\varrho,A)$ be defined by
\begin{align}
    Q_k(z,\varrho,A)=&\int1_{\left\lbrace \zeta\in A,f(z)\leq f(\zeta)\right\rbrace} T_k(z,\varrho,d\zeta)\nonumber\\
    &+1_{\left\lbrace z\in A\right\rbrace}\int1_{\left\lbrace f(\zeta)< f(z)\right\rbrace} T_k(z,\varrho,d\zeta),
    \end{align}\label{eq:Q_better-main}
where $\{T_k(z,\varrho,d\zeta)\}$ weakly converges to $\delta_z(d\zeta)$ for all $z,\varrho\in\mathcal{Z}$. Then, there exists a sequence of natural numbers $N_k$ such that the sequence of distributions $\{\tilde P_k\}$ weakly converges to $\lambda$ for $k\rightarrow\infty$.
\end{corollary}
\begin{proof}
By Assumption \ref{asm:main} $(c)$ and $(g)$, we have that $\tilde{P}_1(\mathbb{B}^*(\epsilon))>0$ for any $\epsilon>0$. By \eqref{eq:Q_better-main}, we have that $$\tilde{P}_k(\mathbb{B}^*(\epsilon))\geq\cdots\geq \tilde{P}_1(\mathbb{B}^*(\epsilon))>0$$ for all $k\in\mathbb{N}$. Hence, Assumption \ref{asm:main} $(h)$ is satisfied. By Theorem \ref{thm:conv-general}, the claim is proved.
\end{proof}

\textbf{Remarks.} To implement the transition of \eqref{eq:Q_better-main}, one first needs to sample a variable $\zeta$ according to $T_k(z,\varrho,d\zeta)$ and observe its reward value $f(\zeta)$; then, the output is $\zeta$ if $f(\zeta)\geq f(z)$ and $z$ otherwise. Such a scheme depends crucially on a reliable way of comparing candidates (e.g., noiseless evaluation). 
\subsection{Supporting lemmas}

\begin{lemma}\label{lem:soft-max-converge-guide}
Under Assumption \ref{asm:main} $(b), (c)$, and $(d)$, the sequence of distributions 
$$\frac{\exp(kf(\zeta))\mu(d\zeta)}{\int\exp(kf(z))\mu(dz)}\Rightarrow \lambda(d\zeta),$$
i.e., weakly converges to $\lambda(d\zeta)$ for $k\rightarrow\infty$.
\end{lemma}
\begin{proof}

By the definition of weak convergence, it suffices to show that for any function $\Psi(\zeta)$ continuous on $\mathcal{Z}$, it holds that \begin{equation}
    \lim_{k\rightarrow\infty}c_k\int\exp(kf(\zeta))\Psi(\zeta)\mu(d\zeta)=\int\Psi(\zeta)\lambda(d\zeta),
    \label{eq:weak-conv}
\end{equation} 
where $c_k=1/\int\exp(kf(z))\mu(dz)$. To proceed,  Let $\mathbb{B}_i=\mathbb{B}(\epsilon_i)=\{\zeta\in\mathcal{Z}:\min_{\zeta'\in\Lambda}\|\zeta'-\zeta\|\leq \epsilon_i\}$ and $\mathbb{D}_i=\{\zeta\in\mathcal{Z}:\min_{\zeta'\in\Lambda}\|\zeta'-\zeta\|\geq \epsilon_i\}$, for $i=0,1,2$ and some $\epsilon_0,\epsilon_1,$ and $\epsilon_2$ to be determined.  For any $\delta>0$, by continuity of $\Psi$, there exists $\epsilon_0>0$ such that $|\Psi(z)-\int\Psi(\zeta)\lambda(d\zeta)|\leq \delta$ for all $z\in\mathbb{B}_0$. Choose some $\epsilon_1>0$ such that $\epsilon_1<\epsilon_0$. Then, we have
\begin{align*}
    &\left|c_k\int\exp(kf(\zeta))\Psi(\zeta)\mu(d\zeta)-\int\Psi(\zeta)\lambda(d\zeta)\right|\\
    &\leq c_k\int_{\mathbb{B}_1}\exp(kf(z))\left|\Psi(z)-\int\Psi(\zeta)\lambda(d\zeta)\right|\mu(d z)+c_k\int_{\mathbb{D}_1}\exp(kf(z))\left|\Psi(z)-\int\Psi(\zeta)\lambda(d\zeta)\right|\mu(d z)\\
    &\leq \delta\underbrace{ c_k\int_{\mathbb{B}_1}\exp(kf(z))\mu(d z)}_{(i)}+2\|\Psi\|_\infty \underbrace{c_k\int_{\mathbb{D}_1}\exp(kf(z))\mu(d z)}_{(ii)},
\end{align*}
where the first inequality is due to triangle inequality, and the second inequality is due to the choice of $\epsilon_1$ (also, recall that $\|\Psi\|_\infty=\sup|\Psi(z)|$). Hence, the lemma is proved if we can show that $(i)\rightarrow 1$ and $(ii)\rightarrow 0$ as $k\rightarrow\infty$.

To this end, let $C_1=\sup_{\zeta\in\mathbb{D}_1}f(\zeta)$. By Assumption \ref{asm:main} $(c)$, there exists $\epsilon_2$ such that $0<\epsilon_2<\epsilon_1$, and $$C_2=\inf_{\zeta\in\mathbb{B}_2}f(\zeta)>C_1.$$
For any $k>0$, we have $$\int_{\mathbb{B}_1}\exp(kf(z)-kC_1)\mu(dz)>\int_{\mathbb{B}_2}\exp(kf(z)-kC_1)\mu(dz)\geq \int_{\mathbb{B}_2}\exp(k(C_2-C_1))\mu(dz).$$
Thus, 
$$\frac{\int_{\mathbb{D}_1}\mu(dz)}{\int_{\mathbb{B}_2}\exp(k(C_2-C_1))\mu(dz)}\geq\underbrace{ \frac{\int_{\mathbb{D}_1}\exp(kf(z))\mu(dz)}{\int_{\mathbb{B}_1}\exp(kf(z))\mu(dz)}}_{(iii)}\geq 0.$$
By driving $k\rightarrow\infty$ to the limit and using the sandwich theorem, we have that $(iii)\rightarrow 0$. This immediately implies that $(i)\rightarrow 1$ and $(ii)\rightarrow 0$ as $k\rightarrow\infty$, hence concluding the proof.
\end{proof}

\begin{lemma}
\label{lem:update-distr-guide}
Let Assumption \ref{asm:main} $(a), (b)$, and $(d)$ be fulfilled. Then, the marginal distributions can be written as \begin{align}
        \tilde{P}_{k+1}(d\zeta)&=\left(\int\tilde{P}_k(dz)\exp(f(z))\right)^{-1}\int \tilde{P}_k(dz)\exp(f(z))Q_k(z,\varrho,d\zeta)M_k(z,d\varrho)+\Delta_{N_k}(d\zeta)\label{eq:update-eq-distr},
    \end{align}
    where  the signed measures $\Delta_{N_k}(d\zeta)$ converge to zero in variation for ${N_k}\rightarrow\infty$ with the rate ${N_k}^{-1/2}.$
\end{lemma}
\begin{proof}
For notational simplicity, we use $N$ for ${N_k}$ throughout the proof. By Assumption \ref{asm:main} $(a)$ and Lemma \ref{lem:distr-update-rule-guide}, the marginal distribution $\tilde{P}_{k+1}(d\zeta)$ is given by:
\begin{align*}
    \tilde{P}_{k+1}(d\zeta)&=\int_{\Omega^N}\chi_k(d\omega_N) \left\lbrace\beta(\omega_N)\sum_{i=1}^N\Lambda(\zeta_i,\varrho_i,\xi_i,d\zeta)\right\rbrace\\
    &=\sum_{i=1}^N\int_{\Omega^N}\chi_k(d\omega_N) \beta(\omega_N)\Lambda(\zeta_i,\varrho_i,\xi_i,d\zeta)\\
    &=\int_{\Omega^N}\chi_k(d\omega_N) \left\lbrace N\beta(\omega_N)\right\rbrace\Lambda(\zeta_1,\varrho_1,\xi_1,d\zeta).
\end{align*}
which can be represented in the form of \eqref{eq:update-eq-distr} with
\begin{align*}
    \Delta_N(d\zeta)&={\int_{\Omega^N}\chi_k(d\omega_N) \Lambda(\zeta_1,\varrho_1,\xi_1,d\zeta) \left\lbrace N\beta(\omega_N)-\left(\int\tilde{P}_k(dz)\exp(f(z))\right)^{-1}\right\rbrace}\\
    &\quad+ {\left(\int\tilde{P}_k(dz)\exp(f(z))\right)^{-1} \left\lbrace \int_{\Omega^N}\chi_k(d\omega_N) \Lambda(\zeta_1,\varrho_1,\xi_1,d\zeta)-\int_\Omega\tilde{P}_k(dz)\exp(f(z))Q_k(z,\varrho,d\zeta)M_k(z,d\varrho)\right\rbrace}\\
    &=(i)+(ii)
\end{align*}
We shall show that $(i)\rightarrow 0$ in variation for $N\rightarrow\infty$ and $(ii)= 0$. 
Due to Assumption \ref{asm:main} $(d)$, the convergence of $(i)$ is equivalent to the fact that $\int|v_N(\zeta)|\mu(d\zeta)\rightarrow 0$, where
\begin{equation*}
    v_N(z)={\int_{\Omega^N}\chi_k(d\omega_N) \exp(f(\zeta_1)+\xi_1)q_k(\zeta_1,\varrho_1,z) \left\lbrace N\beta(\omega_N)-\left(\int\tilde{P}_k(dz)\exp(f(z))\right)^{-1}\right\rbrace}.
\end{equation*}
To proceed, let $\gamma_N=\frac{1}{N}\sum_{i=1}^N\exp(f(\zeta_i)+\xi_i)$ and $\psi(z)=\exp(f(\zeta_1)+\xi_1)q_k(\zeta_1,\varrho_1,z)$. Due to the symmetrical dependence of random elements $\zeta_1,...,\zeta_N$ and $\varrho_1,...,\varrho_N$, as well as the independence of $\xi_1,...,\xi_N$, the random variables $\gamma_N$ converge in mean for $N\rightarrow \infty$ to some random variable $\gamma$ in dependent of all $\gamma_i(\omega_i)$, $y_i=f(\zeta_i)+\xi_i$, for $i\in\mathbb{N}$, and $$\mathbb{E}\gamma=\mathbb{E}\exp(y_i)=\int\exp(f(\zeta)+\xi)\tilde{P}_{k}(d\zeta)F_{k}(d\xi).$$ Equivalently, for any $\delta_1>0$, there exists $N_\gamma(\delta_1)\geq 1$ such that $\mathbb{E}|\gamma_N-\gamma|<\delta_1$ for all $N\geq N_\gamma(\delta_1)$.  Then,
\begin{align}
    |v_N(z)|&=\left|\mathbb{E}\left(\frac{\psi(z)}{\gamma_N}\right)-\frac{\mathbb{E}\psi(z)}{\mathbb{E} \gamma}\right|\\
    &=\frac{1}{\mathbb{E}\gamma}\left|\mathbb{E}\left(\frac{\psi(z)\gamma}{\gamma_N}\right)-{\mathbb{E}\psi(z)}\right|\\
    &\leq \exp(c_f)\left|\mathbb{E}\left(\frac{\psi(z)|\gamma-\gamma_N|}{\gamma_N}\right)\right|\\
    &\leq \exp(2c_f)\|\psi\|_\infty\mathbb{E}|\gamma-\gamma_N|\\
    &\leq L_k \exp(3c_f+c_\xi)\mathbb{E}|\gamma-\gamma_N|\label{eq:var_bound_vN},
\end{align}
where the second equality is due to the independence of $\gamma$ from $\gamma_N$ and $\psi$, the first and second inequalities are due to $\gamma,\gamma_N\geq \exp(-c_f)$ (by Assumption \ref{asm:main} $(b)$), and the last relation is due to $\|\psi\|_\infty\leq \exp(f(\zeta)+\xi))L_k\leq L_k\exp(c_f+c_\xi)$. In order to show that $\int|v_N(z)|\mu(dz)\rightarrow 0$, we need to prove that for any $\delta>0$ and $z\in\mathcal{Z}$, there exists $N^\star(\delta,z)$ such that for $N\geq N^\star(\delta,z)$, there holds $|v_N(z)|\leq \delta$. This can hold if one takes $\delta_1=\delta L_k^{-1}\exp(-3c_f-c_\xi)$ and $N^\star(\delta,z)=N_\gamma(\delta_1)$.

Now, by \eqref{eq:var_bound_vN}, we have that $\int|v_N(\zeta)|\mu(d\zeta)\leq L_k \exp(3c_f+c_\xi)\mathbb{E}|\gamma-\gamma_N|$. From the central limit theorem for symmetrically dependent random variables (see \cite{blum1958central}), it follows that $\mathbb{E}|\gamma-\gamma_N|=\mathcal{O}(N^{-1/2})$. Consequently, we have shown that $\int|v_N(\zeta)|\mu(d\zeta)=\mathcal{O}(N^{-1/2})$.

To show that $(ii)= 0$, note that
\begin{align*}
    &\int_{\Omega^N}\chi_k(d\omega_N) \Lambda(\zeta_1,\varrho_1,\xi_1,d\zeta)-\int_\Omega\tilde{P}_k(dz)\exp(f(z))Q_k(z,\varrho,d\zeta)M_k(z,d\varrho)\\
    &=\int_{\mathcal{Z}}\tilde{P}_k(dz) \exp(f(z))Q_k(z,\varrho,d\zeta)M_k(z,d\varrho)
    \left\lbrace\int \exp\left(\xi\right) F_k(d\xi)-1\right\rbrace,
\end{align*}
which is $0$ by Assumption \ref{asm:main} $(a)$. Hence, we have concluded the proof.
\end{proof}

\section{Additional details for the CityLearn Challenge}

\subsection{Details of optimization model}
We refer the reader to \cite{vazquez2020citylearn} and the corresponding online documentation\footnote{link: \url{https://sites.google.com/view/citylearnchallenge}} for the detailed setup of the contest. We will focus only on our strategy in this document. In particular, we provide details on the construction of the optimization model in \ref{eq:sol-map}. Denote the hourly index by $r \in \{1,2,\cdots ,T \}$, where $T = 24$. Suppose we are at the beginning of the hour $r$. Then we need to plan for the actions for the future hours until the end of the day and execute the plan for the next hour $r$, a.k.a., rolling-horizon planning. Next, we describe hyperparameters, variables, objective, and constraints in \ref{eq:sol-map}.
 
\textbf{Hyperparameters.} Hyperparameters are required to instantiate an optimization and are not part of the optimization variables to be solved by an optimization algorithm.
\begin{itemize}[leftmargin=2em]
    \item The hyperparameters to be set by prior knowledge include:  \emph{(1)} electric heater: efficiency $\eta_{\text{ehH}}$, nominal power $E_{\text{max}}^{\text{ehH}}$;  \emph{(2)} heat pump: technical efficiency $\eta^{\text{hp}}_{\text{tech}}$,  target cooling temperature $t_c^{\text{hp}}$, nominal power $E_{\text{max}}^{\text{hpc}}$; \emph{(3)} electric battery: rate of decay $Cf^{\text{bat}}$, capacity $Cp^{\text{bat}}$, efficiency $\eta_t^{\text{bat}}$; \emph{(4)} heat storage: rate of decay $Cf^{\text{Hsto}}$,  capacity $Cp^{\text{Hsto}}$, efficiency $\eta_t^{\text{Hsto}}$; \emph{(5)} cooling storage: rate of decay $Cf^{\text{Csto}}$,  capacity $Cp^{\text{Csto}}$, efficiency $\eta_t^{\text{Csto}}$. 
    \item The hyperparameters provided by the predictors include: \emph{(1)} hourly coefficient of performance (COP) of heat pump $\text{COP}^C_t = \eta_{\text{tech}}^{\text{hp}} \frac{t_c^{\text{hp}}+273.15}{\text{temp}_t - t_c^{\text{hp}}}$ , where $\text{temp}_t$ is the predicted outside temperature for hour $t$; \emph{(2)} solar generation $E^{PV}_t$; \emph{(3)} electricity non-shiftable load $E^{NS}_t$; \emph{(4)} heating demand $H^{bd}_t$; and \emph{(5)} cooling demand $C^{bd}_t$. At hour $r$, the above predictions are required for hour $r\leq t\leq T$. In our algorithm, predictions are provided by simply averaging the last 2 weeks of data in the corresponding hour.
    \item The hyperparameters to be learned by Algorithm \ref{alg} are the virtual electricity price $\{\theta_t\}_{t=1,...,24}$ for 24 hours. These values are bounded between $[0,10]$.
\end{itemize}

\textbf{Optimization variables.} The  variables for the  optimization at hour $r$ include:
\begin{enumerate}[leftmargin=2em]
    \item Net electricity grid import: $E_t^{\text{grid}}$, $T \geq t \geq r$
    \item Heat pump electricity usage: $E_t^{\text{hpC}}$,  $T \geq t \geq r$
    \item Electric heater electricity usage: $E_t^{\text{ehH}}$,  $T \geq t \geq r$
    \item Electric battery state of charge: $\text{SOC}_t^{\text{bat}}$, $T \geq t \geq r$
    \item Heat storage state of charge: $\text{SOC}_t^{\text{H}}$,  $T \geq t \geq r$
    \item Cooling storage state of charge: $\text{SOC}_t^{\text{C}}$,  $T \geq t \geq r$
    \item Electrical storage action: $a_t^{\text{bat}}$,  $T \geq t \geq r$
    \item Heat storage action: $a_t^{\text{Hsto}}$,  $T \geq t \geq r$
    \item Cooling storage action: $a_t^{\text{Csto}}$,  $T \geq t \geq r$
\end{enumerate}
The actions of the policy at hour $r$ are $a_r^{\text{bat}}$, $a_r^{\text{Hsto}}$, and $a_r^{\text{Csto}}$. The remaining variables are considered auxiliary variables for planning.

\textbf{Objective function.} {The objective function is given by:}
\begin{equation}
    |E^{\text{grid}}_t - E^{\text{grid}}_{t-1}| + \theta_t E^{\text{grid}}_t + \sum_{t' = t+1}^T \Big(|E^{\text{grid}}_{t'} - E^{\text{grid}}_{t'-1}| + \theta_{t'}E^{\text{grid}}_{t'} \Big).
\end{equation}
Note that we use $e_t$ for $E^{\text{grid}}_t$ in the main text. Also, the above objective is used in a standard minimization problem; to make it consistent with the maximization problem in \eqref{eq:sol-map}, we can take the negation of the value.

\textbf{Constraints.} {The constraints include both energy balance constraints and technology constraints.}

\emph{Energy balance constraints:}

\begin{itemize}[leftmargin=2em]
    \item Electricity balance for each hour $t \geq r$:\\ $E_t^{\text{PV}}+ E_t^{\text{grid}} = E_t^{\text{NS}} + E_t^{\text{hpC}}+ E_t^{\text{ehH}} + a_t^{\text{bat}}C_p^{\text{bat}} $ 
   \item Heat balance for each hour $t \geq r$:\\
   $E_t^{\text{ehH}} = a_t^{\text{Hsto}}C_p^{\text{Hsto}} + H_t^{\text{bd}}$
   
   \item Cooling balance for each hour $t \geq r$:\\
   $E_t^{\text{hpC}}\text{COP}_t^{\text{C}} = a_t^{\text{Csto}}C_p^{\text{Csto}} + C_t^{\text{bd}}$

\end{itemize}

\emph{Heat pump technology constraints:}
\begin{itemize}[leftmargin=2em]
    \item Maximum cooling for each hour $t \geq r$:\\
    $E_t^{\text{hpC}} \leq E_{\text{max}}^{\text{hpC}}$
    \item Minimum cooling for each hour $t \geq r$:\\
    $E_t^{\text{hpC}} \geq 0$
\end{itemize}

\emph{Electric heater technology constraints:}

\begin{itemize}[leftmargin=2em]
    \item Maximum limit for each hour $t \geq r$:\\
    $E_t^{\text{ehH}} \leq E_{\text{max}}^{\text{ehH}}$
   \item Minimum limit for each hour $t \geq r$:\\
   $E_t^{\text{ehH}} \geq 0$
\end{itemize}

\emph{Electric battery technology constraints:}

\begin{itemize}[leftmargin=2em]
    \item Initial SOC: \\
    $SOC_r^{\text{bat}} = (1 - C_f^{\text{bat}}SOC_{r-1}^{\text{bat}}) + a_r^{\text{bat}}\eta^{\text{bat}}$
   \item SOC updates for each hour $t \geq r$:\\
   $SOC_t^{\text{bat}} = (1 - C_f^{\text{bat}})SOC_{t-1}^{\text{bat}} + a_t^{\text{bat}}\eta^{\text{bat}}$
\item Action limits for each hour $t \geq r$: \\
$-1 \leq a_t^{\text{bat}} \leq 1$
\item Bounds of SOC or each hour $t \geq r$:\\
$0\leq SOC_t^{\text{bat}} \leq 1$
\end{itemize}

\emph{Heat storage technology constraints:}
\begin{itemize}[leftmargin=2em]
    \item Initial SOC: \\
    $SOC_r^{\text{H}} = (1 - C_f^{\text{Hsto}}SOC_{r-1}^{\text{H}}) + a_r^{\text{Hsto}}\eta^{\text{Hsto}}$
   \item SOC updates for each hour $t \geq r$:\\
   $SOC_t^{\text{H}} = (1 - C_f^{\text{Hsto}})SOC_{t-1}^{\text{H}} + a_t^{\text{Hsto}}\eta^{\text{Hsto}}$
\item Action limits or each hour $t \geq r$: \\
$-1\leq a_t^{\text{Hsto}} \leq 1$
\item Bounds of SOC or each hour $t \geq r$:\\
$0\leq SOC_t^{\text{H}} \leq 1$
\end{itemize}

\emph{Cooling storage technology constraints:}
\begin{itemize}[leftmargin=2em]
    \item Initial SOC: \\
    $SOC_r^{\text{C}} = (1 - C_f^{\text{Csto}}SOC_{r-1}^{\text{C}}) + a_r^{\text{Csto}}\eta^{\text{Csto}}$
   \item SOC updates for each hour $t \geq r$:\\
   $SOC_t^{\text{C}} = (1 - C_f^{\text{Csto}})SOC_{t-1}^{\text{C}} + a_t^{\text{Csto}}\eta^{\text{Csto}}$
  \item Action limits or each hour $t \geq r$: \\
$-1 \leq a_t^{\text{Csto}} \leq 1$
\item Bounds of SOC or each hour $t \geq r$:\\
$0\leq SOC_t^{\text{C}} \leq 1$
\end{itemize}
The above optimization can be formulated as a linear program and solved efficiently. For more implementation details, please refer to our code (submitted as supplementary materials).

\section{Additional experimental results}

\subsection{Official results for the 2021 CityLearn Challenge}

\begin{table}[h]
\centering
\begin{tabular}{@{}llllll@{}}
& 
  \multicolumn{1}{c}{\textbf{ZO-iRL}(ours)} &
  \multicolumn{1}{c}{\textbf{ICD-CA}} &
  \multicolumn{1}{c}{\textbf{IDLab}} &
  \multicolumn{1}{c}{\textbf{Breakfast Club}} \\ \toprule
\textbf{Total score}      

& \hspace{0.2cm}\textbf{0.944}
& \hspace{0.1cm}1.070
& 1.070
& \hspace{0.4cm}1.130
\\
\textbf{total last year}     
& \hspace{0.2cm}\textbf{0.942}
& \hspace{0.1cm}1.052
& 1.077
& \hspace{0.4cm}1.067
 \\
\textbf{coord. score}     
& \hspace{0.2cm}\textbf{0.915}
& \hspace{0.1cm}1.107
& 1.094
& \hspace{0.4cm}1.195
 \\
 \textbf{coord. score last year}     
& \hspace{0.2cm}\textbf{0.918}
& \hspace{0.1cm}1.074
& 1.098
& \hspace{0.4cm}1.095
 \\
 \textbf{carbon emissions}     
& \hspace{0.2cm}1.003
& \hspace{0.1cm}\textbf{1.000}
& 1.028
& \hspace{0.4cm}1.003
 \\
\end{tabular}
\caption{Official results for the 2021 CityLearn Challenge \cite{nagy2021citylearn}. Here, the total score is the average of all 6 cost metrics considered in the competition. The coordination score is the average of the first 4 metrics (see the main paper for these metrics). Last year scores are calculated based on the performance of the last year within the total 4-year simulation period.}
\end{table}

\subsection{Hyperparameters of ZO-iRL and baselines}

\begin{table}[htbp]
\centering
\begin{tabular}{lc}
\textbf{Parameter}                           & \textbf{Value}        \\ \hline
\# of parameter candidates $N_k$                                      & 3                     \\
Initial variance $\iota_1$                  & 0.4                   \\
Guidance signal $\rho$& specified in the main text\\
Guidance rate $\alpha_k$                                          & 1                     \\
Duration of one episode (hours)                                   & 24                    \\
Range of virtual electricity price                                & $[0,5]$               \\
State-action trajectory buffer size (days)                        & 7
\end{tabular} \caption{Hyperparameters for ZO-iRL.}
\end{table}

\begin{table}[htbp]
\centering
\begin{tabular}{lllllll}
\multicolumn{1}{c}{\multirow{2}{*}{\textbf{Parameter}}} & \multicolumn{6}{c}{\textbf{Value}}                          \\  
\multicolumn{1}{c}{}                                    & DDPG                     & DQN  & PPO  & TD3  & A2C  & SAC  \\ \hline
Learning rate                                           & \multicolumn{1}{c}{1e-3} & 1e-4 & 3e-4 & 1e-3 & 7e-4 & 3e-4 \\
\# of epochs                                            & NA                       & NA   & 10   & NA   & 5    & NA   \\
Buffer size                                             & 1e6                      & 1e6  & NA   & 1e6  & NA   & 1e6  \\
Batch size                                              & 100                      & 32   & 64   & 100  & NA   & 256  \\
Discount factor                                         & 0.99                     & 0.99 & 0.99 & 0.99 & 0.99 & 0.99
\end{tabular} \caption{Parameter values used for RL baselines. NA means not applicable. ADAM optimizer is used for each baseline, where the policy is given by the NN architecture with a $\tanh$ activation function and two layers of 256 units each. }
\end{table}

\subsection{Results for all climate zones}

Here, we provide results for all climate zones. Note that ZO-iRL performs some random parameter exploration in the first few weeks, which results in worse performance. However, over time, performance improves due to the guided ES, as shown in Fig. \ref{fig:Learning curve3}.

\begin{table*}[h]
\centering
\scalebox{0.95}{\begin{tabular}{@{}llllllllll@{}}
\textbf{Method} & 
  
  \multicolumn{1}{c}{\textbf{SAC}} &
  \multicolumn{1}{c}{\textbf{A2C}} &
  \multicolumn{1}{c}{\textbf{DDPG}} &
  \multicolumn{1}{c}{\textbf{PPO}} &
  \multicolumn{1}{c}{\textbf{TD3}}&\multicolumn{1}{c}{\textbf{ZO-iRL}}  \\ \hline
\textbf{ramping cost}      

& 1.244 (0.196)
& 2.714 (0.186)
& 1.327 (0.181)
& 2.718 (0.079)
& 1.500 (0.228)
& 0.750 (0.010)
 \\
\textbf{1-load factor}     

& 1.162 (0.026)
& 1.265 (0.012)
& 1.192 (0.034)
& 1.273 (0.022)
& 1.279 (0.120)
& 1.028 (0.001)
 \\
\textbf{avg. daily peak}   

& 1.195 (0.047)
& 1.413 (0.024) 
& 1.218 (0.033) 
& 1.371 (0.018)
& 1.351 (0.121)
& 0.994 (0.003)
 \\
\textbf{peak demand}       

& 1.081 (0.001)
& 1.149 (0.045)
& 1.119 (0.041) 
& 1.419 (0.082)
& 1.195 (0.161)
& 0.950 (0.020)
 \\
\textbf{net electric. peak} 

& 0.992 (0.001)
& 1.011 (0.002)
& 0.993 (0.004) 
& 1.013 (0.001)
& 1.099 (0.2066)
& 1.008 (2e-4)
 \\
\textbf{carbon emissions}     
&   0.998 (0.001)
&  1.017 (0.002)     
& 0.999 (0.004)    
& 1.023 (0.002)
& 1.104 (0.203)
& 1.010 (1e-4)   \\  \hline
\textbf{total score}       
& 1.112 (0.044)
& 1.428 (0.042)
& 1.141 (0.047)
& 1.469(0.026)
& 1.255 (0.141)
& 0.957 (0.004)
\end{tabular}}
\caption{Comparison of ZO-iRL and baselines for \textbf{Climate Zone 2}. The reported values are the average and standard deviation (in brackets) across 10 independent runs.}
\end{table*}

\begin{table*}[h]
\centering
\scalebox{0.95}{\begin{tabular}{@{}llllllllll@{}}
\textbf{Method} & 
  
  \multicolumn{1}{c}{\textbf{SAC}} &
  \multicolumn{1}{c}{\textbf{A2C}} &
  \multicolumn{1}{c}{\textbf{DDPG}} &
  \multicolumn{1}{c}{\textbf{PPO}} &
  \multicolumn{1}{c}{\textbf{TD3}}&\multicolumn{1}{c}{\textbf{ZO-iRL}}  \\ \hline
\textbf{ramping cost}      

& 1.098 (0.011)
& 2.986 (0.223)
& 1.367 (0.182)
& 3.154 (0.066)
& 1.263 (0.127)
& 0.775 (0.004)
 \\
\textbf{1-load factor}     

& 1.140 (0.002)
& 1.272 (0.009)
& 1.188 (0.040)
& 1.375 (0.038)
& 1.240 (0.136)
& 1.043 (0.001)
 \\
\textbf{avg. daily peak}   

& 1.169 (0.002)
& 1.441 (0.024) 
& 1.242 (0.063) 
& 1.439 (0.025)
& 1.277 (0.177)
& 1.010 (0.003)
 \\
\textbf{peak demand}       

& 1.182 (0.001)
& 1.272 (0.043)
& 1.199 (0.022) 
& 1.472 (0.084)
& 1.259 (0.127)
& 0.952 (0.014)
 \\
\textbf{net electric. peak} 

& 0.994 (0.002)
& 1.013 (0.003)
& 0.996 (0.004) 
& 1.010 (0.001)
& 1.099 (0.206)
& 1.006 (3e-4)
 \\
\textbf{carbon emissions}     
&  1.000 (0.002)
&  1.019 (0.003)     
& 1.002 (0.004)    
& 1.020 (0.001)
& 1.104 (0.203)
& 1.007 (1e-4)   \\  \hline
\textbf{total score}       
& 1.097 (0.002)
& 1.501 (0.047)
& 1.166 (0.048)
& 1.578(0.025)
& 1.207 (0.152)
& 0.966 (0.003)
\end{tabular}}
\caption{Comparison of ZO-iRL and baselines for \textbf{Climate Zone 3}. The reported values are the average and standard deviation (in brackets) across 10 independent runs.}
\end{table*}

\begin{table*}[h]
\centering
\scalebox{0.95}{\begin{tabular}{@{}llllllllll@{}}
\textbf{Method} & 
  
  \multicolumn{1}{c}{\textbf{SAC}} &
  \multicolumn{1}{c}{\textbf{A2C}} &
  \multicolumn{1}{c}{\textbf{DDPG}} &
  \multicolumn{1}{c}{\textbf{PPO}} &
  \multicolumn{1}{c}{\textbf{TD3}}&\multicolumn{1}{c}{\textbf{ZO-iRL}}  \\ \hline
\textbf{ramping cost}      

& 1.024 (0.006)
& 2.814 (0.468)
& 1.566 (0.230)
& 3.070 (0.108)
& 1.310 (0.072)
& 0.739 (0.003)
 \\
\textbf{1-load factor}     

& 1.117 (0.003)
& 1.229 (0.024)
& 1.185 (0.036)
& 1.449 (0.019)
& 1.187 (0.038)
& 1.013 (0.006)
 \\
\textbf{avg. daily peak}   

& 1.126 (0.002)
& 1.411 (0.072) 
& 1.249 (0.070) 
& 1.429 (0.015)
& 1.262 (0.051)
& 1.003 (0.001)
 \\
\textbf{peak demand}       

& 1.134 (1e-4)
& 1.238 (0.069)
& 1.155 (0.036) 
& 1.444 (0.061)
& 1.205 (0.099)
& 0.999 (0.026)
 \\
\textbf{net electric. peak} 

& 0.987 (0.002)
& 1.010 (0.006)
& 0.995 (0.004) 
& 1.007 (0.002)
& 0.990 (0.003)
& 1.007 (4e-4)
 \\
\textbf{carbon emissions}     
&  0.994 (0.002)
&  1.017 (0.006)     
& 1.000 (0.004)    
& 1.015 (0.001)
& 0.996 (0.004)
& 1.009 (8e-4)   \\  \hline
\textbf{total score}       
& 1.064 (0.001)
& 1.453 (0.103)
& 1.192 (0.060)
& 1.569(0.032)
& 1.158 (0.042)
& 0.962 (0.003)
\end{tabular}}
\caption{Comparison of ZO-iRL and baselines for \textbf{Climate Zone 4}. The reported values are the average and standard deviation (in brackets) across 10 independent runs.}
\end{table*}

\begin{table*}[h]
\centering
\scalebox{0.95}{\begin{tabular}{@{}llllllllll@{}}
\textbf{Method} & 
  
  \multicolumn{1}{c}{\textbf{SAC}} &
  \multicolumn{1}{c}{\textbf{A2C}} &
  \multicolumn{1}{c}{\textbf{DDPG}} &
  \multicolumn{1}{c}{\textbf{PPO}} &
  \multicolumn{1}{c}{\textbf{TD3}}&\multicolumn{1}{c}{\textbf{ZO-iRL}}  \\ \hline
\textbf{ramping cost}      

& 1.245 (0.162)
& 2.603 (0.158)
& 1.320 (0.161)
& 2.673 (0.047)
& 1.385 (0.144)
& 0.789 (0.005)
 \\
\textbf{1-load factor}     

& 1.241 (0.055)
& 1.314 (0.011)
&  1.212 (0.033)
& 1.332 (0.032)
& 1.253 (0.043)
& 1.045 (0.011)
 \\
\textbf{avg. daily peak}   

& 1.172 (0.049)
& 1.347 (0.025) 
& 1.197 (0.043) 
& 1.346 (0.021)
& 1.221 (0.044)
& 1.004 (0.002)
 \\
\textbf{peak demand}       

& 1.285 (0.097)
& 1.356 (0.021)
& 1.203 (0.035) 
& 1.535 (0.184)
& 1.231 (0.066)
& 1.014 (0.022)
 \\
\textbf{net electric. peak} 

& 0.989 (0.002)
& 1.003 (0.002)
& 0.990 (0.006) 
& 1.003 (8e-4)
& 0.993 (0.003)
& 1.004 (0.001)
 \\
\textbf{carbon emissions}     
&  0.995 (0.002)
&  1.009 (0.002)     
& 0.997 (0.007)    
& 1.015 (7e-4)
& 0.999 (0.002)
& 1.005 (0.001)   \\  \hline
\textbf{total score}       
& 1.154 (0.057)
& 1.438 (0.032)
& 1.153 (0.041)
& 1.484(0.042)
& 1.180 (0.049)
& 0.977 (0.005)
\end{tabular}}
\caption{Comparison of ZO-iRL and baselines for \textbf{Climate Zone 5}. The reported values are the average and standard deviation (in brackets) across 10 independent runs.}
\end{table*}

\begin{figure}[htbp]
  \centering
  \includegraphics[width=\columnwidth]{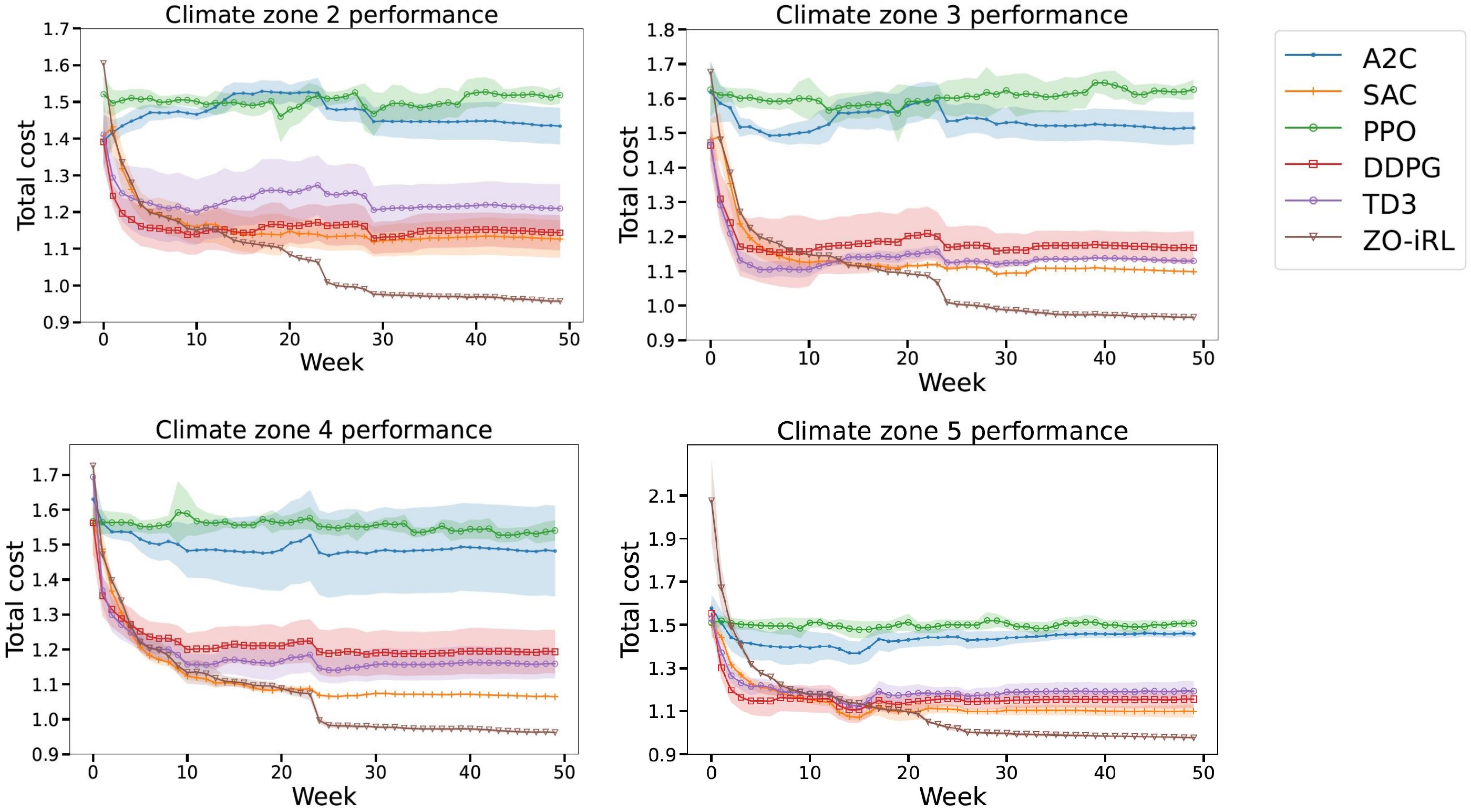}
\caption{Learning curves of ZO-iRL  and baselines for Climate Zones 2--5. Note that ZO-iRL is the only method that consistently achieves a cost below 1 across different runs in different climate zones. }
\label{fig:Learning curve3}
\end{figure}

\clearpage
\newpage

\subsection{Visualization of parameter evolution}
In this section, we visualize the evolution of parameters under ZO-iRL for some buildings. We also juxtapose the corresponding patterns of empirical peak counts, net electricity use, electricity demand, heating demand, and cooling demand. The empirical count is calculated for each hour as the number of times the corresponding hour has the top 2 net electricity usage in a week. The higher the empirical counts, the more frequent the corresponding hour has peak usage. We also note that electricity usage is higher than electricity demand due to the additional energy demand for heating and cooling.

\begin{figure}[h]
\centering
\begin{subfigure}{.42\textwidth}
  \centering
  \includegraphics[width=\columnwidth]{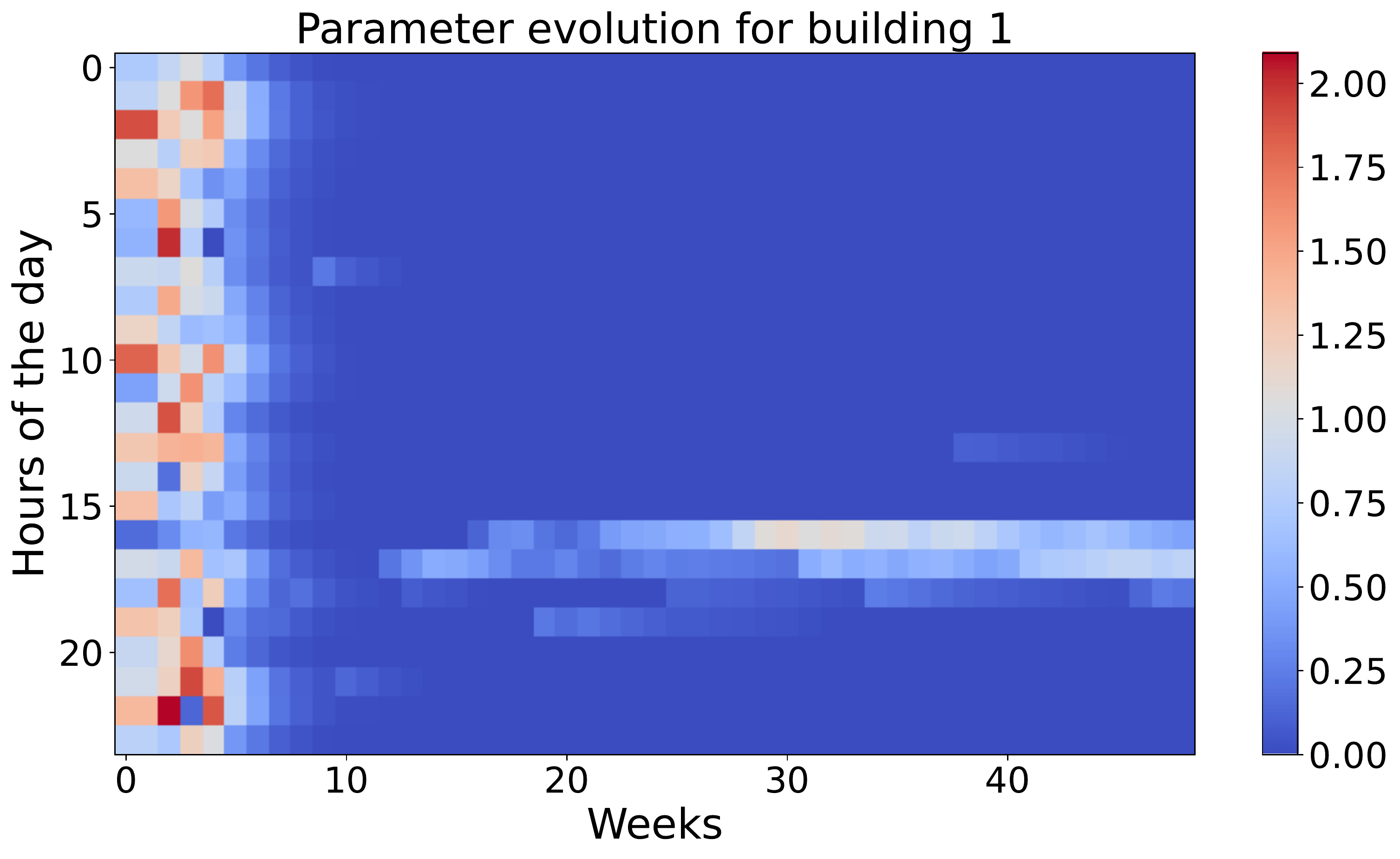}
\caption{virtual electricity prices $\theta$}
\end{subfigure}
\begin{subfigure}{.4\textwidth}
  \centering
  \includegraphics[width=\columnwidth]{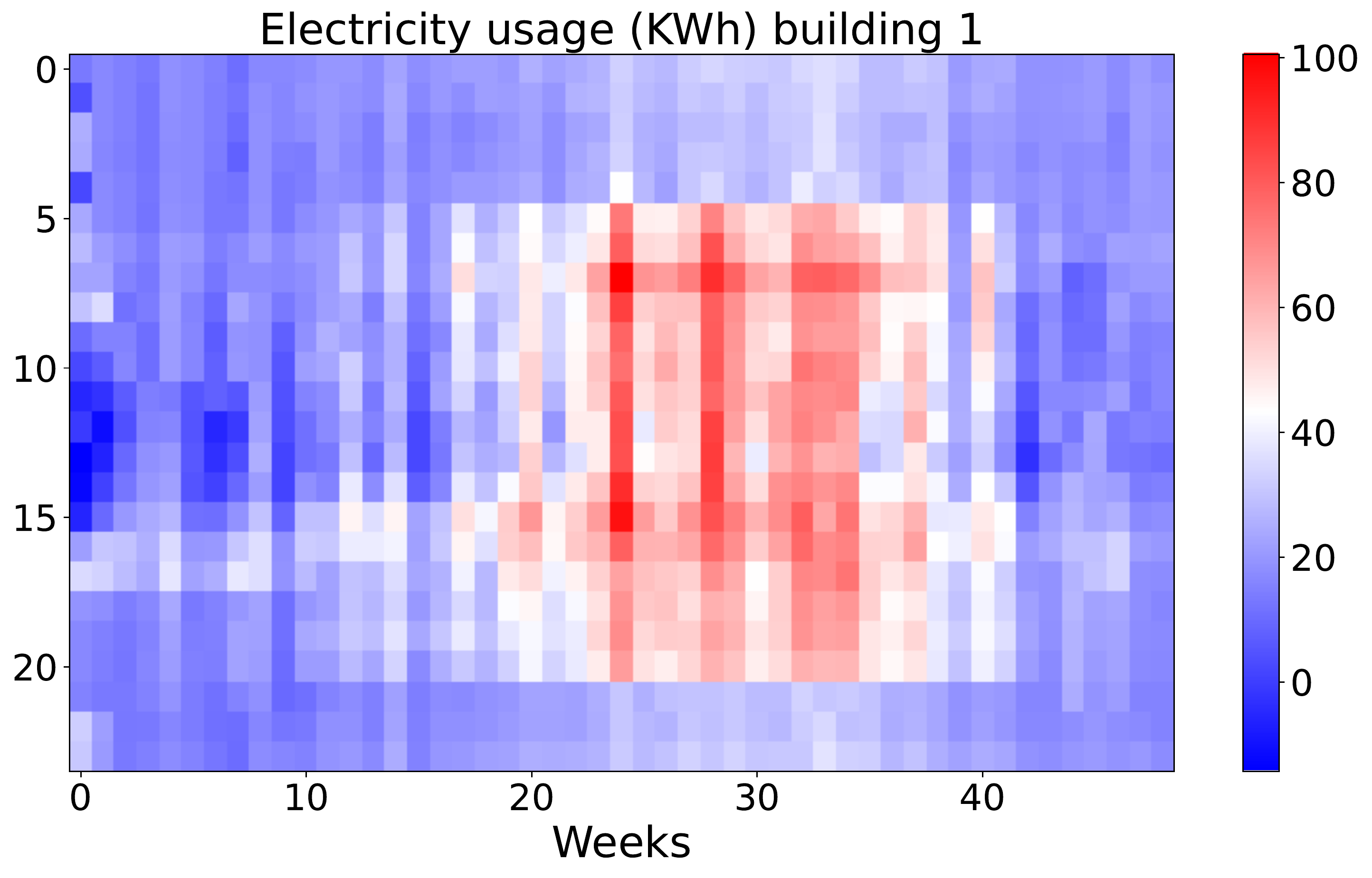}
\caption{net electricity usage}
\end{subfigure}
\begin{subfigure}{.4\textwidth}
  \centering
  \includegraphics[width=\columnwidth]{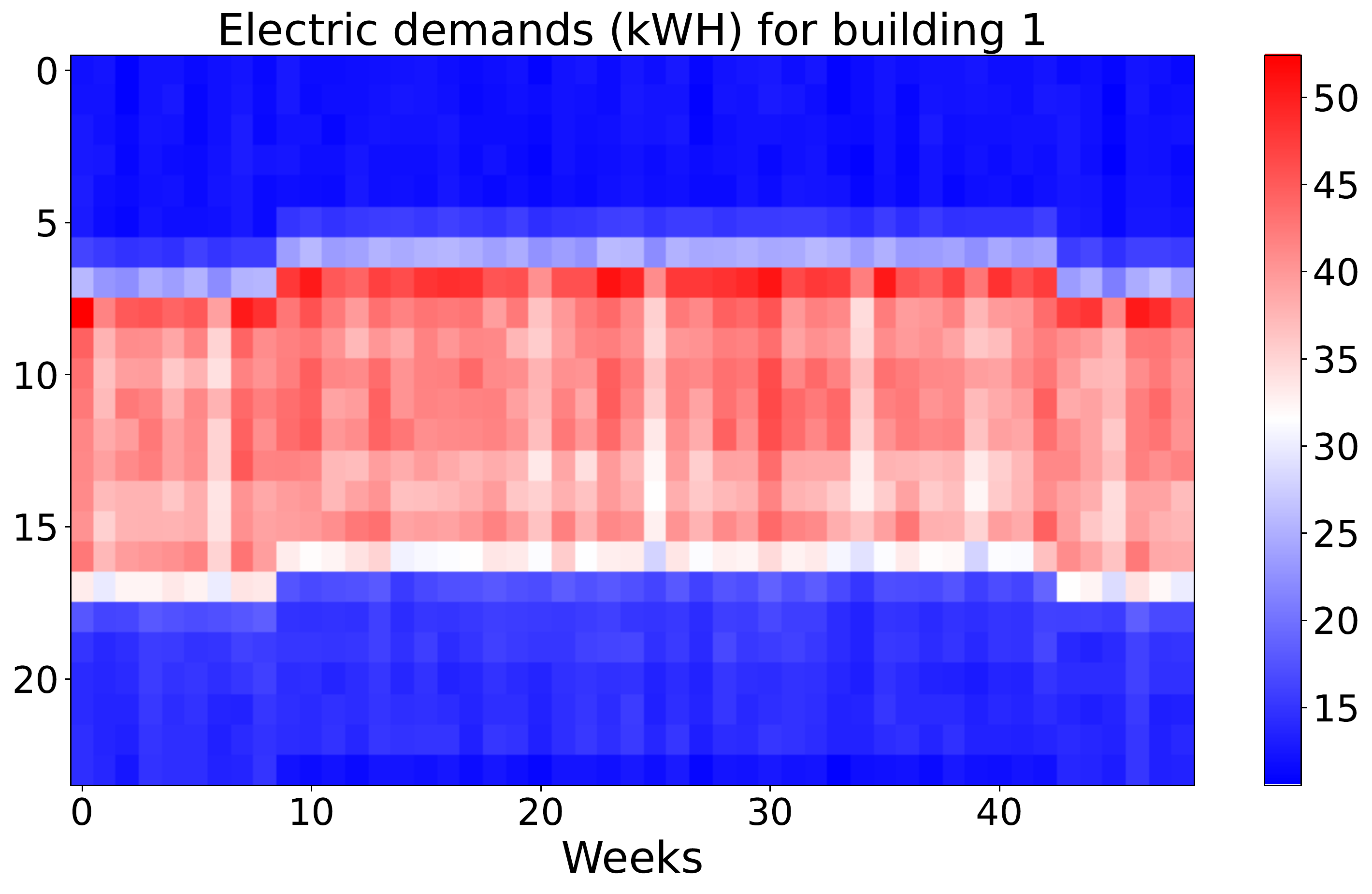}
\caption{electricity demands}
\end{subfigure}
\begin{subfigure}{.39\textwidth}
  \centering
  \includegraphics[width=\columnwidth]{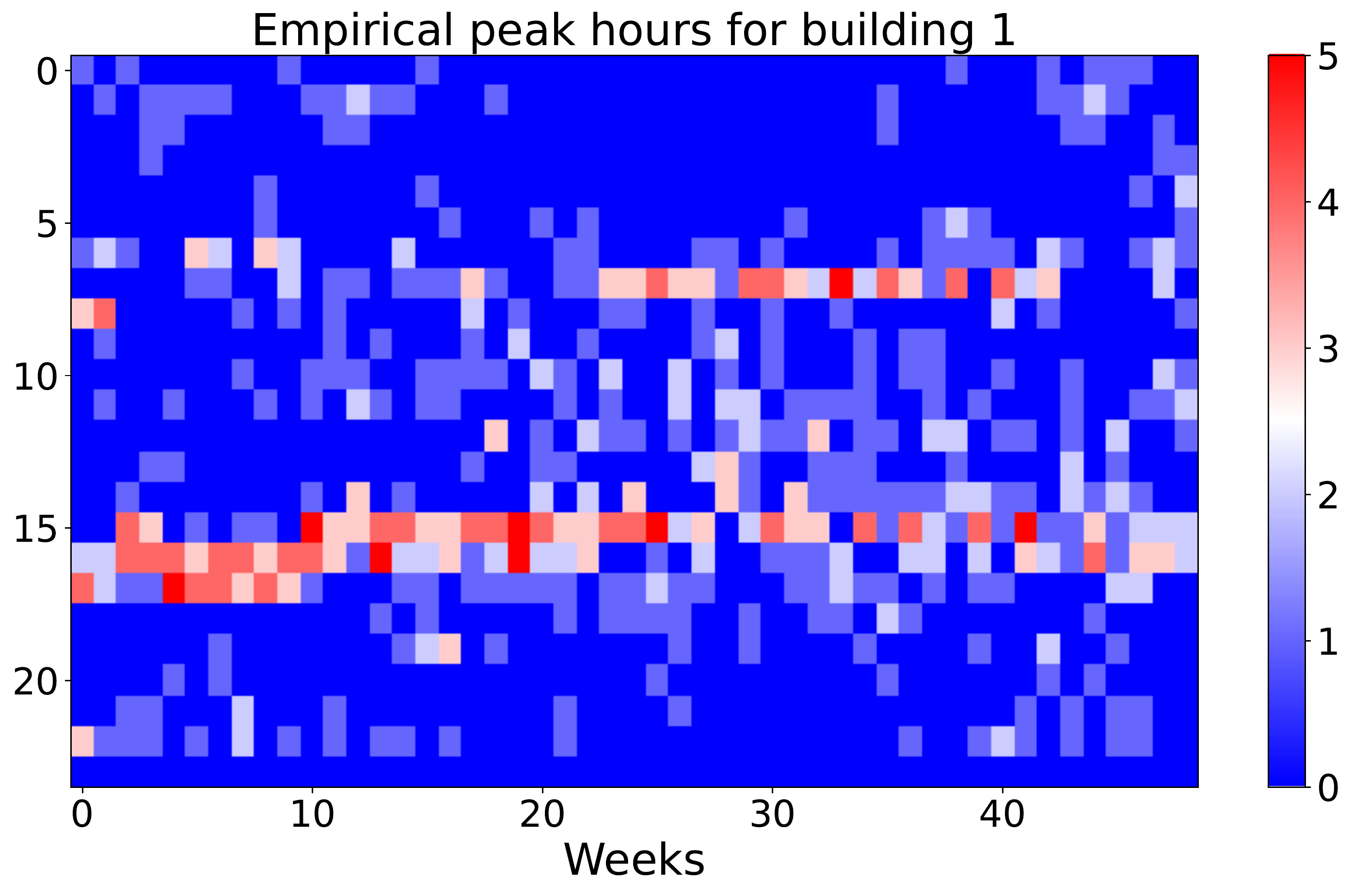}
\caption{empirical counts of peaks}
\end{subfigure}
\begin{subfigure}{.4\textwidth}
  \centering
  \includegraphics[width=\columnwidth]{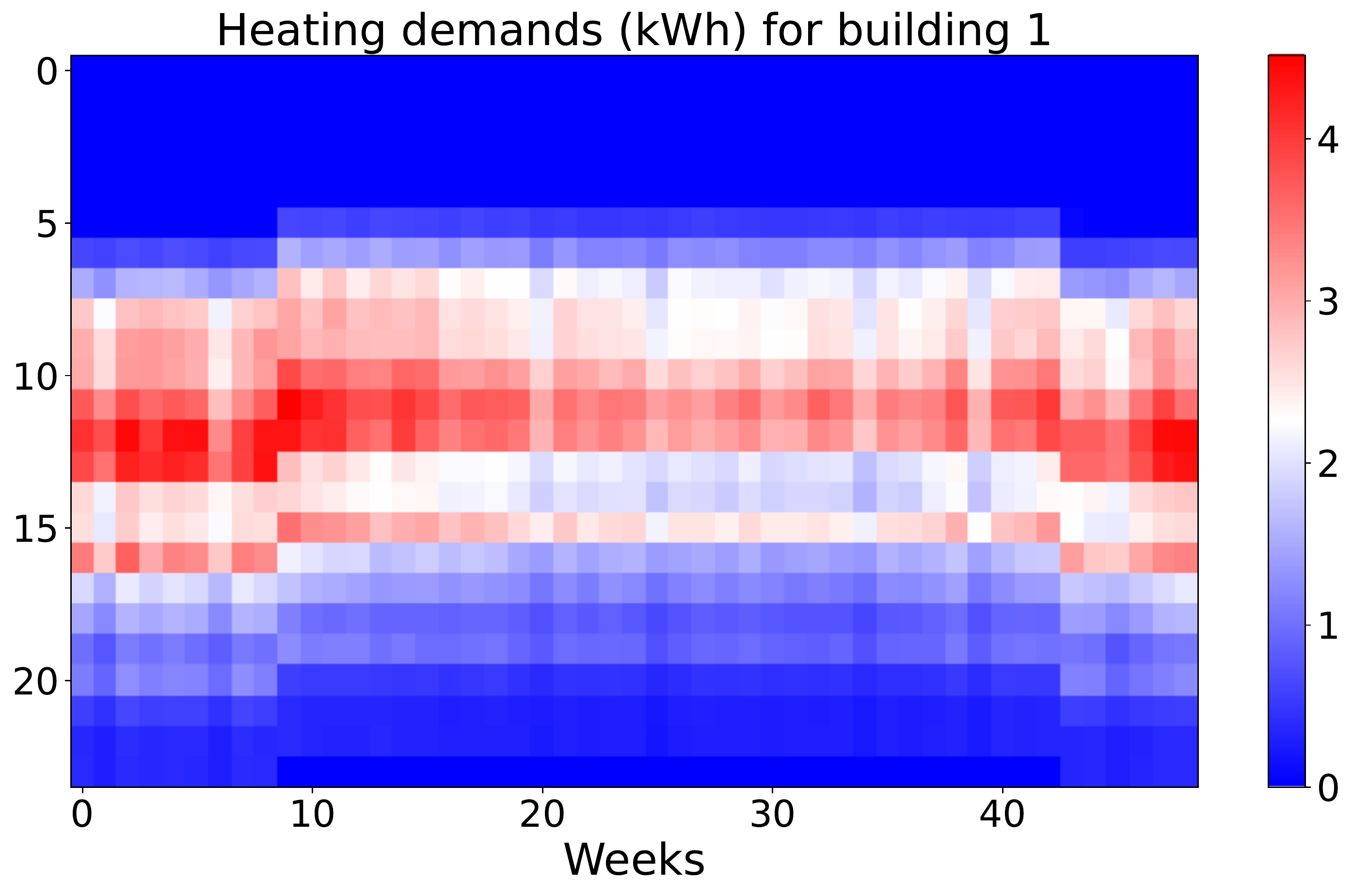}
\caption{heating demands}
\end{subfigure}
\begin{subfigure}{.4\textwidth}
  \centering
  \includegraphics[width=\columnwidth]{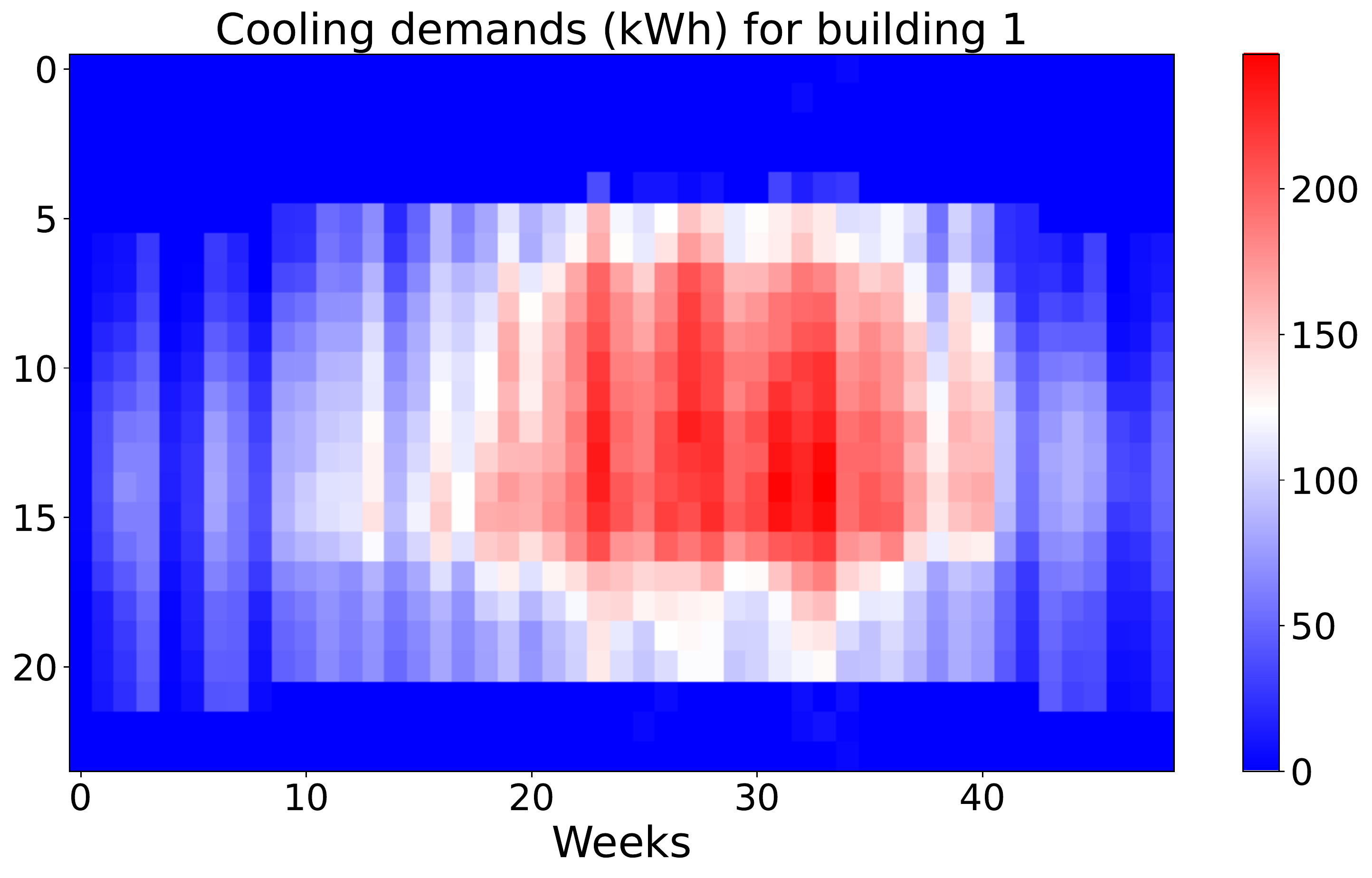}
\caption{cooling demands}
\end{subfigure}
\caption{Visualization of (a) parameter learning, (b) net electricity  use (c) electric loads, (d) empirical counts of peaks, (e) heating demand, and (d) cooling demand for Building 1. It can be observed that Building 1 continues to increase the virtual electricity prices for hours around 16--18 in response to consistently observed peaks in those hours. Due to storage controls, the net electricity usage pattern is smoother (spreading throughout the day) than the demand patterns, as observed in all buildings.}
\end{figure}

\begin{figure}[h]
\centering
\begin{subfigure}{.42\textwidth}
  \centering
  \includegraphics[width=\columnwidth]{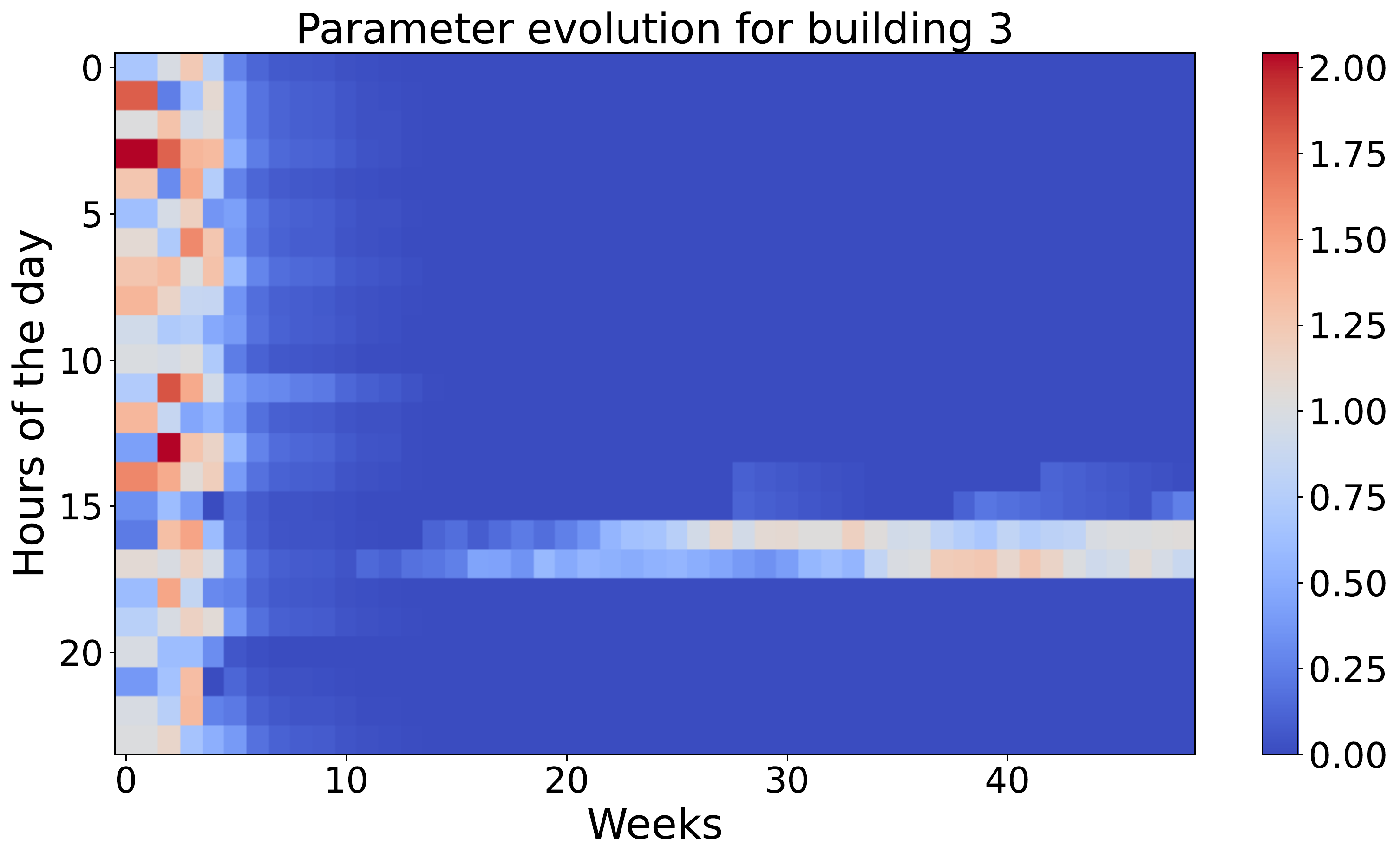}
\caption{virtual electricity prices $\theta$}
\end{subfigure}%
\begin{subfigure}{.4\textwidth}
  \centering
  \includegraphics[width=\columnwidth]{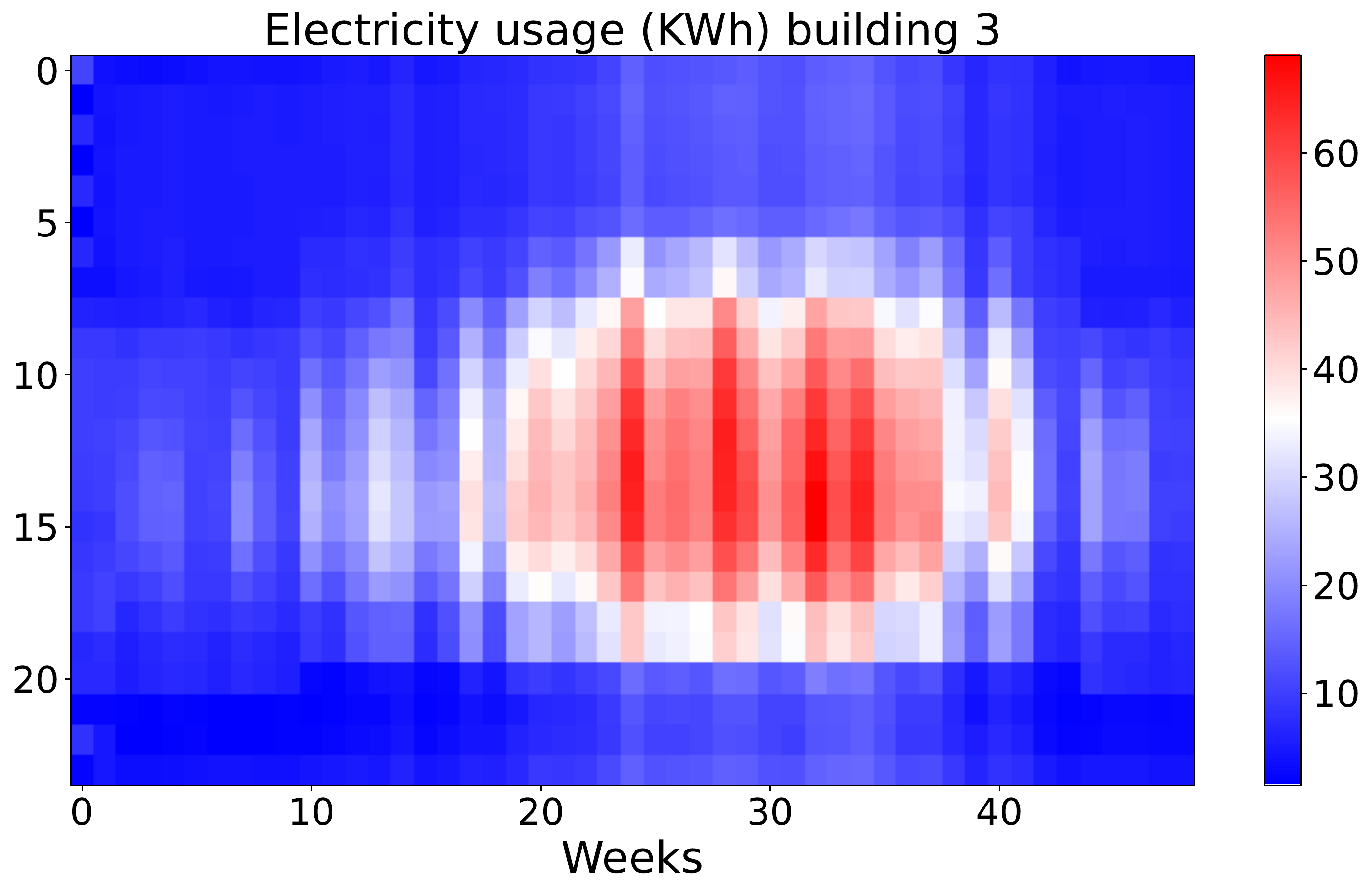}
\caption{net electricity usage}
\end{subfigure}
\begin{subfigure}{.4\textwidth}
  \centering
  \includegraphics[width=\columnwidth]{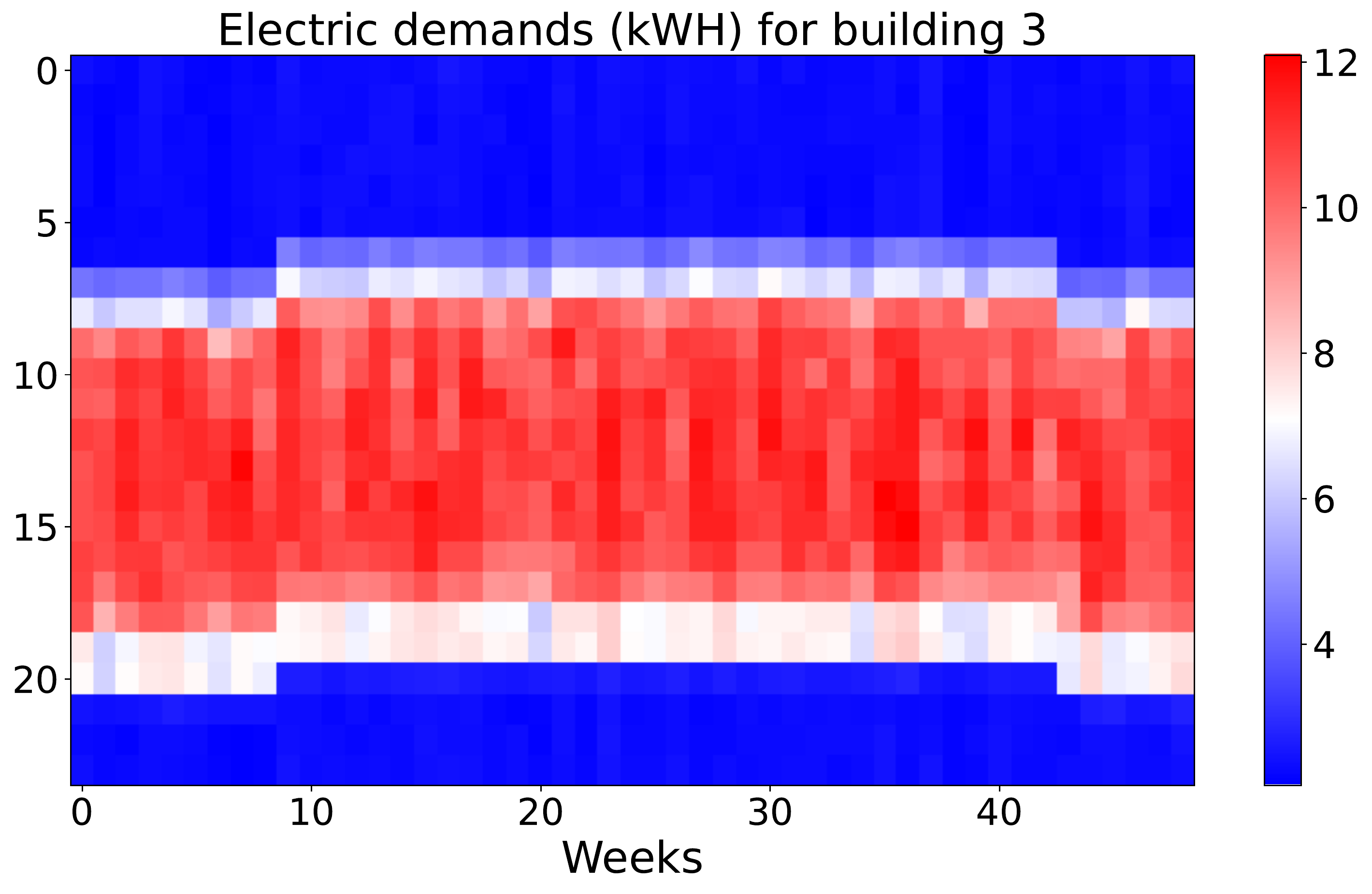}
\caption{electricity demands}
\end{subfigure}
\begin{subfigure}{.4\textwidth}
  \centering
  \includegraphics[width=\columnwidth]{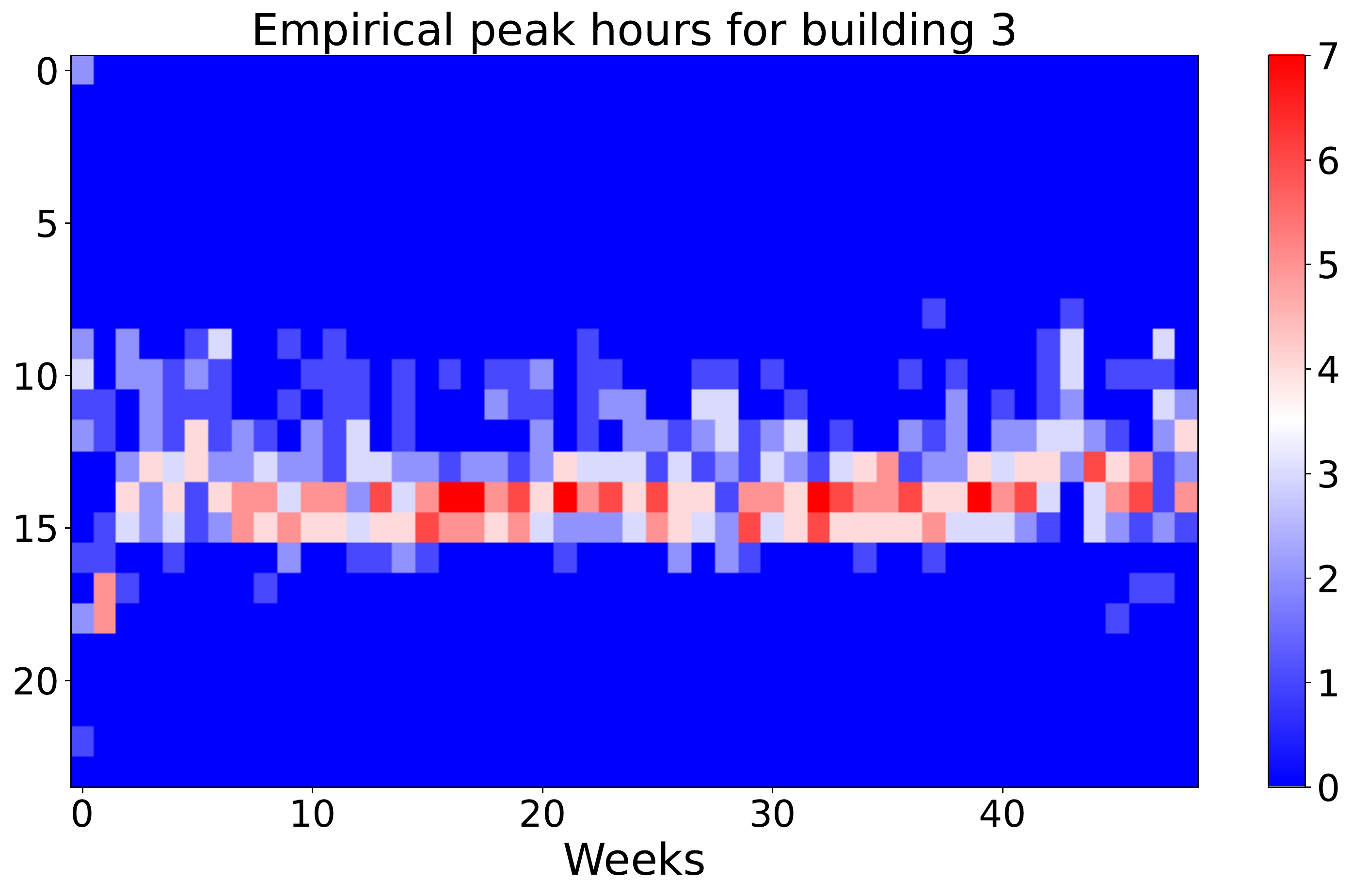}
\caption{empirical counts of peaks}
\end{subfigure}
\begin{subfigure}{.4\textwidth}
  \centering
  \includegraphics[width=\columnwidth]{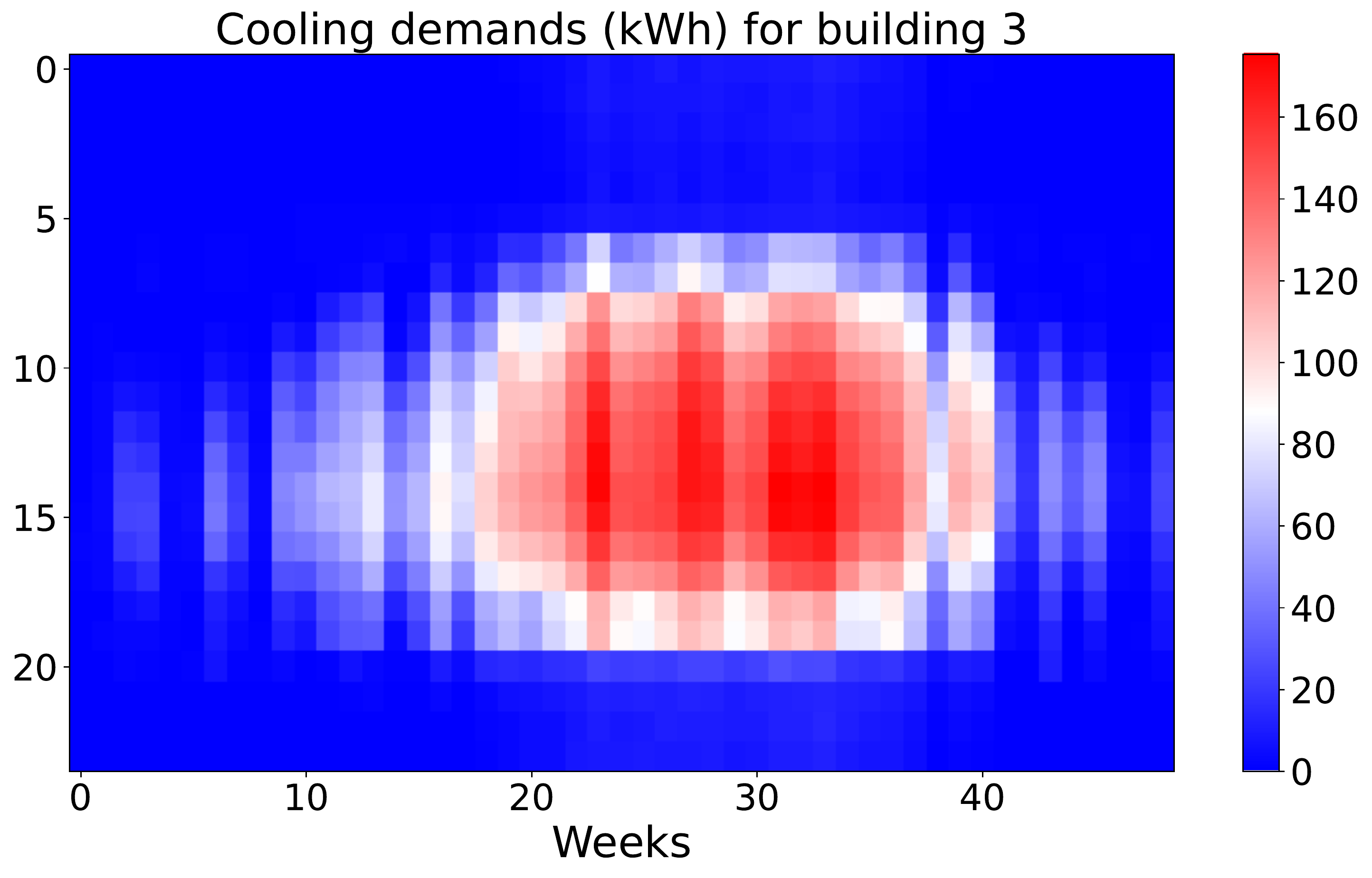}
\caption{cooling demands}
\end{subfigure}
\caption{Similar patterns can be observed for Building 3. Note that for this building, there is no heating demands.}
\end{figure}

\begin{figure}[h]
\centering
\begin{subfigure}{.42\textwidth}
  \centering
  \includegraphics[width=\columnwidth]{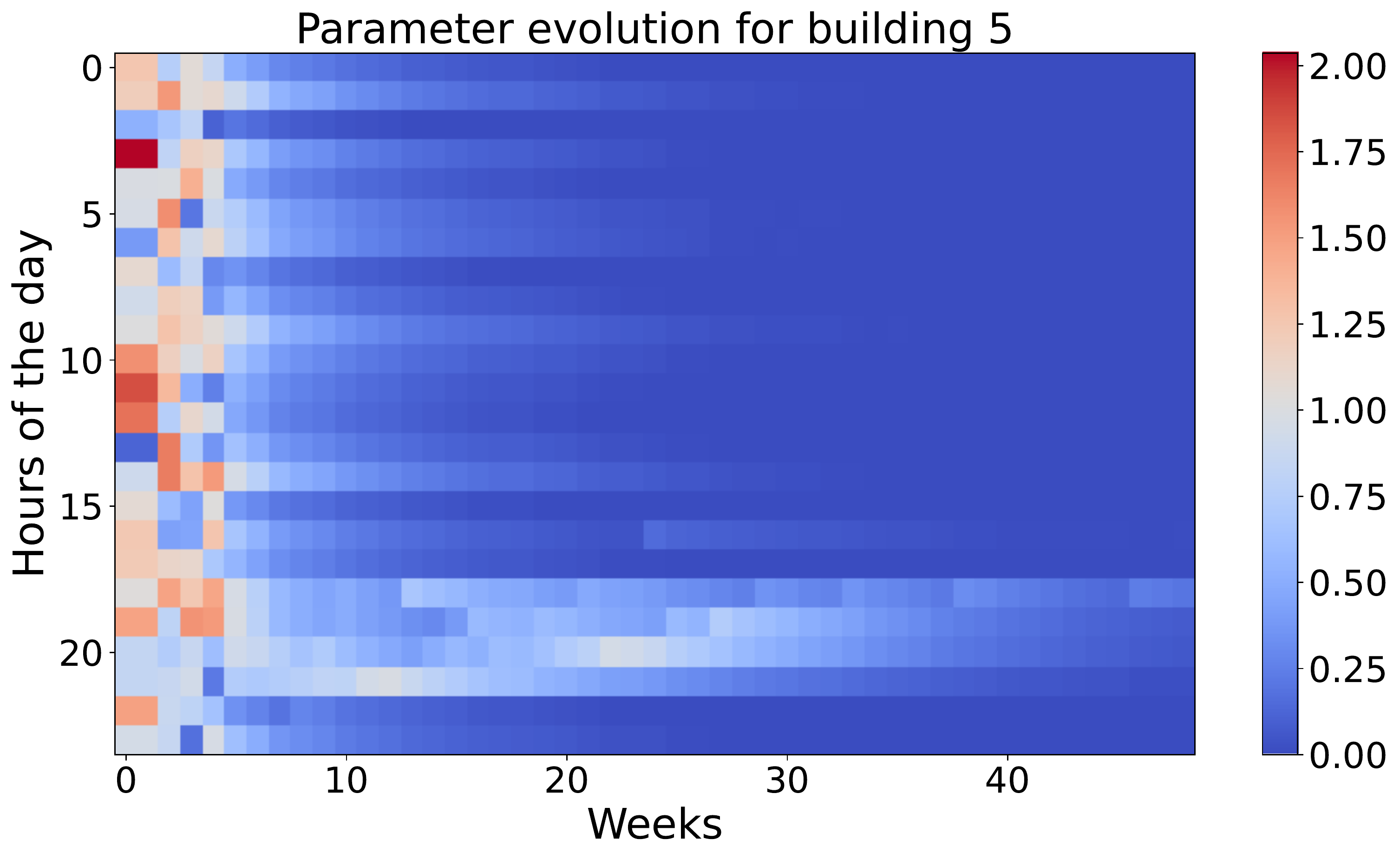}
\caption{virtual electricity prices $\theta$}
\end{subfigure}%
\begin{subfigure}{.4\textwidth}
  \centering
  \includegraphics[width=\columnwidth]{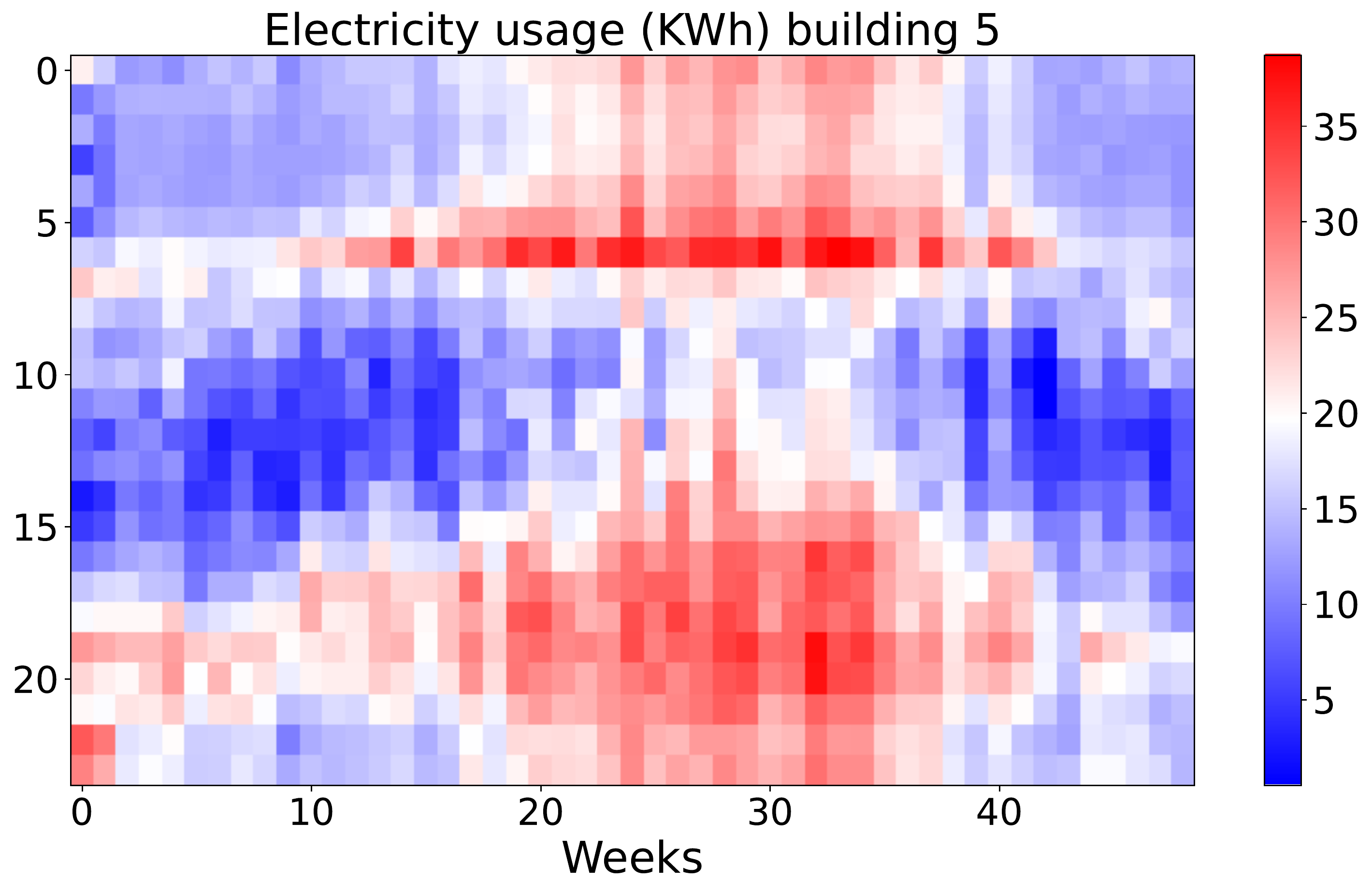}
\caption{net electricity usage}
\end{subfigure}
\begin{subfigure}{.4\textwidth}
  \centering
  \includegraphics[width=\columnwidth]{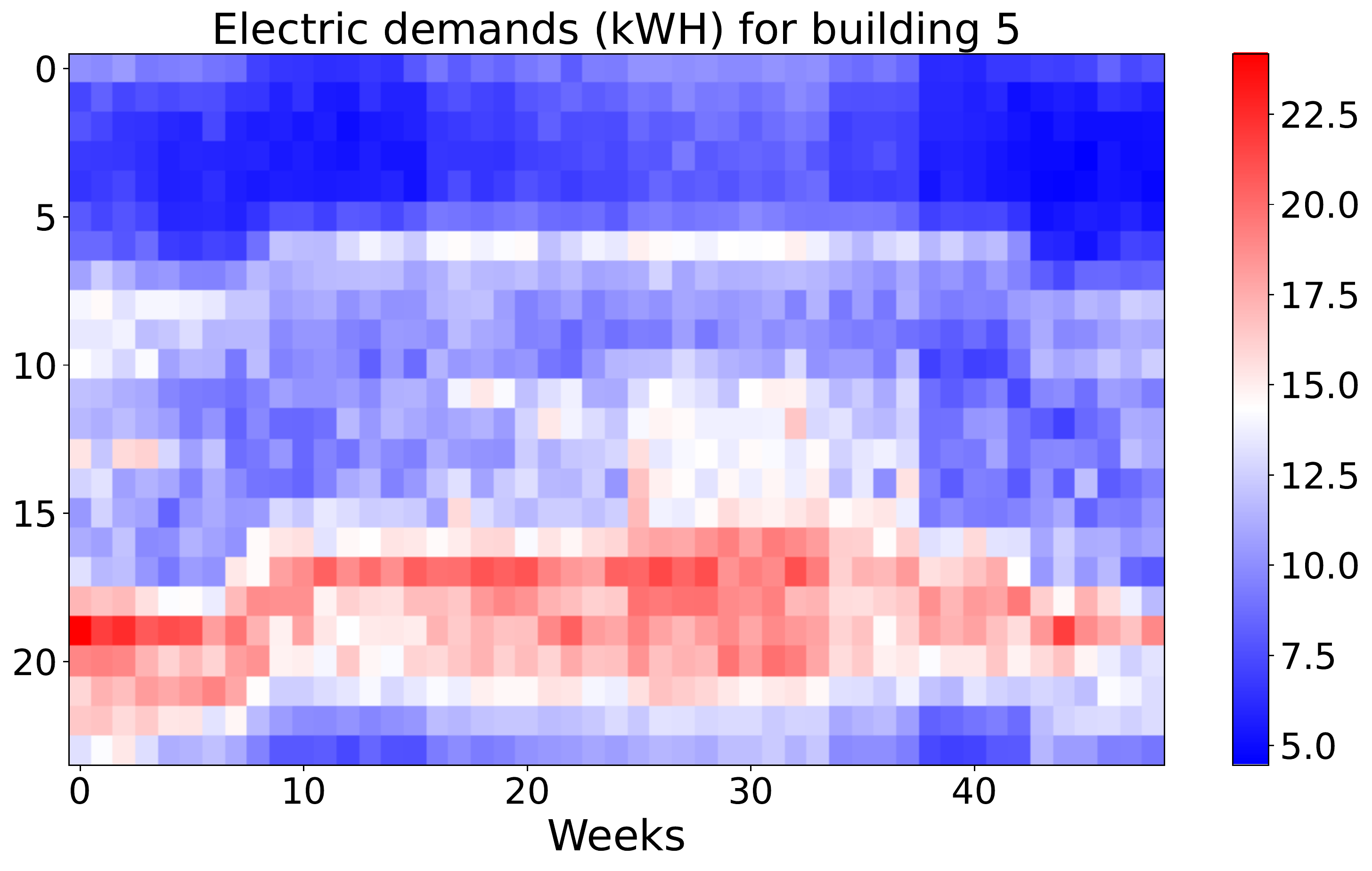}
\caption{electricity demands}
\end{subfigure}
\begin{subfigure}{.4\textwidth}
  \centering
  \includegraphics[width=\columnwidth]{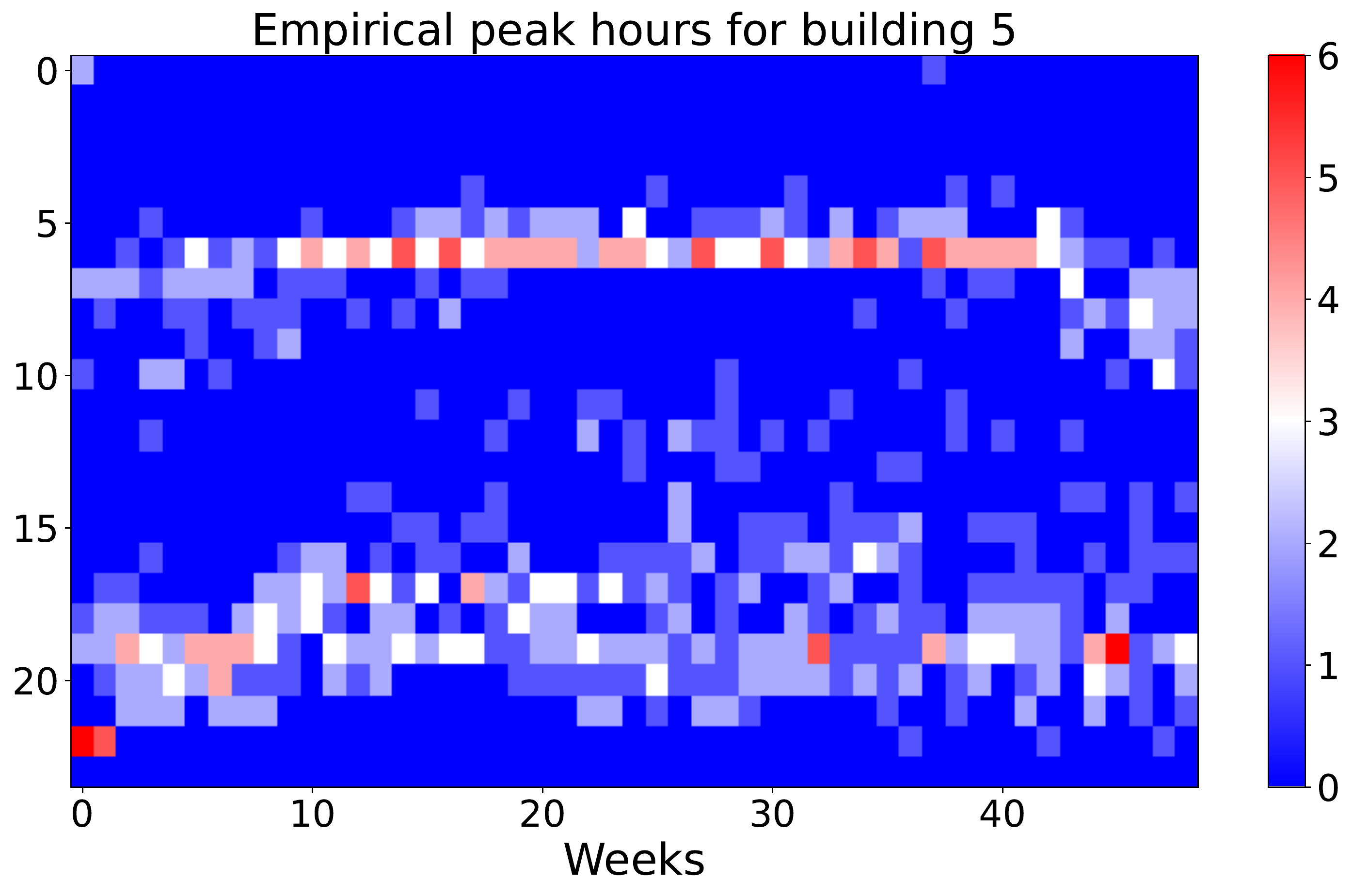}
\caption{empirical counts of peaks}
\end{subfigure}
\begin{subfigure}{.4\textwidth}
  \centering
  \includegraphics[width=\columnwidth]{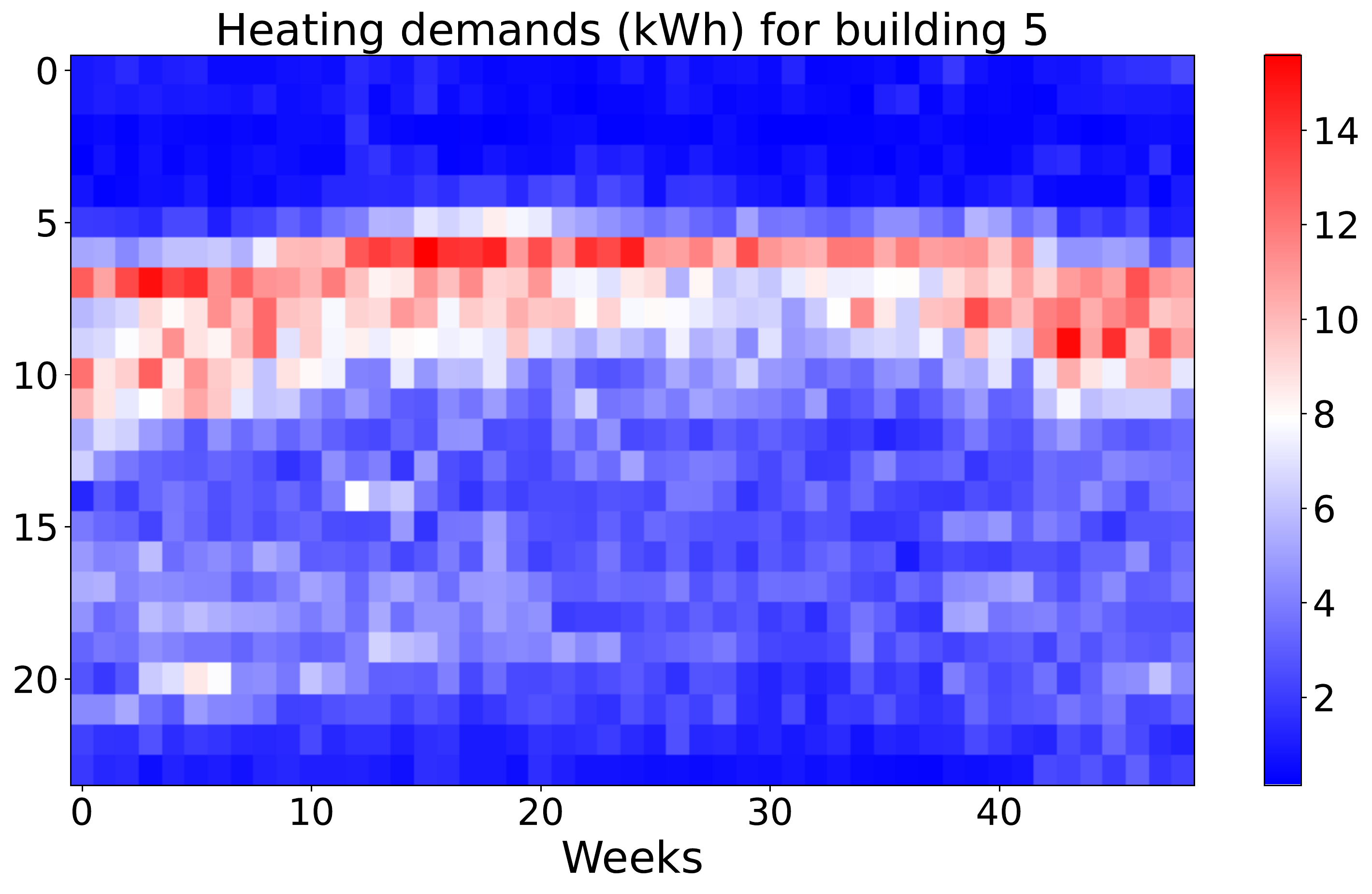}
\caption{heating demands}
\end{subfigure}
\begin{subfigure}{.4\textwidth}
  \centering
  \includegraphics[width=\columnwidth]{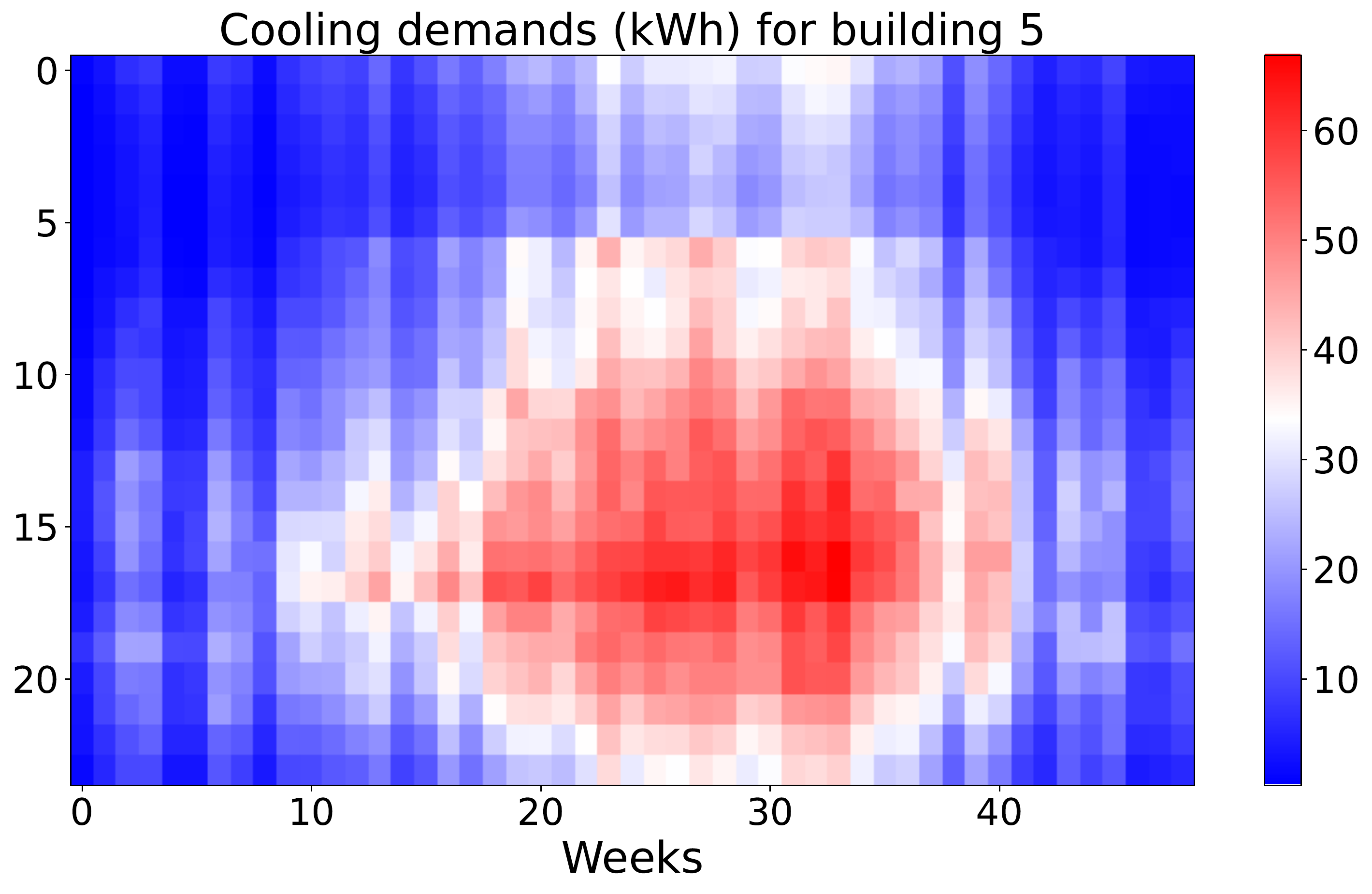}
\caption{cooling demands}
\end{subfigure}
\caption{Visualization of (a) parameter learning, (b) net electricity demand, (c) electric loads, (d) empirical counts of peaks, (e) heating demand, and (f) cooling demand for Building 8. It can be observed that Building 8 increases the virtual electricity price during hours 17--23 in response to high electricity peaks. As peak issues are mitigated, virtual electricity prices eventually decline, as can be seen after week 30.}
\end{figure}

\begin{figure}[h]
\centering
\begin{subfigure}{.42\textwidth}
  \centering
  \includegraphics[width=\columnwidth]{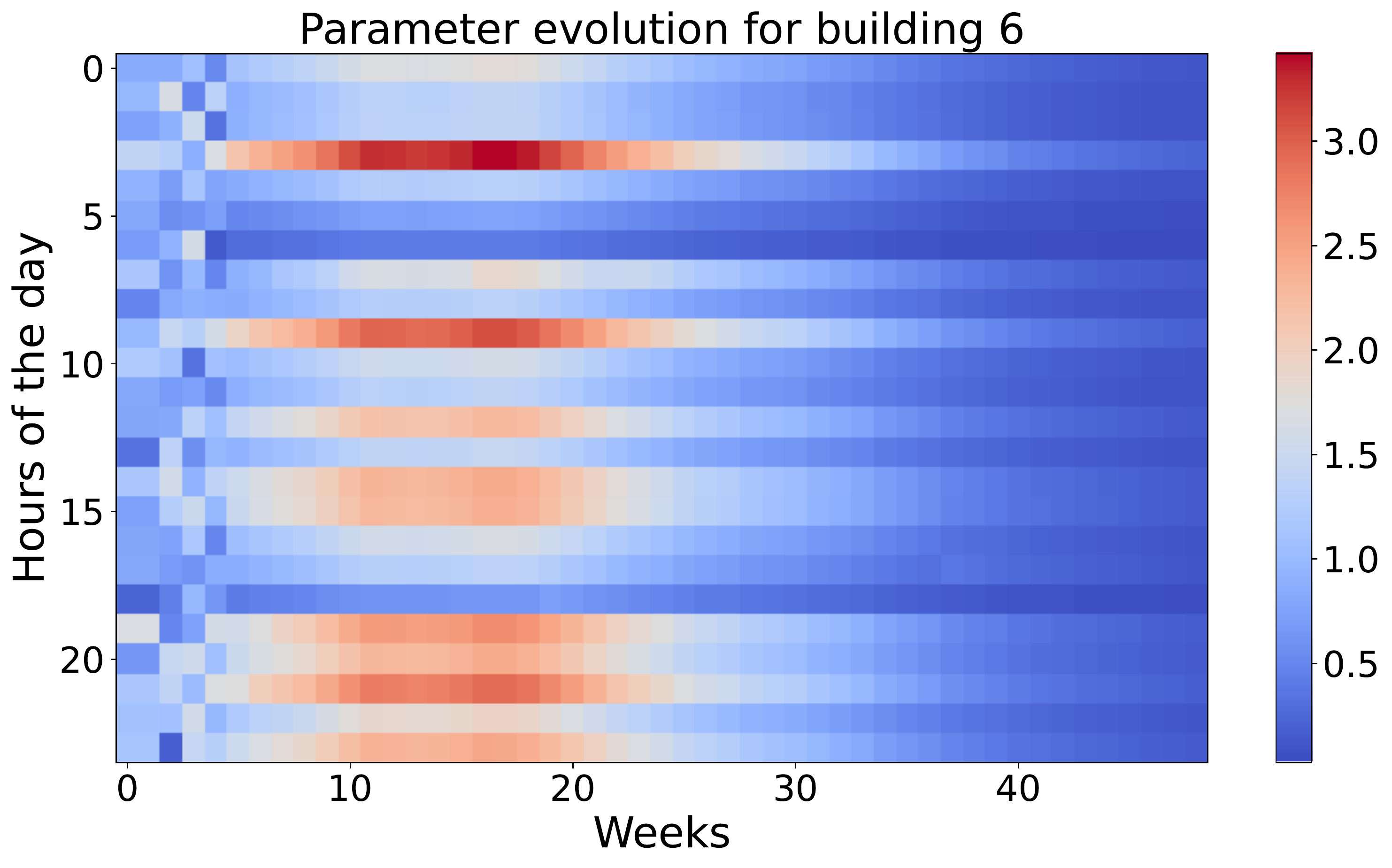}
\caption{virtual electricity prices $\theta$}
\end{subfigure}%
\begin{subfigure}{.4\textwidth}
  \centering
  \includegraphics[width=\columnwidth]{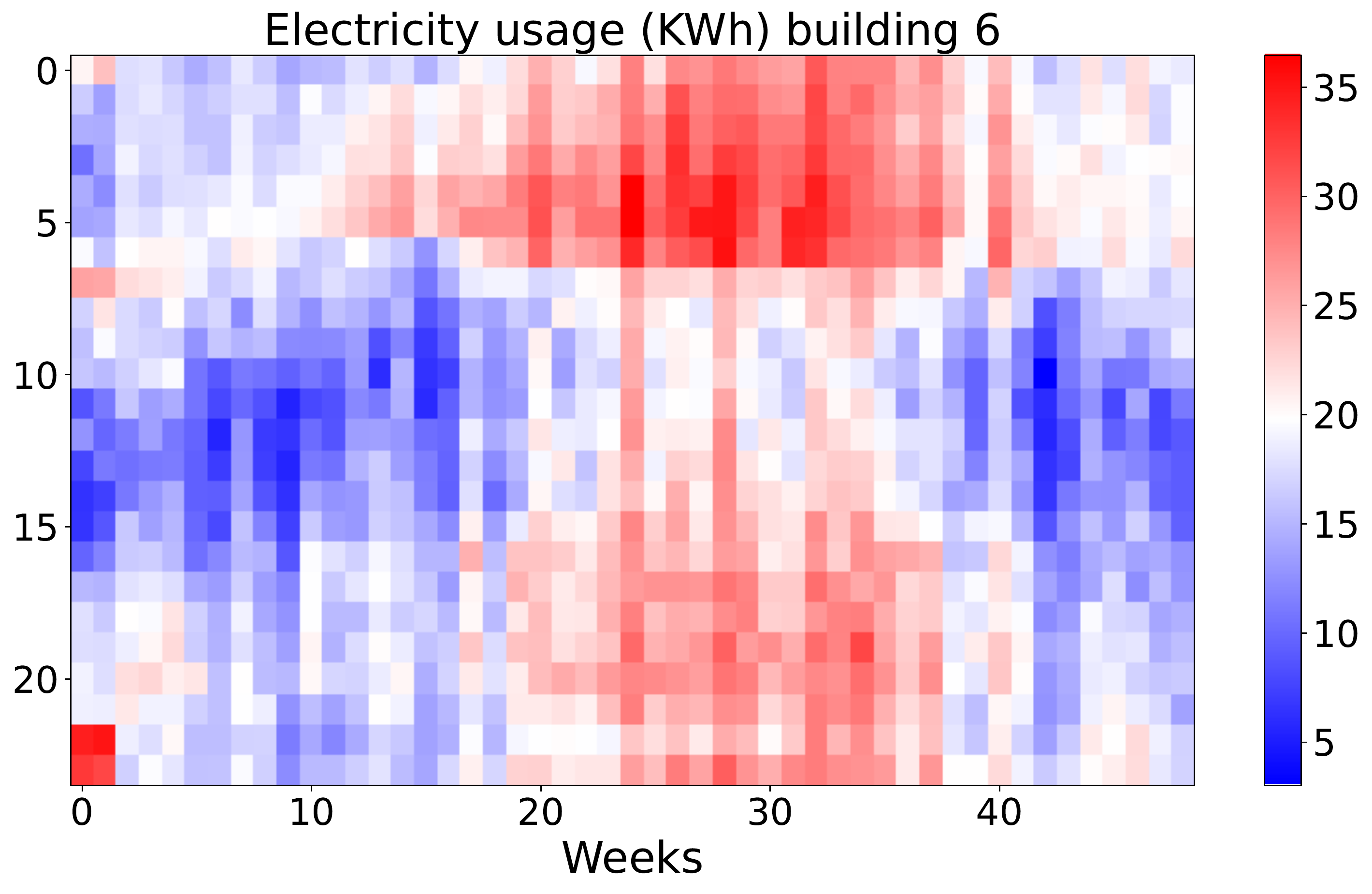}
\caption{net electricity usage}
\end{subfigure}
\begin{subfigure}{.4\textwidth}
  \centering
  \includegraphics[width=\columnwidth]{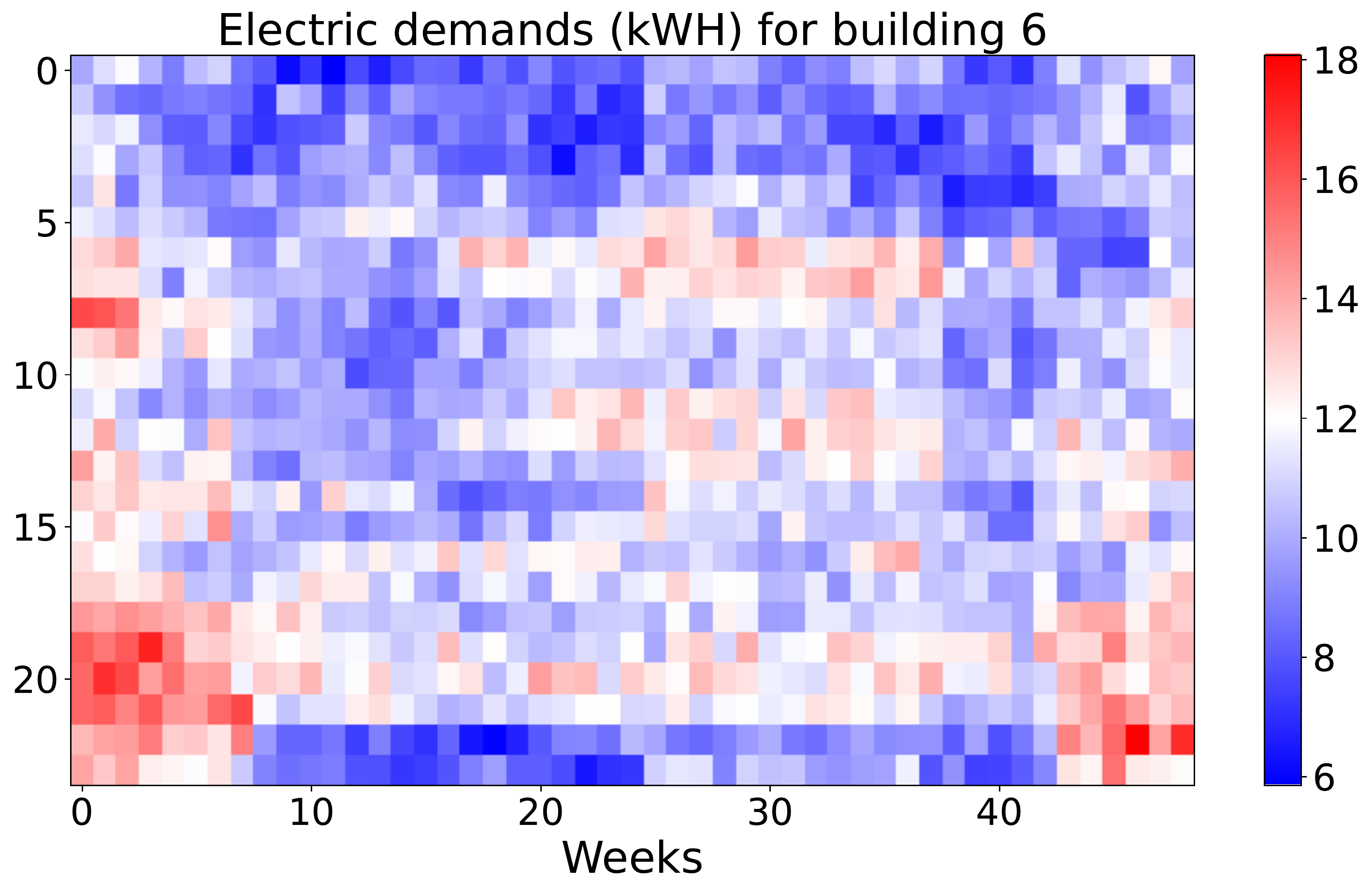}
\caption{electricity demands}
\end{subfigure}
\begin{subfigure}{.4\textwidth}
  \centering
  \includegraphics[width=\columnwidth]{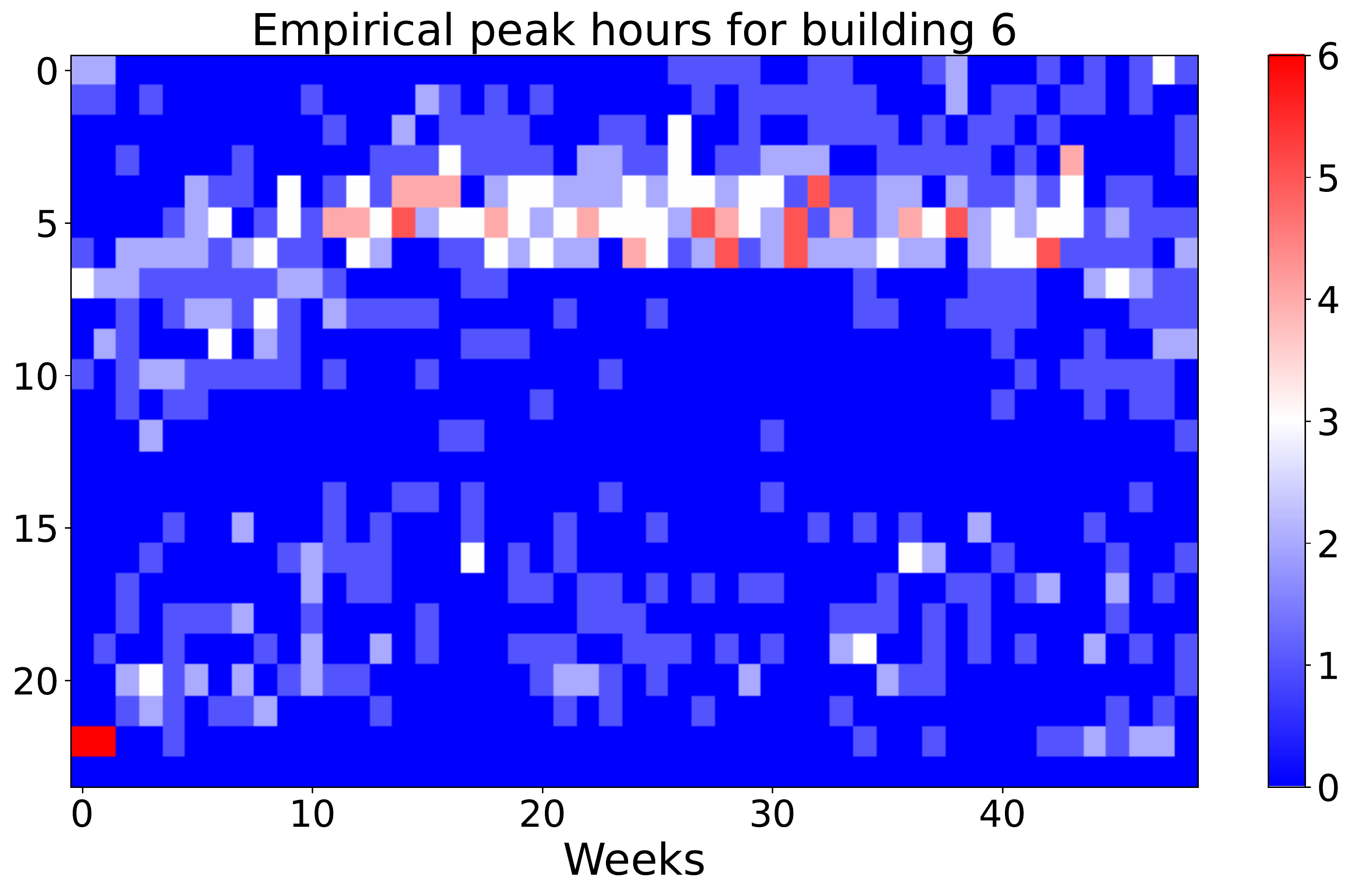}
\caption{empirical counts of peaks}
\end{subfigure}
\begin{subfigure}{.4\textwidth}
  \centering
  \includegraphics[width=\columnwidth]{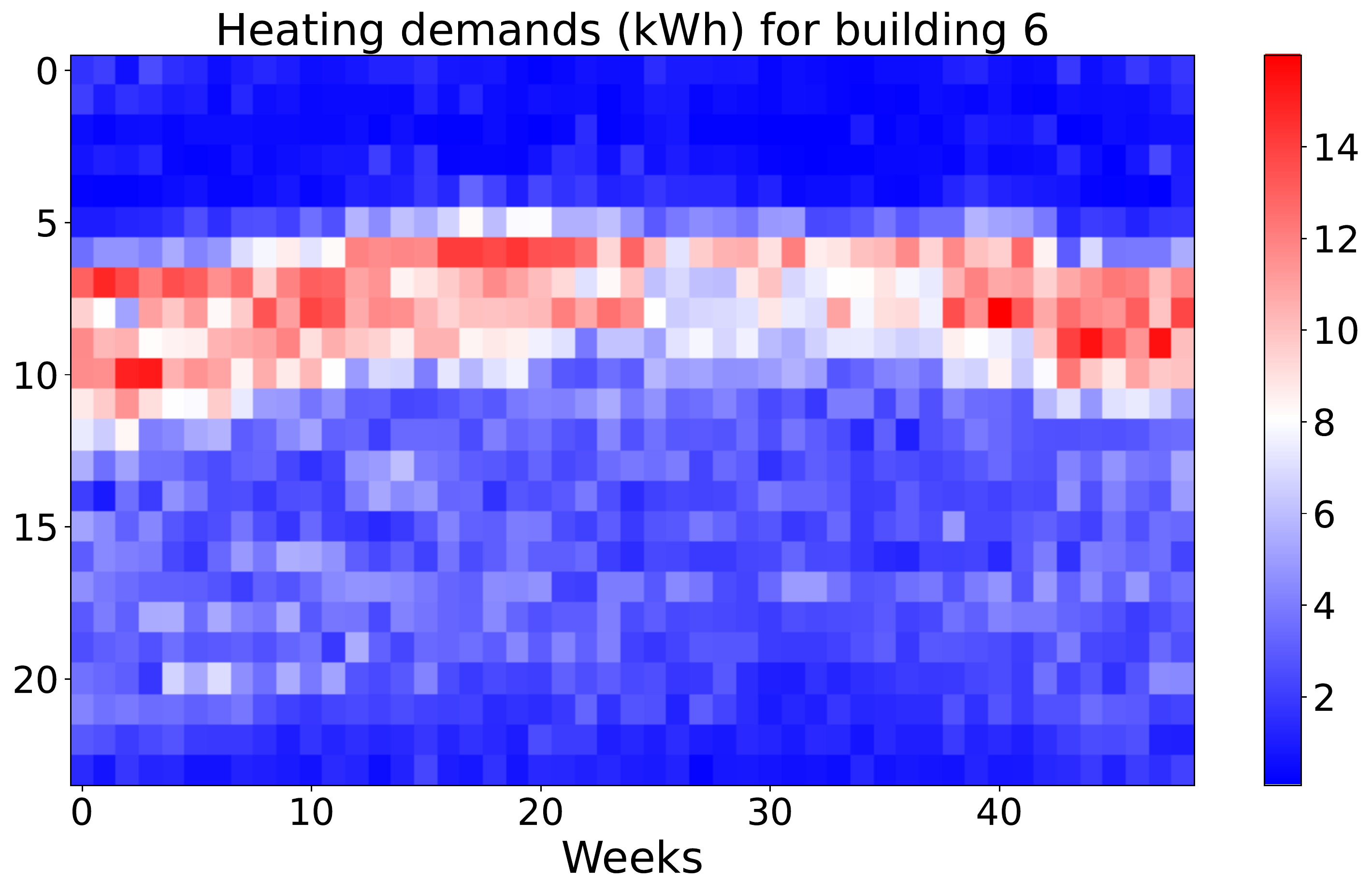}
\caption{heating demands}
\end{subfigure}
\begin{subfigure}{.4\textwidth}
  \centering
  \includegraphics[width=\columnwidth]{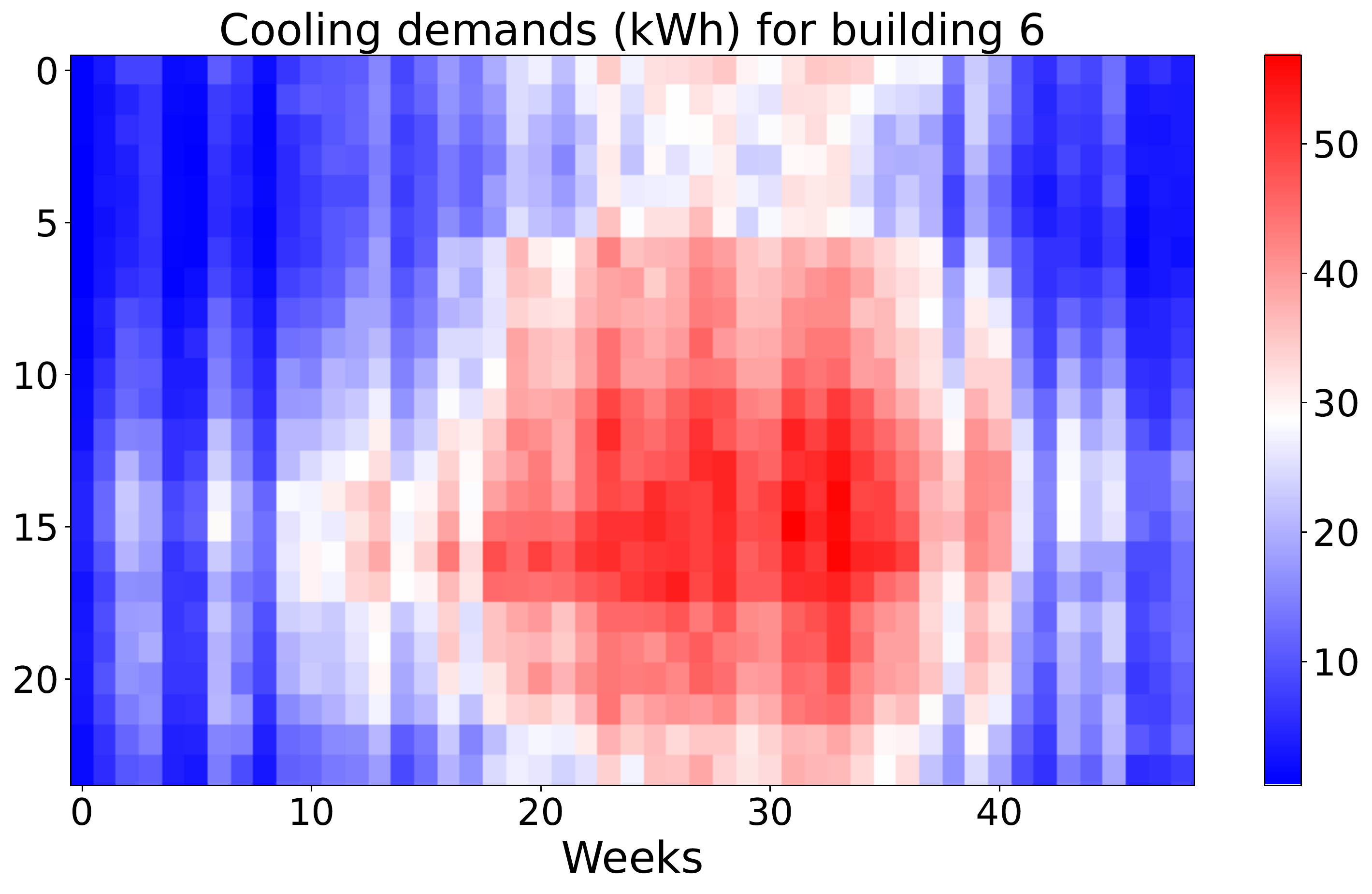}
\caption{cooling demands}
\end{subfigure}
\caption{Similar patterns can be observed for Building 6, where virtual electricity prices rise in response to electricity peaks. For this building, the peaks are more dispersed throughout the day. }
\end{figure}

\begin{figure}[h]
\centering
\begin{subfigure}{.42\textwidth}
  \centering
  \includegraphics[width=\columnwidth]{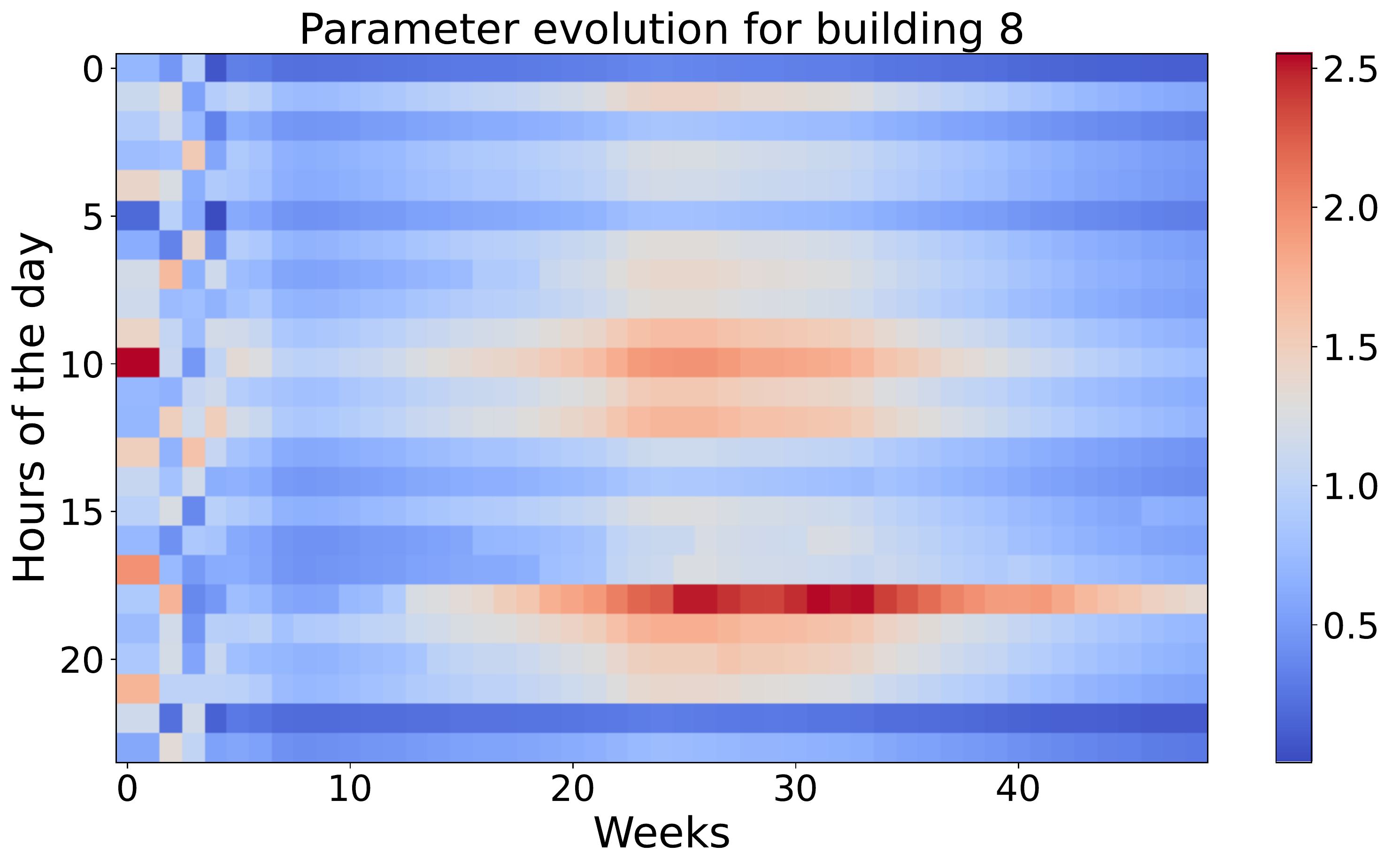}
\caption{virtual electricity prices $\theta$}
\end{subfigure}%
\begin{subfigure}{.4\textwidth}
  \centering
  \includegraphics[width=\columnwidth]{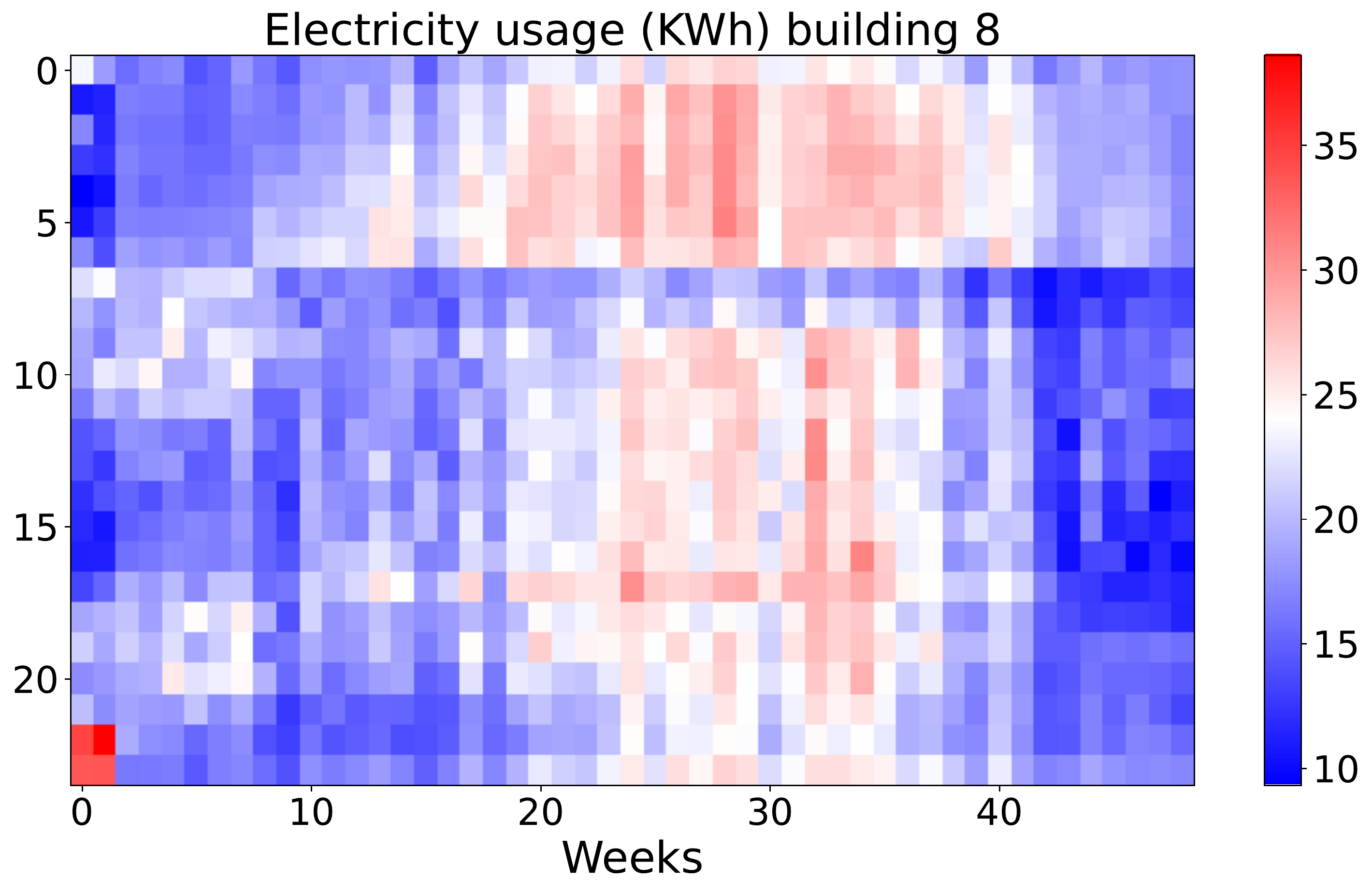}
\caption{net electricity usage}
\end{subfigure}
\begin{subfigure}{.4\textwidth}
  \centering
  \includegraphics[width=\columnwidth]{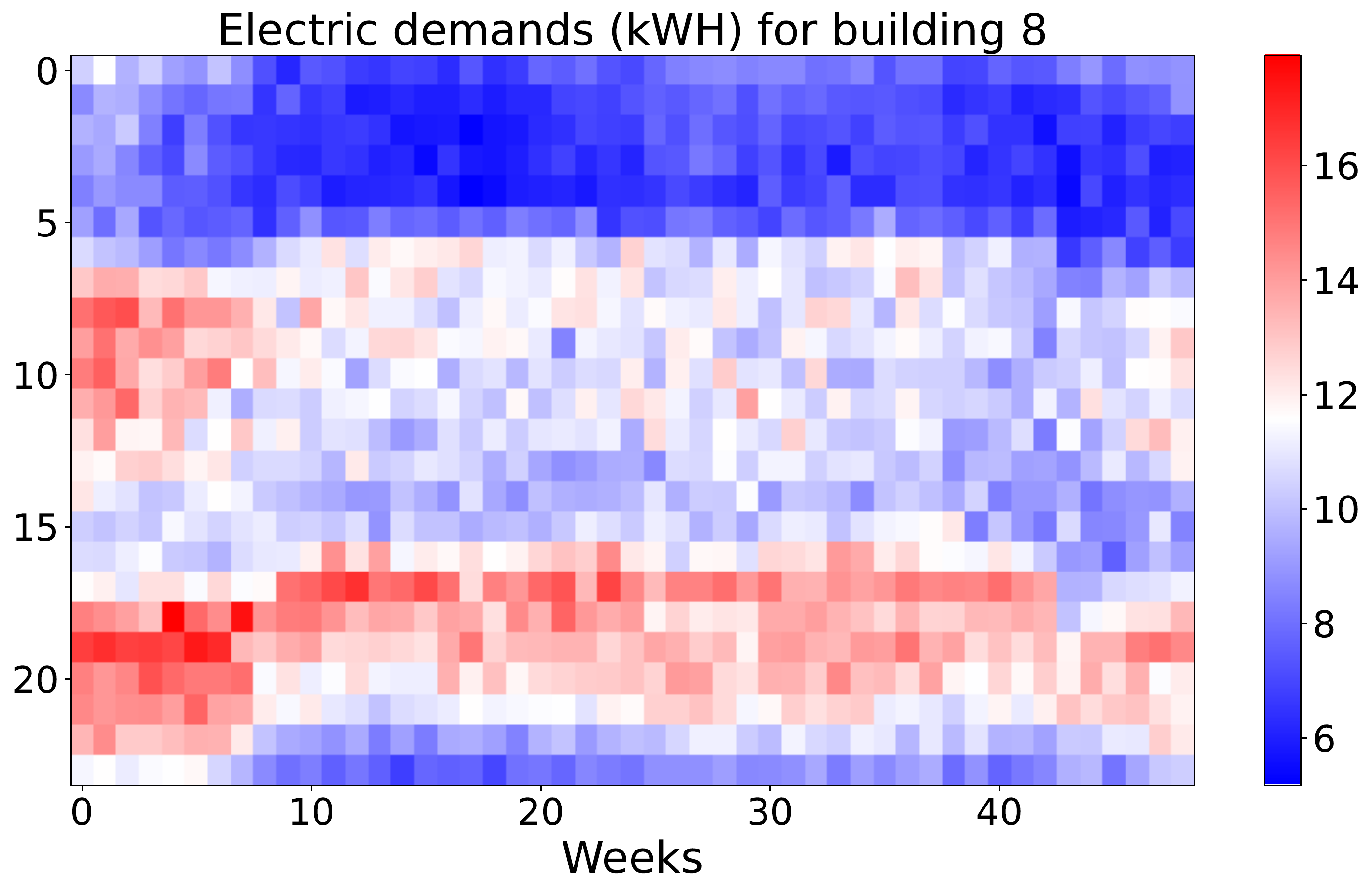}
\caption{electricity demands}
\end{subfigure}
\begin{subfigure}{.4\textwidth}
  \centering
  \includegraphics[width=\columnwidth]{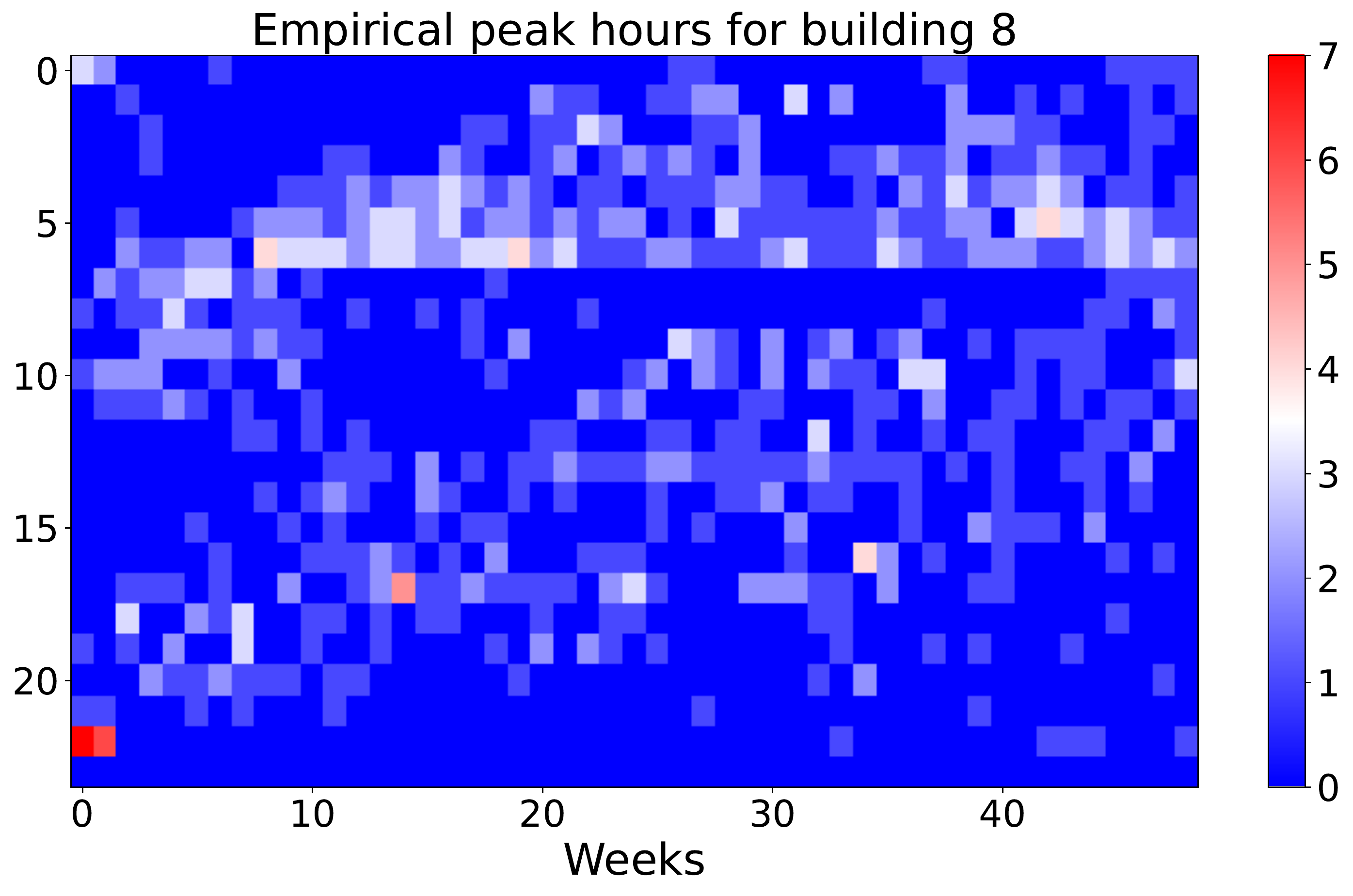}
\caption{empirical counts of peaks}
\end{subfigure}
\begin{subfigure}{.4\textwidth}
  \centering
  \includegraphics[width=\columnwidth]{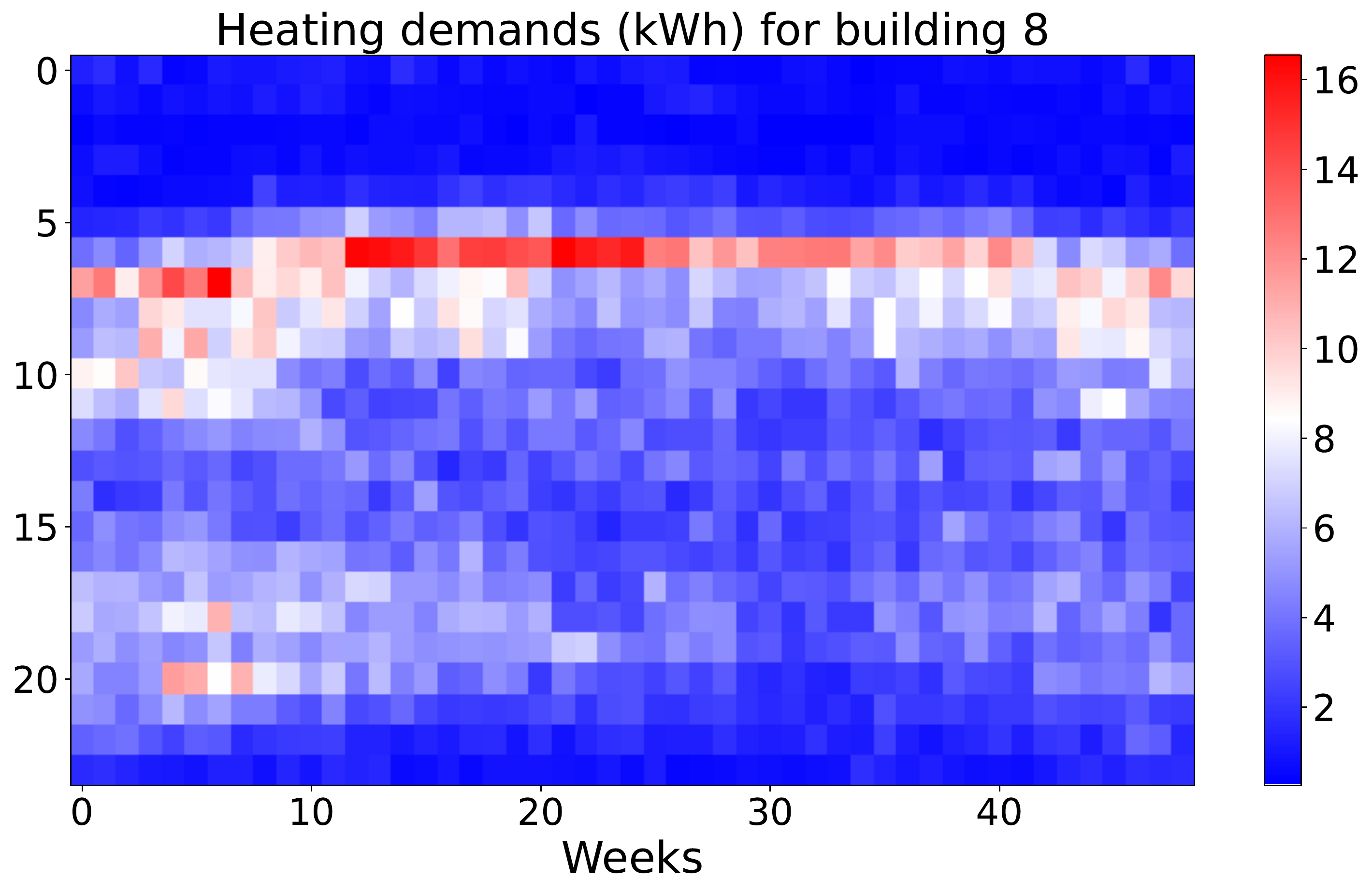}
\caption{heating demands}
\end{subfigure}
\begin{subfigure}{.4\textwidth}
  \centering
  \includegraphics[width=\columnwidth]{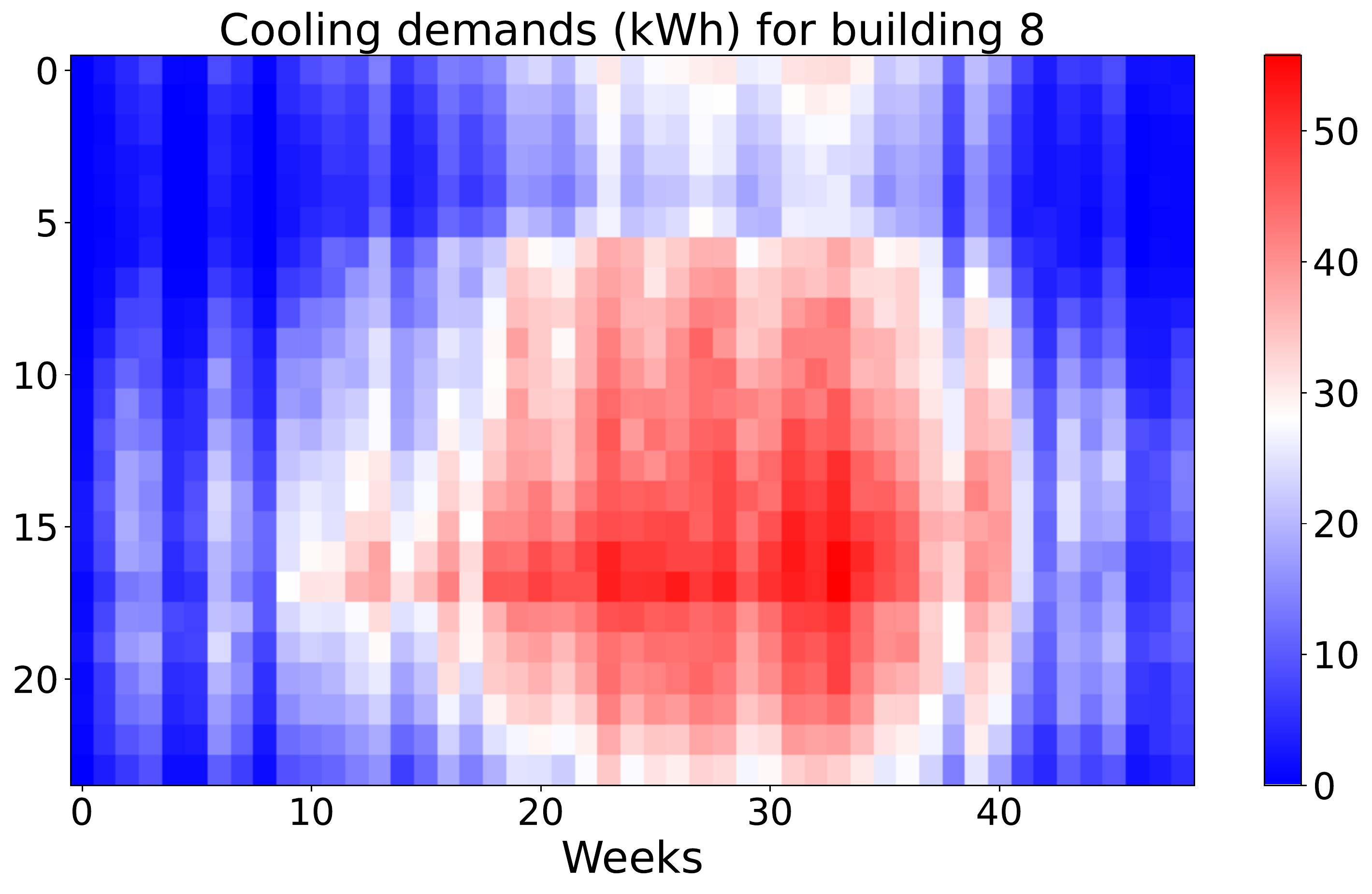}
\caption{cooling demands}
\end{subfigure}
\caption{Similar patterns are observed for Building 8.}
\end{figure}

\clearpage
\newpage

\subsection{Sensitivity analysis}
In this section, we perform sensitivity analysis of ZO-iRL for different parameters/design choices.

\subsubsection{Initial variance of the transition distribution.}
Initial variance $\iota_1$ determines how much randomness we inject during each evolutionary update. Note that in each iteration, we reduce the variance by $1/{k^2}$. We can observe from Table \ref{tab:var-sensitive} that the ZO-iRL algorithm is robust to the initial variance of the transition probability.

\begin{table}[htbp]
\centering
\begin{tabular}{c|ccc}
\textbf{}            & \multicolumn{3}{c}{\textbf{initial variance}} \\ \hline
                     & 0.1           & 0.3           & 0.5           \\
\textbf{total score} & 0.960 (5e-4)  & 0.961 (7e-4)  & 0.962 (0.001)
\end{tabular}\caption{Performance of ZO-iRL for different initial variance values.}\label{tab:var-sensitive}
\end{table}

\subsubsection{Number of parameter candidates.}
Here, we examine the number of candidates sampled for each update, $N_k\in\{3, 5,7\}$. In general, the performance on ZO-iRL is robust to this parameter. There is a trade-off between the number of candidates used in each update and the frequency of updates; as we increase $N_k$, we can expect to find a better candidate in the larger pool; however, we may also decrease the frequency of updates as the evaluation of each candidate takes one week in an online setting. As a result, it appears that increasing the number of candidates sampled does not help improve performance.

\begin{table}[htbp]
\centering
\begin{tabular}{c|ccc}
\textbf{}            & \multicolumn{3}{c}{\textbf{number of candidates}} \\ \hline
                     & 3           & 5          & 7           \\
\textbf{total score} & 0.959 (0.001)  & 0.964 (0.002)  & 0.965 (0.002)
\end{tabular}\caption{Performance of ZO-iRL for different numbers of sampled candidates per update.}\label{tab:number-sensitive}
\end{table}

\subsubsection{Guidance signal.} 
We report the sensitivity of the ZO-iRL algorithm to guidance signal parameters. We consider variants of the guidance signal with respect to 1) number of hours of top electricity use in the past day; 2) incremental value. We keep the guidance learning rate fixed at $\alpha_k = 1$. In the main text, we consider the top-2 hours of electricity usage to be assigned a value of 0.02 (the rest hours are adjusted accordingly so that the sum over all hours of the guidance signal is 0). Here, we consider the following variants: \emph{(a)} top-1 electricity hour to be assigned values of 0.02; \emph{(b)} top-2 electricity hour to be assigned values of 0.04; \emph{(c)} top-3 electricity hour to be assigned values of 0.04; \emph{(d)} top-6 electricity hour to be assigned values of 0.02. We can see from Table \ref{tab:guidance-sensitive} that the proposed algorithm is robust to these variants.

\begin{table}[htbp]
\centering
\begin{tabular}{c|cccl}
\textbf{}            & \multicolumn{4}{c}{\textbf{guidance parameters}}                                                      \\ \hline
                     & top-1& top-2 & top-3 & top-6 \\
\textbf{total score} & 0.962 (2e-4)            & 0.963 (0.003)           & 0.963(0.002)            & 0.975 (0.002)          
\end{tabular}\caption{Performance of ZO-iRL for variants of guidance signals.}\label{tab:guidance-sensitive}
\end{table}

\end{document}